\documentclass[11pt]{article}
\usepackage{makeidx}  % allows for indexgeneration
\usepackage{fullpage}
\usepackage{times}
\usepackage{float}
\usepackage{subfigure}
\usepackage{color}
\usepackage{url}
\usepackage{graphicx}
\usepackage{amsmath}
\usepackage{amssymb}

\usepackage{tikz}
\usepackage{ifthen}
\usepackage{fp}
\usetikzlibrary{arrows,shapes}
\usepackage{algorithmic}
\usepackage[ruled]{algorithm2e}

\usepackage{bbm}

\newcommand{\qed}{\mbox{}\hspace*{\fill}\nolinebreak\mbox{$\rule{0.6em}{0.6em}$}}

\newcommand{\expect}{{\mathbb E}}
\newcommand{\prob}{{\mathbb{P}}}

\definecolor{gray}{rgb}{0.5,0.5,0.5}
\newcommand{\e}{{\epsilon}}

\newcommand{\B}{{\bf B}}

\newtheorem{theorem}{Theorem}
\newtheorem{lemma}[theorem]{Lemma}
\newtheorem{claimp}[theorem]{Claim}
\newtheorem{claim}[theorem]{Claim}

\newtheorem{definition}[theorem]{Definition}
\newtheorem{remark}[theorem]{Remark}

\newenvironment{proof}{{\bf Proof:}}{$\qed$\par}

\newenvironment{proofof}[1]{\noindent{\bf Proof of #1:}}{$\qed$\par}

\usepackage{hyperref}
\hypersetup{
bookmarksnumbered
}

\renewcommand{\b}{{\mbox {\bf b}}}

\newcommand{\I}{\mathbf I}
\renewcommand{\P}{\mathbf P}

\newcommand{\rs}{Ruzsa-Szemer\'edi~}

\newcommand{\bool}{\{0, 1\}}
\newcommand{\F}{{\mathcal F}}
\newcommand{\wt}[1]{\widetilde{#1}}

{\makeatletter
 \gdef\xxxmark{%
   \expandafter\ifx\csname @mpargs\endcsname\relax % in minipage?
     \expandafter\ifx\csname @captype\endcsname\relax % in figure/caption?
       \marginpar{xxx}% not in a caption or minipage, can use marginpar
     \else
       xxx % notice trailing space
     \fi
   \else
     xxx % notice trailing space
   \fi}
 \gdef\xxx{\@ifnextchar[\xxx@lab\xxx@nolab}
 \long\gdef\xxx@lab[#1]#2{{\bf [\xxxmark #2 ---{\sc #1}]}}
 \long\gdef\xxx@nolab#1{{\bf [\xxxmark #1]}}
}

\begin{document}
\setcounter{page}{0}
\title{\Large Better bounds for matchings in the streaming model}
\date{}

\author{Michael Kapralov\\EPFL}

\maketitle              % typeset the title of the contribution

\begin{abstract}
In this paper we present improved bounds for approximating maximum matchings in bipartite graphs in the streaming model. First, we consider the question of how well maximum matching can be approximated in a single pass over the input when $\tilde O(n)$ space is allowed, where $n$ is the number of vertices in the input graph. Two natural variants of this problem have been considered in the literature: (1) the edge arrival setting, where edges arrive in the stream and (2) the vertex arrival setting, where vertices on one side of the graph arrive in the stream together with all their incident edges. The latter setting has also been studied extensively in the context of {\em online algorithms}, where each arriving vertex has to either be matched irrevocably or discarded upon arrival.  In the online setting, the celebrated algorithm of Karp-Vazirani-Vazirani achieves a $1-1/e$ approximation by crucially using randomization (and using $\tilde O(n)$ space). Despite the fact that the streaming model is less restrictive in that  the algorithm is not constrained to match vertices irrevocably upon arrival,  the best known approximation in the streaming model with vertex arrivals and $\tilde O(n)$ space is the same factor of $1-1/e$.

 We show that no (possibly randomized) single pass streaming algorithm  constrained to use $\tilde O(n)$ space can achieve a better than $1-1/e$ approximation to maximum matching, even in the vertex arrival setting.  This leads to the striking conclusion that no single pass streaming algorithm can get any advantage over online algorithms unless it uses significantly more than $\tilde O(n)$ space. Additionally, our bound yields the best known impossibility result for approximating matchings in the {\em edge arrival} model (improving upon the bound of $2/3$ proved by Goel at al[SODA'12]).
 
 Second, we consider the problem of approximating matchings in multiple passes in the vertex arrival setting. We show that a simple fractional load balancing approach achieves approximation ratio $1-e^{-k}k^{k-1}/(k-1)!=1-\frac1{\sqrt{2\pi k}}+o(1/k)$ in $k$ passes using linear space. Thus, our algorithm achieves the best possible $1-1/e$ approximation in a single pass and improves upon the $1-O(\sqrt{\log\log k/k})$ approximation in $k$ passes due to Ahn and Guha[ICALP'11]. Additionally, our approach yields an efficient solution to the Gap-Existence problem considered by Charles et al[EC'10].
\end{abstract}

\newpage
\newcommand{\ifcomplete}[1]{}
\newcommand{\ifshort}[1]{#1}
\newcommand{\iflong}[1]{}
\newcommand{\duallabel}[1]{\label{#1}}
\newcommand{\dualref}[1]{\ref{#1}}
\newcommand{\dualeqref}[1]{\eqref{#1}}

%!TEX root = ./doc-t.tex 
\section{Introduction}

The need to process modern massive data sets necessitates rethinking classical solutions to many combinatorial optimization problems from the point of view of space usage and type of access to the data that algorithms assume. Applications in domains such as processing web-scale graphs, network monitoring or data mining among many others prohibit solutions that load the whole input into memory and assume random access to it. The streaming model of computation has emerged as a more realistic model for processing modern data sets. In this model the input is given to the algorithm as a stream, possibly with multiple passes allowed. The goal is to design algorithms that require small space and ideally one or a small constant number of passes over the data stream to compute a (often approximate) solution. For many problems with applications in network monitoring, it has been shown that space polylogarithmic in the size of the input is often sufficient to compute very good approximate solutions. On the other hand, even basic graph algorithms have been shown to require $\Omega(n)$ space in the streaming model\cite{FKMSZ05}, where $n$ is the number of vertices. A common relaxation is to allow $O(n\cdot \text{polylog}(n))$ space, a setting often referred to as the {\em semi-streaming} model. 

\subsection{Matchings in the streaming model} 
The problem of approximating maximum matchings in bipartite graphs has received significant attention recently, and very efficient small-space solutions are known when multiple passes are allowed\cite{feigenbaum,mcgregor,eggert,guha-ahn-1,kmm11}. The best known algorithm due to Ahn and Guha \cite{guha-ahn-1} achieves a $1-O(\sqrt{\log\log k/k})$ in $k$ passes for the weighted as well as the unweighted version of the problem using $\tilde O(kn)$ space.

All algorithms mentioned above require at least two passes to achieve a nontrivial approximation. The problem of approximating  matchings in  a single pass has recently received significant attention\cite{gkk:streaming-soda12, kmm11}. Two natural variants of this problem have been considered in the literature: (1) the edge arrival setting, where edges arrive in the stream and (2) the vertex arrival setting, when vertices on one side of the graph arrive in the stream together with all their incident edges. The latter setting has also been studied extensively in the context of {\em online algorithms}, where each arriving vertex has to either be matched irrevocably or discarded upon arrival.  

In a single pass, the best known approximation in the edge arrival setting is still $1/2$, achieved by simply keeping a maximal matching (this was improved to $1/2+\e$ for a constant $\e>0$ under the additional assumption of random edge arrivals~\cite{kmm11}).  It was shown in \cite{gkk:streaming-soda12} that no $\tilde O(n)$ space algorithm can achieve a better than $2/3$ approximation in this setting.

In the vertex arrival setting, the best known algorithms achieve an approximation of $1-1/e$. The assumption of vertex arrivals allows one to leverage results from online algorithms \cite{kvv,my11,kmt11}. In the online model vertices on one side of the graph are known, and vertices on the other side arrive in an adversarial order. The algorithm has to either match a vertex irrevocably or discard upon arrival. The celebrated algorithm of Karp-Vazirani-Vazirani achieves a $1-1/e$ approximation for the online problem by crucially using randomization (additionally, this algorithm only uses $\tilde O(n)$ space). A {\em deterministic} single pass $\tilde O(n)$ space $1-1/e$ approximation in the vertex arrival setting was given in \cite{gkk:streaming-soda12} (such a deterministic solution is provably impossible in the online setting). In \cite{gkk:streaming-soda12}, the authors also showed by analyzing a natural one-round communication problem that no  single-pass streaming algorithm that uses $\tilde O(n)$ space can obtain a better than $3/4$ approximation in the vertex arrival setting. They also provided a protocol for this communication problem that matches the $3/4$ approximation ratio, suggesting that new techniques would be needed to prove a stronger impossibility result.

\paragraph{Recent work.} The lower bound presented in this paper has recently been improved to $\frac1{1+\ln 2}\approx 0.591$ by~\cite{Kapralov21} for the more general edge arrival model, following exciting developments in online matching~\cite{WangW15,EpsteinLSW13,GamlathKMSW19}. A $2/3$-approximation to maximum matching in a single pass over a randomly ordered stream of edges  in $n\log^{O(1)} n$ space has recently been given by~\cite{Bernstein20} (reducing the space complexity of the approach of~\cite{AssadiBBMS19} from $\widetilde O(n^{3/2})$ to $n \log^{O(1)} n$), and very recently improved to $2/3+\Omega(1)$ by ~\cite{AssadiBehnezhad21}. A $1-O(1/\sqrt{k})$ approximation in $k$ passes using $n\log^{O(1)} n$ space was given by~\cite{AssadiLT21}.

\subsection{Our results}
In this paper, we improve upon the best known bounds for both the single pass and multi-pass settings. In the single pass setting, we prove an optimal impossibility result for vertex arrivals, which also yields the best known impossibility result in the edge arrival model. For the multipass setting, we give a simple algorithm that improves upon the approximation obtained by Ahn and Guha in the vertex arrival setting, as well as yields an efficient solution to the Gap-Existence problem considered by Charles et al\cite{dev10}.

\medskip
\noindent
{\bf Lower bounds.} Our main result is an optimal bound on the best approximation ratio that a single-pass $\tilde O(n)$ space streaming algorithm can achieve in the vertex arrival setting:
\begin{theorem}\duallabel{thm:main}
No (possibly randomized) one-pass streaming algorithm can obtain a $(1-1/e+c)$-approximation to the maximum matching with probability at least $3/4$ for any constant $c>0$, unless it uses at least $n^{1+\Omega_c(1/\log\log n)}$ space, even in the vertex arrival model.
\end{theorem}

\begin{remark}
In fact, we prove a more refined statement: for every integer $k\geq 2$ if the edge set is partitioned among $k$ players communicating in the number-in-hand model (with the $i$-th player sending a single message to the $(i+1)$-th after receiving a message from the $(i-1)$-th player) no algorithm can achieve a $1-(1-1/k)^k+\Omega(1)$ approximation to maximum matching unless it uses $n^{1+\Omega(1/\log\log n)}$ communication. 
\end{remark}
We note that this bound is matched by the randomized KVV algorithm\cite{kvv} for the online problem and the deterministic $\tilde O(n)$ space algorithm of \cite{gkk:streaming-soda12}. One striking consequence of our bound is that no single-pass streaming algorithm  can improve upon the more constrained {\em online} algorithm of KVV, which has to make irrevocable decisions, unless is uses significantly more than $\tilde O(n)$ space. Our bound also improves upon the best known bound of $2/3$ for small space one-pass streaming algorithms in the {\em edge arrival model}.

\medskip
\noindent
It was shown in \cite{gkk:streaming-soda12} via an analysis of the natural two-party communication problem that no one-pass streaming algorithm that uses $\tilde O(n)$ space can achieve approximation better than $2/3$ in the edge arrival setting and $3/4$ in the vertex arrival setting. Furthermore, the authors also gave a communication protocol that proves the optimality of both bounds for the communication problem, thus suggesting that a more intricate approach would be needed to prove better impossibility results.
While the lower bounds from \cite{gkk:streaming-soda12} follow from a construction of a distribution on inputs that consists of two parts and hence yields a two-party communication problem, here we obtain an improvement by constructing hard input sequences that consist of  $k$ parts instead of two, getting a lower bound that approaches $1-1/e$ for large $k$. 

\noindent{\bf Upper bounds.} We show that a simple algorithm based on fractional load balancing achieves the optimal $1-1/e$ approximation in a single pass and $1-\frac1{\sqrt{2\pi k}}+o(k^{-1/2})$ approximation in $k$ passes,  improving upon the best known algorithms for this setting:
\begin{theorem}\duallabel{thm:main-ubound}
There exists an algorithm for approximating the maximum matching $M$ in a bipartite graph $G=(P, Q, E)$ with the $P$ side arriving in the stream to factor $1-e^{-k}k^{k-1}/(k-1)!=1-\frac1{\sqrt{2\pi k}}+O(k^{-3/2})$ in $k$ passes using $O(|P|+|Q|)$ space.   The algorithm can be implemented to run in nearly linear time in the number of edges in the graph per pass, with space complexity $\widetilde{O}(|P|+|Q|)$.
\end{theorem}

\noindent{\bf The gap-existence problem.} In \cite{dev10} the authors give an algorithm for the closely related {\em gap-existence} problem. In this problem the algorithm is given a bipartite graph $G=(A, I, E)$, where $A$ is the set of advertisers with budgets $B_a, a\in A$ and $I$ is the set of impressions. The graph is lopsided in the sense that $|I|\gg |A|$. A matching $M$ is {\em complete} if $|M\cap \delta(i)|=1$ for all $i\in I$ and $|M\cap \delta(a)|=B_a$ for all $a\in A$.
The gap-existence problem consists of distinguishing between two cases:
\begin{description}
\item[(YES)] there exists a complete matching with budgets $B_a$;
\item[(NO)]  there does not exist a complete matching with budgets $\lfloor (1-\e)B_a\rfloor$.
\end{description}  

 The approach of \cite{dev10} is via sampling the $I$ side of the graph, and yields a solution that allows for non-trivial subsampling when the budgets are large. In particular, they obtain an algorithm with runtime $O\left(\frac{|A|\log |A|}{\e^2}\cdot \frac{|I|}{\min_a |B_a|}\right)$, which is sublinear in the size of the graph when all budgets are large. 
In Section~\dualref{sec:gap} we improve significantly upon their result, showing 
\begin{theorem}\duallabel{thm:gap}
Gap-Existence can be solved in $O(\log (\frac1{\e}\sum_{a\in B_a}B_a)/\e^2)$ passes using space $O(\sum_{a\in A}B_a/\e)$. The time taken for each pass is nearly linear in the representation of the graph.
\end{theorem}
It should also be noted that the result of \cite{dev10} could be viewed as a single pass algorithm, albeit with the stronger assumption that the arrival order in the stream is random.

\medskip
\noindent
{\bf Organization:} We start by presenting a toy version of our lower bound construction in Section~\ref{sec:toy-construction}. The construction in Section~\ref{sec:toy-construction} does not give a strong streaming lower bound, but captures most of the properties of our hard input distribution, while at the same time being quite simple to describe. In Section~\dualref{sec:main} we give the actual lower bound construction and prove Theorem~\ref{thm:main}. Our basic multipass algorithm for approximating matchings is presented in Section~\dualref{sec:multipass}, and the algorithm for Gap-existence is given in Section~\dualref{sec:gap}. 

%!TEX root = ./doc-t.tex 
\section{A toy construction}\label{sec:toy-construction}
 In this section we show that for every integer $k\geq 2$ there exists a distribution ${\mathcal D}$ on input instances to the bipartite matching problem such that a graph $G$ with $N$ vertices sampled from distribution ${\mathcal D}$ has a nearly perfect matching with high probability, but any single-pass streaming algorithm that maintains a subset of edges of $G$ in memory and outputs a matching in the subset of edges retained cannot achieve a better than $1-(1-1/k)^k+\delta$ approximation for a constant $\delta>0$ unless it maintains $\Omega(N\log N)$ edges. 
 
 We define a family of graphs that forms the basis of our hard input instances in Section~\ref{sec:1}. In Section~\ref{sec:2} we define a hard input distribution based on these graphs, prove Theorem~\ref{thm:main-toy} (our main result in this section), which provides the $1-(1-1/k)^k+\delta$ upper bound on the approximation ratio that an algorithm that stores $o(n\log n)$ edges.

\newcommand{\E}{{\mathcal E}}

\subsection{Construction of the input family of graphs}\label{sec:1}

We construct bipartite graphs $G=(S, T, E)$, with $S$ and $T$ the two sides of the bipartition. 

\paragraph{Vertices of $G$: the $T$ side of the bipartition} Let $k\geq 2$ be a large constant integer. Let $m\geq 1$ a multiple of $k$ be a sufficiently large integer.  Let $T=[m]^n$, i.e. vertices in $T$ are vectors of dimension $n$, with each co-ordinate taking values in $[m]=\{1, 2,\ldots, m\}.$ This way we have $N:=|T|=m^n$, so $n=\Omega(\log N)$ for every constant $m$. The vertices on the $S$ side of the bipartition will also be associated with points on the hypercube $[m]^n$, as defined below.

\renewcommand{\I}{{\mathcal I}}

\paragraph{Vertices of $G$: the $S$ side of the bipartition} To define the vertices in the partition $S=S_0\cup S_1\cup \ldots \cup S_k$, we first partition the set of coordinates $[n]$ into $k$ equal size blocks $[n]=B_1 \cup \ldots \cup B_k$. 
Graphs $G=G(j_1,\ldots, j_k)$ will be parameterized by  a sequence $(j_1,\ldots, j_k)\in B_1\times\ldots \times B_k$ of coordinates. Also for each point $x\in [m]^n$ let $Z_x$ be an independent Bernoulli 0/1 random variable with expectation $1/k$ -- we will later choose some fixing of these random variables for the final construction. Then for every $i=0,\ldots, k$ we let 
\begin{equation}\label{eq:def-ti-si}
\begin{split}
T_i&=\left\{y\in [m]^n: y_{j_r}\in (m/k, m]\text{~for all~}r=1,\ldots, i\right\}\\
S_i&=\left\{x\in T_i: Z_x=1\right\}.\\
\end{split}
\end{equation}
Note that $T_0=T$, and for every $i=0,\ldots, k-1$ the set $S_i$ is a subsampling of $T_i$ at rate $1/k$. We also let, for every $i=0,\ldots, k-1$ and $j\in B_{i+1}$ 
\begin{equation*}
\begin{split}
T_i^j&=\left\{y\in T_i: y_{j}\in (m/k, m]\right\}\\
S_i^j&=\left\{x\in S_i: Z_x=1\text{~and}~x_j\in (m/k, m]\right\}.\\
\end{split}
\end{equation*}

We also define for each $i=0,\ldots, k-1$
\begin{equation}\label{eq:s-star-i-def}
\begin{split}
S^*_i&=\left\{x\in S_i: x_{j_r}\in (m/k, m]~\text{~for all~}r=i+1,\ldots, k\right\}.\\
\end{split}
\end{equation}

We will use
\begin{theorem}[Chernoff bound]\label{thm:chernoff-bounds}
Let $X_1,\ldots, X_n$ be independent Bernoulli random variables, let $\mu:=\expect[\sum_{i=1}^n X_i]$. Then 
for every $\delta\in (0, 1)$ one has 
$\prob[|\sum_{i=1}^n X_i-\mu|>\delta\mu]\leq 2e^{-\delta^2\mu/3}$.
\end{theorem}

We first note that 

\begin{lemma}
\label{lm:set-sizes}
For any $k\geq 2$ the following conditions hold. 
{\bf (1)} For every choice of $j_1,\ldots, j_k$ and every $i=0,\ldots, k$ one has $|T_i| = (1 - \frac{1}{k})^i|T|$.
For every $\eta\in (0, 1/2)$ there exists an event $\E_{set-sizes}$ that occurs with probability at least $1-k(\log N)^k e^{-\Omega(\eta^2 N/k)}$ over the random variables $Z_x, x\in [m]^n$ such that conditioned on $\E_{set-sizes}$ one has for every choice of $(j_1,\ldots, j_k)\in B_1\times \ldots \times B_k$ simultaneously for every $i=0,\ldots, k-1$ {\bf (2)}  $|S_i|=(1\pm \eta)|T_i|/k$, {\bf (3)}  $|S_i^j|=(1\pm O(\eta))(1-1/k)|S_i|$,  and {\bf (4)}  $|S^*_i|=(1\pm \eta) |T_k|/k$ (note that this quantity does not depend on $i$). 
\end{lemma}
\begin{proof}
{\bf (1)} follows directly by definition of $T_i$. For {\bf (2)} we first note that by an application of Chernoff bounds for a fixed collection $j_1,\ldots, j_k$ one has $|S_i|=(1\pm \eta)|T_i|/k$ with probability at least $1-e^{-\Omega(\eta^2 N/k)}$, where we used the fact that $(1-1/k)^i\geq (1-1/k)^k\geq (1-1/2)^2$ for every $i=0,\ldots, k$, since $k\geq 2$ by assumption of the lemma. A union bound over at most $(\log N)^k$ choices for $j_1,\ldots, j_k$ and $k$ choices for $i$ gives the result of the lemma. The third and fourth bound follow analogously.
\end{proof}

We need the following simple lemma:
\begin{lemma}\label{lm:typical-degrees}
For every $i=0,\ldots, k-1$, every $(j_1,\ldots, j_i)\in B_1\times \ldots\times B_i$ the following conditions hold. For every $j\in B_{i+1}$, every $z\in T_i^j$ let $\deg_j(z)$ denote the number of $j'\in B_{i+1}\setminus \{j\}$ such that $z\in T_i^{j'}$.  Let $\overline{\deg}_j(z)$ denote the number of $j'\in B_{i+1}\setminus \{j\}$ such that $z\in T_i\setminus T_i^j$. Then for every $\eta\in (0, 1/2)$ one has $\deg_j(z)\in (1\pm \eta) (1-1/k) (|B_{i+1}|-1)$ and $\overline{\deg}_j(z)\in (1\pm \eta) (|B_{i+1}|-1)/k$ for all but a $N^{-\Omega(\eta^2/m^2)}$ fraction of $z\in T$. The same bounds hold for $z\in S_i^j$.
\end{lemma}
\begin{proof}
Recall that $T_i=\left\{y\in [m]^n: y_{j_r}\in (m/k, m]\text{~for all~}r=1,\ldots, i\right\}$. We thus have
$$
T_i^j=\left\{y\in [m]^n: y_{j_r}\in (m/k, m]\text{~for all~}r=1,\ldots, i\text{~and~}y_{j}\in (m/k, m]\right\}
$$
and 
$$
T_i^{j'}=\left\{y\in [m]^n: y_{j_r}\in (m/k, m]\text{~for all~}r=1,\ldots, i\text{~and~}y_{j'}\in (m/k, m]\right\}.
$$
Since $j_r\in B_r$ for every $r=1,\ldots, i$, and $j, j'\in B_{i+1}$, and $B_1, \ldots, B_k$ are disjoint, we have that coordinate $y_{j'}$ is unconstrained in $T_i^j$, a uniformly random $z\in T_i^j$ satisfies $z_{j'}\in (m/k, m]$ with probability exactly $1-1/k$. Furthermore, these events are independent for different collections of coordinates in $B_{i+1}\setminus \{j\}$. Select $z\in T_i^j$ uniformly at random. For $j'\in B_{i+1}\setminus \{j\}$ let $F_{j'}=1$ if $z\in T_i^{j'}$ and $F_{j'}=0$ otherwise (note that $\expect[F_{j'}]=1-1/k$ for every $j'\in B_{i+1}\setminus \{j\}$). We now have by the Chernoff bound (Theorem~\ref{thm:chernoff-bounds}) that for every $\eta\in (0, 1/2)$ 
$$
\prob_{z\sim UNIF(T_i^j)}\left[\sum_{j'\in B_{i+1}\setminus \{j\}} F_{j'}\not \in (1\pm \eta) (1-1/k) (|B_{i+1}|-1)\right]\leq  2e^{-\Omega(\eta^2 |B_{i+1}|})=N^{-\Omega(\eta^2/m^2)},
$$
where we used the fact that $|B_{i+1}|=n/m=(\log_m N)/m$ in the last transition. This proves the first claim. The proof of the second and third claim is analogous.
\end{proof}

\paragraph{Edges of $G$.}

For each $i=0,\ldots, k-1$ edges of the subgraph $G=(P_i, Q, E_i)$ will be associated with coordinates in $B_{i+1}$, as we now describe. Specifically, each coordinate $j\in B_{i+1}$ will correspond to a set of edges in $G$ that form a rather large near-matching (of size $\Omega(N/k)$, as described below).

For each $i=0,\ldots, k-1$ the edge set $E_i\subseteq S_i\times T_i$ are defined as follows.  For each coordinate $j\in B_i$ for each $x\in [m]^n$ we let 
$$
\text{line}_j(x)=\{x'\in [m]^n: (x'-x)_s=0\text{~for all~}s\neq j\}
$$
denote the line through $x$ in coordinate direction $j$. Note that $|\text{line}_j(x)|=m$ for all $x$. Furthermore, we have 
\begin{lemma}\label{lm:line-properties}
For every $\eta\in (0, 1/2)$, if $C>0$ is a sufficiently large constant, then for $m\geq C\eta^{-2} k \log \eta^{-1}$ a multiple of $k$,  for every $i=0,\ldots, k-1$, every $(j_1,\ldots, j_i)\in B_1\times \ldots \times B_i$ for each $y\in T_i$ one has for each $j\in B_{i+1}$
\begin{itemize}
\item[(1)] $|\text{line}_j(y)|=m$ and $\text{line}_j(y)\subseteq T_i$;
\item[(2)] $|\text{line}_j(y)\setminus T_i^j|=m/k$;
\item[(3)] there exists an event $\E_{large-lines}(j_1,\ldots, j_i, j)$ that occurs with probability at least $1-e^{-\Omega(\eta^2 N/k)}$ such that conditioned on $\E_{large-lines}(j_1,\ldots, j_i, j)$ the number of $y\in T_i$ such that $|\text{line}_j(y)\cap S_i^j|\not \in (1\pm \eta)|\text{line}_j(y)|(1-1/k)/k$ is upper bounded by $\eta^2 |T_i|$.
%$\prob_Z\left[|\text{line}_j(y)\cap S_i|\not \in (1\pm \eta)|\text{line}_j(x)|/k\right]\leq \eta$, where the probability is over $Z_x$ for $x\in T_i\cap \text{line}_j(y)$.
\end{itemize}
In particular, there exists an event $\E_{large-lines}$ that occurs with probability at least $1-k(\log N)^{k} e^{-\Omega(\eta N/k)}$ such that for every $i=0,\ldots, k-1$, every collection $j_1,\ldots, j_i$, every $j\in B_{i+1}$ one has that 
the number of $y\in T_i$ such that $|\text{line}_j(y)\cap S_i|\not \in (1\pm \eta)|\text{line}_j(x)|/k$ is upper bounded by $2\eta^2 |T_i|$.
\end{lemma}
\begin{proof}
The first claim follows since, due to the assumption that $y\in T_i$ we have
\begin{equation*}
\begin{split}
\text{line}_j(y)&=\left\lbrace y'\in [m]^n: (y'-y)_s=0\text{~for all~}s\neq j \right\rbrace\\
&=\left\lbrace y'\in [m]^n: (y'-y)_s=0\text{~for all~}s\neq j, y'_{j_r}\in (m/k, m]\text{~for all~}r=1,\ldots, i\right\rbrace \\
&\subseteq T_i\\
\end{split}
\end{equation*}
since $j\neq j_1,\ldots, j_i$ due to the assumption that $j\in B_{i+1}$.

The second claim follows similarly. For the third claim note that 
\begin{equation*}
\begin{split}
\expect_Z\left[|\text{line}_j(y)\cap S_i^j|\right]&=\sum_{y'\in \text{line}_j(y) \cap T_i^j} \prob_Z[y\in S_i^j]\\
&=|\text{line}_j(x)|(1-1/k)/k,
\end{split}
\end{equation*}
where we used the fact that $|\text{line}_j(y)\cap T_i^j|=m/k$ for every $y\in T_i$ by {\bf (2)} and $|\text{line}_j(y)|=m$ by {\bf (1)}.
Since $m/k\geq C\eta^{-2}\log \eta^{-1}$ for a constant $C>0$ by assumption of the lemma, the claim follows by the Chernoff bound (Theorem~\ref{thm:chernoff-bounds}). The final claim follows by a union bound over all choices of $i, j_1,\ldots, j_i, j$.
\end{proof}

We now condition on the event $\E_{large-lines}$ from Lemma~\ref{lm:line-properties}, so that that $|\text{line}_j(x)\cap S_i^j|\not \in (1\pm \eta)|\text{line}_j(x)|/k$ for all $i=0,\ldots, k-1$, $j\in B_{i-1}$ and all but $2\eta^2  |T_i|$ choices of $x\in T_i$. 

\paragraph{Defining the edges induced by $T_i\cup S_i$.} We now define the edges of $G=G(j_1,\ldots, j_k)$ induced by $T_i\cup S_i$ (note that these edges are a function of the prefix $(j_1,\ldots, j_i)$ only). The edge set is a union of a large number of induced subgraphs of constant size. We will need 
\begin{definition}[Typical line]\label{def:typical-line-toy} For every $i=0,\ldots, k-1$, every $(j_1,\ldots, j_i)\in B_1\times \ldots \times B_i$, $j\in B_{i+1}$, for $z\in T_i$ we say that $\text{line}_j(z)$ is typical if $|\text{line}_j(z)\cap S_i^j|\in (1\pm \eta)|\text{line}_j(z)|(1-1/k)/k=(1\pm \eta)(1-1/k)m/k$ and atypical otherwise.
\end{definition} 

For every $y\in T_i$, if $\text{line}_j(y)$ is typical, let $\widetilde{\text{line}}_j(y)$ be an arbitrary subset of $\text{line}_j(y) \cap S_i^j$ of size $(1-\eta)|\text{line}_j(y)|(1-1/k)/k=(1-\eta)(1-1/k)\cdot m/k$, and let $\widetilde{\text{line}}_j(y):=\emptyset$ otherwise. \if 0 Now for every $i=0,\ldots, k-1$, every collection $(j_1,\ldots, j_i)\in B_1\times \ldots \times B_i$ and $j\in B_{i+1}$ we define matching $M_{i, j}\subseteq E_i$ as follows.  For each $y\in T_i$ we construct a matching of (most of) the set $\text{line}_j(y)\cap (T_i\setminus T_i^j)$ to (most of) the set $\text{line}_j(y)\cap S_i^j$.
 Note that, conditioned on $\E_{large-lines}$ (defined in Lemma~\ref{lm:line-properties}), one has 
$|\text{line}_j(y)\cap (T_i\setminus T_i^j)|=m/k$ and $|\text{line}_j(y)\cap S_i^j|\in (1\pm \eta)|\text{line}_j(y)|(1-1/k)/k=(1\pm \eta) (1-1/k) m/k$ for most $y\in T_i$. \fi
We now define the edge set of $E_i$. For every $j\in B_{i+1}$, every $y\in T_i$ include a complete bipartite graph between $\widetilde{\text{line}_j(y)}$ and $\text{line}_j(y)\cap (T_i\setminus T_i^j)$, i.e. 
\begin{equation}\label{eq:edges-ei-def}
E_i=\bigcup_{j\in B_{i+1}} E_i^j, \text{~where~}E_i^j=\bigcup_{y\in T_i} \widetilde{\text{line}_j(y)} \times (\text{line}_j(y) \cap (T_i\setminus T_i^j)).
\end{equation}

Note that for every $a\in \text{line}_j(y)\cap (T_i\setminus T_i^j)$ and $b\in \widetilde{\text{line}_j(y)}$ we have $(a-b)_{q}=0$ for all $q\neq j$, $a_j\in [1, m/k]$ and $b_j\in (m/k, m]$. 
We now prove that for every $j$ there exists a matching of (most of) $S_i$ to $T_i\setminus T_i^j$. 

First note that it follows immediately that there exists a matching of at least a $(1-1/k-O(\eta+\eta^2 k))$ fraction of $S_i$ to $T_i\setminus T_i^j$. Indeed, for every $y\in T_i\setminus T_i^j$ such that $\text{line}_j(y)$ is typical as per Definition~\ref{def:typical-line-toy} one can match $\widetilde{\text{line}_j(y)}$, which constitutes a $(1-\eta)(1-1/k)$ fraction of $\text{line}_j(y)$, to $\text{line}_j(y)\cap (T_i\setminus T_i^j)$ through the edges of the complete bipartite graph $\widetilde{\text{line}_j(y)} \times (\text{line}_j(y) \cap (T_i\setminus T_i^j))$.  At the same time the number of $y$'s that belong to atypical lines is at most $2\eta^2 |T_i|=O(\eta^2 k) |S_i|$ by conditioning on $\E_{large-lines}$ and the high probability event $\E_{set-sizes}$ from Lemma~\ref{lm:set-sizes}. While this would have sufficed for proving a $1-1/e$ lower bound, we would like to get a lower bound of $1-(1-1/k)^k$ for every $k\geq 2$. For that we need the slightly harder
\begin{lemma}\label{lm:match-local}
For every $\eta\in (0, 1/2)$, if $C>0$ is a sufficiently large constant, then for $m\geq C\eta^{-2} k \log \eta^{-1}$ a multiple of $k$,  conditioned on $\E_{large-lines}$ (defined in Lemma~\ref{lm:line-properties}) and $\E_{set-sizes}$ (defined in Lemma~\ref{lm:set-sizes}) for every $i=0,\ldots, k-1$, every $(j_1,\ldots, j_i)\in B_1\times \ldots \times B_i$ for each $j\in B_{i+1}$ there exists  a matching of at least $(1-O(\eta+\eta^2 k))|S_i|-N^{-\Omega(\eta^2/m^2)}$ nodes in $S_i$ to $T_i\setminus T_i^j$ for sufficiently large $N$.
\end{lemma}
\begin{proof}
Let $C>0$ be sufficiently large as prescribed by Lemma~\ref{lm:line-properties}.
We prove the existence of the required matching by exhibiting a fractional matching of appropriate size, which implies the result by the integrality of the bipartite matching polytope. The construction proceeds over three steps.

\noindent {\bf Step 1} For every $x\in S_i$ such that $\text{line}_j(x)$ is typical put fractional mass $k/m$ on every edge in $\widetilde{\text{line}_j(x)} \times (\text{line}_j(x) \cap (T_i\setminus T_i^j))$. 
Since $|\text{line}_j(x) \cap (T_i\setminus T_i^j)|=m/k$ by Lemma~\ref{lm:line-properties}, {\bf (2)}, this places a unit of mass on the neighborhood of every vertex in $\widetilde{\text{line}_j(x)}$.
Since $|\widetilde{\text{line}_j(x)}|=(1-\eta)(1-1/k)\cdot m/k$ by definition, this places fractional mass $(1-\eta)(1-1/k)$ on every $y\in \text{line}_j(x)\cap (T_i\setminus T_i^j)$, leaving at least $1/k$ capacity on each such $y$. We assign more fractional mass to use the remaining $1/k$ mass up to an $O(\eta)$ term in step 2.

\noindent {\bf Step 2} For every $x\in S_i\setminus S_i^j$ put fractional mass 
\begin{equation}\label{eq:def-eps}
\e:=\frac1{(m/k)\cdot (1+\eta)(1-1/k) (|B_j|-1)}
\end{equation} on every edge connecting $x$ to $y\in T_i$.  Note that these edges correspond to coordinates $j'\in B_{i+1}\setminus \{j\}$. In particular, if $(x, y)$ is an edge corresponding to coordinate $j'$, then we have $y_q=x_q$ for all $q\neq j'$, and in particular it must be that $y\in T_i\setminus T_i^j$. 

\noindent{\bf Step 3} Let $\deg_j(x)$ denote the number of $j'\in B_{i+1}\setminus \{j\}$ such that $x\in T_i^{j'}$, and let $\overline{\deg_j}(y)$ denote the number of $j'\in B_{i+1}\setminus \{j\}$ such that $y\in T_i\setminus T_i^j$. We now remove all fractional mass assigned to vertices  $x\in S_i$ with $\overline{\deg_j}(x)\not \in (1\pm \eta) \frac1{k} (|B_{i+1}|-1)$ and vertices $y\in T_i$ with $\deg_j(z)\not \in (1\pm \eta) (1-1/k) (|B_{i+1}|-1)$. We refer to such nodes as {\em atypical}.

We now prove upper and lower bounds on the fractional mass assigned by this rule to every $x\in S_i^j, y\in T_i\setminus T_i^j$. This establishes feasibility of the fractional solution and lower bounds its value respectively.

\paragraph{Upper bounding load (feasibility).} For every $j'\in B_{i+1}\setminus \{j\}$ every vertex $x$ is either connected to exactly $|\text{line}_j(x)\cap (T_i\setminus T_i^j)|=m/k$ nodes in $T_i\setminus T_i^j$ with edges in $E_i^{j'}$ or zero nodes (when $x$ belongs to an atypical line in direction $j'$). In the former case coordinate $j'$ contributes exactly $\e\cdot (m/k)$ fractional mass (where $\e$ is defined in~\eqref{eq:def-eps}), and in the latter it contributes $0$. We now get that the total mass contributed to $x$ by directions $j'\neq j$ is no larger than $\deg_j(x)\cdot (m/k)\cdot \e=\deg_j(x)\cdot (m/k)\cdot \frac1{(1+\eta)(m/k)\cdot (1-1/k) (|B_j|-1)}$.  By Lemma~\ref{lm:typical-degrees} for all but $N^{1-\Omega(\eta^2/m^2)}$ of $x\in S_i$ one has
\begin{equation}\label{eq:deg-x-lub}
\deg_j(z)\in (1\pm \eta) (1-1/k) (|B_{i+1}|-1).
\end{equation}
We call such $x$ typical.
We  thus get that the total mass assigned to edges incident on typical $x\in S_i\setminus S_i^j$ is upper bounded by $(1+\eta) (1-1/k) (|B_{i+1}|-1)\cdot (m/k)\cdot \frac1{(m/k)(1+\eta)\cdot (1-1/k) (|B_j|-1)}\leq 1$, and the fractional assignment is feasible for all but $N^{1-\Omega(\eta^2/m)}$ nodes (i.e. for all typical nodes as per definition above). 

Similarly, get by Lemma~\ref{lm:typical-degrees} for all but $N^{1-\Omega(\eta^2/m^2)}$  of $y\in T_i\setminus T_i^j$ one has 
\begin{equation}\label{eq:deg-y-lub}
\overline{\deg_j}(y)\in (1\pm \eta) \frac1{k}(|B_{i+1}|-1).
\end{equation}
Now note that $|\widetilde{\text{line}_j}(x)|=(1-\eta)m(1-1/k)/k$ for every $x$ such that the corresponding line is typical.  The degree in $E_i^j$ of a vertex $y\in T_i\setminus T_i^j$  such that $\text{line}_{j'}(y)$  is thus exactly $(1-\eta)m(1-1/k)/k$ if the corresponding line is typical, and is zero otherwise. The amount of mass assigned to $y$ is thus $\overline{\deg_j}(y)\cdot ((1-\eta)m (1-1/k)/k)\cdot \frac1{(1+\eta)(m/k)\cdot (1-1/k)(|B_j|-1)}\leq (1-\eta)/k\leq 1/k$. Thus, together with the amount of mass assigned in {\bf Step 1} to vertices $y\in T_i\setminus T_i^j$, our assignment is feasible for all but $N^{1-\Omega(\eta^2/m)}$ nodes (i.e. for all typical nodes as per definition above).  

\paragraph{Lower bounding fractional matching size.}  In Step 1 we assigned $(1-\eta)(1-1/k)$ to every node in $y\in T_i\setminus T_i^j$ that belongs to a typical line in direction $j$. The number of such nodes is at least $(1-O(\eta^2 k))|S_i|$ by Lemma~\ref{lm:line-properties}, {\bf (3)} together with Lemma~\ref{lm:set-sizes}, since we condition on $\E_{large-lines}$ and $\E_{set-sizes}$. In Step 2 we assigned $\e:=\frac1{(m/k)\cdot (1+\eta)(1-1/k) (|B_j|-1)}$ mass to every edge from $x\in S_i\setminus S_i^j$ to $y\in T_i\setminus T_i^j$ along some direction $j'\in B_{i+1}\setminus \{j\}$ if the corresponding line is typical.  Thus, for every $j'$ we assigned $\e \cdot (m/k)$ mass to every $x$ that belonged to a typical line in direction $j'$ (all but $O(\eta^2 k) |S_i|$ such $x$ for every direction $j'$ by conditioning on $\E_{large-lines}$). Altogether $x\in S_i\setminus S_i^j$ thus contributed at least 
\begin{equation*}
\begin{split}
&\sum_{j'\in B_{i+1}\setminus \{j\}} \sum_{\substack{x\in S_i\setminus S_i^j: x \text{~typical~and~}\\\text{line}_{j'}(x)~\text{~typical}}} \e\cdot (m/k)\\
&\geq \sum_{j'\in B_{i+1}\setminus \{j\}} \left(-\e\cdot (m/k)\cdot \eta^2 |T_i|+\sum_{x\in S_i\setminus S_i^j: x\text{~typical}} \e\cdot (m/k)\right)\\
&=-\eta^2 \e (m/k)\cdot |B_{i+1}|\cdot |T_i|+ \sum_{x\in S_i\setminus S_i^j: x\text{~typical}} \e\cdot (m/k)\cdot \deg_j(x)\\
&=-O(\eta^2) |T_i|+ \sum_{x\in S_i\setminus S_i^j: x\text{~typical}} \e\cdot (m/k)\cdot \deg_j(x),\\
\end{split}
\end{equation*}
where we used the fact that, conditioned on $\E_{large-lines}$, by Lemma~\ref{lm:line-properties}, {\bf (3)} for every $i$ and every $j\in B_{i+1}$ all but $\eta^2 |T_i|$ belong to typical lines in direction $j$, as well as the definition of $\e$ in~\eqref{eq:def-eps}. We now lower bound the second term:

\begin{equation*}
\begin{split}
\sum_{x\in S_i\setminus S_i^j: x\text{~typical}} \e\cdot (m/k)\cdot \deg_j(x)&\geq \sum_{x\in S_i\setminus S_i^j: x\text{~typical}} \e\cdot (m/k) (1- \eta) (1-1/k) (|B_{i+1}|-1)\\
&\geq  \sum_{x\in S_i\setminus S_i^j: x\text{~typical}} (1-O(\eta))\\
&\geq (1-O(\eta)) |S_i\setminus S_i^j|-N^{-\Omega(\eta^2/m^2)},
\end{split}
\end{equation*}
where the first transition is by definition of typical $x$, and the second is by Lemma~\ref{lm:typical-degrees}. Putting the bounds above together shows that we constructed a fractional matching of size at least $(1-O(\eta+\eta^2 k))|S_i|-N^{-\Omega(\eta^2/m^2)}$, as required.
\end{proof}

%Note that in every phase $t=1,\ldots, k+1$ the algorithm is presented with a union of matchings $M_{t, j}, j\in B_t$, after which one of these matchings (matching $M_{t, j_t}$) is selected uniformly at random, and the corresponding coordinate is used to define $Q_{t+1}$, the part of the graph that received subsequent edges. In the figure below (Fig.~\ref{laminar1}) the red edges correspond to the matching $M_{1,j_1}$, and the black edges correspond to the remaining induced matchings.
%\begin{figure}[h]
%  \centering
%  \includegraphics[scale=0.5]{}
%\caption{Edge set of $G_1=(P_1, Q, E_1)$. The matching $M_{1, j_1}$ is shown in red.}  \label{laminar1}
%\end{figure}

%This defines the edges between $P_t$ and $Q_t$ for $t=1,\ldots, k$. The set $P_{k+1}$ is a special set of size equal to the size of $Q_{k+1}$. There is an arbitrary perfect matching between $P_{k+1}$ and $Q_{k+1}$ (see Fig.~\ref{laminar2}).

%\begin{figure}[h]
%  \centering
%  \includegraphics[scale=0.5]{}
% \caption{Edge set of $G$.}   \label{laminar2}
%\end{figure}

\subsection{Hard input distribution and its analysis}\label{sec:2}
\paragraph{Hard input distribution.} First select values of random variables $\{Z_x\}_{x\in [m]^n}$ so that $\E_{set-sizes}$ and $\E_{large-lines}$ occur (we will verify that this is feasible later in the proof of Theorem~\ref{thm:main}, where we set parameters).  The input graph is generated as follows. First for every $i=0,\ldots, k-1$ let $j_i$ be uniformly random in $B_i$. Then for each $i=0,\ldots, k-1$ the edges of the graph induced by $S_i\cup T_i$, namely $E_i$ (defined in~\eqref{eq:edges-ei-def})  arrive in the stream in an arbitrary order.  Finally, a perfect matching of $T_k$ to a fresh set $S_k$ of vertices on the $S$ side arrives.  We denote this distribution over input graphs by ${\mathcal D}$. In this section we are assuming a stylized model, where after every stage the algorithm must select $s=o(N\log N)$ edges to keep in memory, and at the end of the stream must output a  matching in the subgraph that it maintained.  We show in Theorem~\ref{thm:main} that no such algorithm can achieve a better than $1-1/e$ approximation to maximum matching. More specifically, we show that no algorithm can achieve a significantly better than factor $1-(1-1/k)^k$ approximation on a $k$-stage input instance for every constant $k\geq 2$.

\paragraph{Intuition for the construction and lower bound.} We will show in that in order to have performance better than $1-(1-1/k)^k+\delta$ on our instance the algorithm needs to store at least $\Omega(\delta N/k)$ edges from at least one of the sets $E_i^{j_{i+1}}$ (see~\eqref{eq:edges-ei-def}), for some $i=0,\ldots, k-1$. However, since at each step $j_{i+1}$ is uniformly random in $B_{i+1}$  this is impossible if the algorithm can only store $s=o(N\log N)$ edges (i.e. any sublinear fraction of the total number of edges in the graph).

The analysis relies on the several auxiliary lemmas. First, we show that the input graph contains a large matching:
\begin{lemma}\label{lm:matching-size-lb}
For every $\eta\in (0, 1/2)$, if $C>0$ is a sufficiently large constant, then for $m\geq C\eta^{-2} k \log \eta^{-1}$ a multiple of $k$,  conditioned on $\E_{large-lines}$ (defined in Lemma~\ref{lm:line-properties}) and $\E_{set-sizes}$ (defined in Lemma~\ref{lm:set-sizes}), every $(j_1,\ldots, j_k)\in B_1\times \ldots \times B_k$ the graph $G=G(j_1,\ldots, j_k)$ contains a matching of size at least $(1-O(\eta+\eta^2 k))|S|$ if $N$ is sufficiently large.
\end{lemma}
\begin{proof}
Let $C>0$ be sufficiently large as dictated by Lemma~\ref{lm:match-local}. Now by Lemma~\ref{lm:match-local} for every $i=0,\ldots, k-1$ match at least $(1-O(\eta+\eta^2 k))|S_i|-N^{-\Omega(\eta^2/m^2)} |T|$ of $S_i$ to $T_i\setminus T_i^j$. Then match $S_k$ to $T_k$. For every fixed $k, \eta, m$, if $N$ is sufficiently large (i.e. if $n=\log_m N$ is sufficiently large), one has $N^{-\Omega(\eta^2/m^2)}<\eta/k$ and is thus absorbed in the $O(\eta)$ error term.
\end{proof}

\if 0 The basic idea to prove that no online algorithm can give a better than $1-\frac{1}{e}$ approximation, is that the online algorithm can not pick many edges from the set $M_{i,j_{i+1}}$, since only after the edges arrive in the $i^{th}$ stage, is the value of $j_{i+1}$ revealed.
Hence, the algorithm can only get $\frac1{|B_{i+1}|}=\frac{1}{\frac{n}{k}}$ edges of this matching in expectation, which can be arbitrarily small. We then argue that at least some of the edges of the matchings $M_{i, j_{i+1}}$ are critical to constructing a matching of size larger than $(1-1/e)$ fraction of the optimum.
\fi

The following lemma is the source of hardness of our input instance:
\begin{lemma}\label{lm:sparse-cut}
For every $k\geq 2$, every $\eta\in (0, 1/2)$, every integer $m$ a multiple of $k$, every $(j_1,\ldots, j_k)\in B_1\times \ldots \times B_k$, for every $i=0,\ldots, k-1$ for every edge $(x, y)$, $x\in S^*_i$ either $y\in T_k$ or $(x, y)\in E_i^{j_{i+1}}$.
\end{lemma}
\begin{proof}
Consider a edge $(x,y)$ with $x \in S_i^*$ that is not in $E_i^{j_{i+1}}$. We  now show that $y \in T_k$, proving the lemma.
Let $j\neq j_{i+1}\in B_{i+1}$ be such that $(x, y)\in E_i^j$ -- such a $j$ exists by definition of the edge set $E_i$ (recall~\eqref{eq:edges-ei-def}).  This in particular means that $j_r \neq j ~\text{~for all~}r=1,\ldots, k$, since $j\in B_{i+1}$ and the blocks $B_r, r=1,\ldots, k$ are disjoint. By definition of $E_i^j$ we have $y \in T_i\setminus T_i^j$ and $x \in S_i^j$. Furthermore,  we have $(x-y)_s = 0$ for all $s \neq j$.  We thus have $x_{j_{i+1}} = y_{j_{i+1}}$ for all $i=0,\ldots, k-1$. But since $x_{j_{i+1}} > \frac{m}{k}$ for all $i=0,\ldots, k-1$ (by definition of $S^*_i$ in~\eqref{eq:s-star-i-def} and assumption that $x \in S_i^*$), this implies $y_{j_{i+1}} > \frac{m}{k}$ for all $i=0,\ldots, k-1$, so $y\in T_k$ (by definition of $T_k$, see~\eqref{eq:def-ti-si}).\\
\end{proof}

\begin{lemma}\label{lm:cut}
For every $k\geq 2$, $\eta\in (0,1/2)$, if $m$ is an integer multiple of $k$ such that $m\geq C\eta^{-2} k \log \eta^{-1}$ for a sufficiently large constant $C>0$, and if the input graph $G=G(j_1,\ldots, j_r)$ is selected according to the input distribution ${\mathcal D}$ defined above, the following conditions hold. If the streaming algorithm, after being presented with edges revealed in the $i$-th stage for $i=0,\ldots, k-1$, must store a number of edges after each phase, with the overall set of edges remembered over all stages denoted by $E'$, then any matching $M_{ALG}$ contained in $E'$ satisfies
$$
|M_{ALG}|\leq \left(1-(1-1/k)^k\right)|T|+\sum_{i=0}^{k-1} |E_i^{j_{i+1}}\cap E'|+O(\eta)|T|.
$$
\end{lemma}
\begin{proof}
Let the constant $C>0$ be sufficiently large as dictated by Lemmas~\ref{lm:line-properties} and~\ref{lm:match-local}.
We consider the standard reduction of bipartite matching to max-flow (i.e. connect source $s$ to $S$, sink $t$ to $T$) and exhibit a cut in the graph $(S\cup \{s\}, T\cup \{t\}, E')$ of value at most 
$\left(1-(1-1/k)^k\right)|T|+\sum_{i=0}^{k-1} |E_i^{j_{i+1}}\cap E'|+O(\eta)|T|$. By max-flow/min-cut theorem this gives the result.

We now exhibit a cut in this graph and upper bound its size. The source side of the cut is $\{s\}\cup S_k\cup T_k\cup \bigcup_{i=0}^{k-1} S^*_i$. By Lemma~\ref{lm:sparse-cut} edges incident on $S^*_i, i=0,\ldots, k-1$ either belong to the matching $M_i^{j_{i+1}}$ or go to $T_k$, so edges incident on $S^*_i, i=0,\ldots, k-1$ contribute at most $\sum_{i=0}^{k-1} |E_i^{j_{i+1}}\cap E'|$ to the cut value. We thus have that the value of the cut is bounded by
\begin{equation}\label{eq:cut-value}
|T_k|+\sum_{i=0}^{k-1} |S_i\setminus S^*_i|+\sum_{i=0}^{k-1}|E_i^{j_{i+1}}\cap E'|.
\end{equation}

It remains to bound the size of $T_k$, as well as the sizes of $S_i\setminus S^*_i$. We condition on the event $\E_{set-sizes}$ and Lemma~\ref{lm:set-sizes}. Conditioned on this event we have $|T_k|=(1-1/k)^k |T|$ and 
$$
|S^*_i|=(1\pm \eta) |T_k|/k=(1+O(\eta))(1-1/k)^k/k.
$$
Similarly, we have by Lemma~\ref{lm:set-sizes}, {\bf (1)}  that $|T_i|=(1-1/k)^i|T|$, and thus by Lemma~\ref{lm:set-sizes}, {\bf (2)} that $|S_i|=(1\pm O(\eta)) (1-1/k)^i |T|/k$. Using these bounds we get

\begin{equation*}
\begin{split}
\sum_{i=0}^{k-1} |S_i\setminus S^*_i|&=\sum_{i=0}^{k-1} |S_i|-\sum_{i=0}^{k-1} |S^*_i|\\
&=(1\pm O(\eta))\sum_{i=0}^{k-1} (1-1/k)^i |T|/k-(1\pm O(\eta)) (1-1/k)^k |T|\\
&=(1\pm O(\eta)) (1-(1-1/k)^k)|T|-(1\pm O(\eta)) (1-1/k)^k |T| \text{~~~~(by summing the geometric series)}\\
&=(1\pm O(\eta)) (1-2(1-1/k)^k)|T|
\end{split}
\end{equation*}

Putting the bounds above together with~\eqref{eq:cut-value}, we thus have that the size of the cut is bounded by 
\begin{equation*}
\begin{split}
&|T_k|+\sum_{i=0}^{k-1} |S_i\setminus S^*_i|+\sum_{i=0}^{k-1}|E_i^{j_{i+1}}\cap E'|\\
&=(1-1/k)^k |T|+(1\pm O(\eta)) (1-2(1-1/k)^k)|T|+\sum_{i=0}^{k-1}|E_i^{j_{i+1}}\cap E'|\\
&=(1\pm O(\eta))(1-(1-1/k)^k) |T|+\sum_{i=0}^{k-1}|E_i^{j_{i+1}}\cap E'|\\
&=\left(1-(1-1/k)^k\right)|T|+\sum_{i=0}^{k-1} |E_i^{j_{i+1}}\cap E'|+O(\eta)|T|
\end{split}
\end{equation*}
as required.
\end{proof}

We now prove
\begin{theorem}\label{thm:matching-ub-toy}
For every $k\geq 2$, for any $\eta\in (0,1/k^3)$, if $m\geq C\eta^{-2} k\log \eta^{-1}$ for a sufficiently large absolute constant $C>0$ is a multiple of $k$, then if the graph $G=G(j_1,\ldots, j_k)$ is selected according to the input distribution ${\mathcal D}$ defined above, and the algorithm, after being presented with edges revealed in the $i$-th stage, stores $s=o(\log N)$ edges, the following conditions hold. If $M_{ALG}$ is the maximum matching in the set of edges $E'$ that the algorithm stored over all the stages, one has
$$
|M_{ALG}|\leq \left(1-(1-1/k)^k\right)|T|+O(\eta)|T|
$$
with probability at least $99/100$. 

\end{theorem}
\begin{proof}
Denote the set of edges that the algorithm commits to after seeing the subgraph $S_i\times T_i$ by $\tilde E_i$. By Lemma~\ref{lm:cut} the size of the matching that the algorithm outputs at the end is upper bounded by
$$
(1-1/k)^k |T|+\sum_{i=0}^{k-1} |E_i^j\cap \widetilde E_i|.
$$
We will show that with high probability $\sum_{i=0}^{k-1} |E_i^j\cap \widetilde E_i|\leq \sum_{i=0}^{k-1} s /|B_{i+1}|=O(k^2\cdot s/n)$, where $s$ is the number of edges that the algorithm stores at every step. Recall that for each $i=0,\ldots, k-1$, conditioned on $(j_1,\ldots, j_i)$, the special index $j_{i+1}$ is chosen uniformly at random in $B_{i+1}$, implying that 
$$
\expect_{j_{i+1}}\left[\left.\left| E_i^{j_{i+1}} \cap \widetilde E_i\right|\right| j_1,\ldots, j_i\right]=\frac1{|B_{i+1}|}\sum_{j\in B_{i+1}} |E_i^j \cap \widetilde E_i|=\frac{|\widetilde E_i|}{|B_{i+1}|}=s/|B_{i+1}|.
$$
Summing over all $i=0,\ldots, k-1$, we get
$$
\expect\left[\sum_{i=0}^{k-1} |E_i^j\cap \widetilde E_i|\right]=\sum_{i=0}^{k-1} s /|B_{i+1}|=O(k^2 \cdot s/n)=o(k^2 |T|)=o(|T|),
$$
since $s=o(\log N)$ by assumption of the theorem and $\log N=\Theta(n)$ (as $m$ is a constant). The result now follows by Markov's inequality.
\end{proof}

\begin{theorem}\label{thm:main-toy}

For every $k\geq 2$, every $\delta\in (0, 1)$, there exists an input distribution ${\mathcal D}$  on bipartite graphs such that any streaming algorithm that stores $s=o(N\log N)$ edges achieves an approximation ratio of at most $1-(1-1/k)^k+\delta$.
\end{theorem}
\begin{proof}
Consider the distribution ${\mathcal D}$ with $\eta=c\delta/k^3$ for a sufficiently small constant $c>0$ and $m\geq C\eta^{-2} k \log \eta^{-1}$ a multiple of $k$ for the constant $C>0$ from Lemma~\ref{lm:line-properties}. Then by Theorem~\ref{thm:matching-ub-toy} one has 
$$
|M_{ALG}|\leq \left(1-(1-1/k)^k\right)|T|+O(\eta)|T|\leq \left(1-(1-1/k)^k\right)|T|+(\delta/2)|T|=\left(1-(1-1/k)^k+\delta/2\right)N
$$
with probability at least $99/100$.  At the same time by Lemma~\ref{lm:matching-size-lb}, conditioned on conditioned on $\E_{large-lines}$ (defined in Lemma~\ref{lm:line-properties}) and $\E_{set-sizes}$ (defined in Lemma~\ref{lm:set-sizes}), every $(j_1,\ldots, j_k)\in B_1\times \ldots \times B_k$ the graph $G=G(j_1,\ldots, j_k)$ contains a matching of size at least $(1-O(\eta+\eta^2 k))|S|\geq (1-\delta/10)N$ if $N$ is sufficiently large. Thus, the approximation ratio achieved by the algorithm is at most 
$$
\frac{1-(1-1/k)^k+\delta/2}{1-\delta/10}\leq 1-(1-1/k)^k+\delta,
$$
as required.

It remains to note that by Lemma~\ref{lm:line-properties} the event $\E_{large-lines}$ occurs with probability at least $1-k(\log N)^{k} e^{-\Omega(\eta N/k)}\geq 1-o(1)$ over the choice of $Z_x$'s since $k$ and $\eta$ are independent of $N$ by our setting of parameters. Similarly, by Lemma~\ref{lm:set-sizes} the event $\E_{set-sizes}$ occurs with probability at least 
$1-k(\log N)^{k} e^{-\Omega(\eta N/k)}\geq 1-o(1)$. Thus, the upper bound on the approximation ratio achieve by the algorithm holds with probability at least $99/100-o(1)\geq 98/100$, as required.

\end{proof}

%!TEX root = ./doc-t.tex 
\section{Single pass streaming lower bound}\duallabel{sec:main}
\renewcommand{\S}{\mathcal{S}}
\renewcommand{\E}{{\mathcal E}}

\newcommand{\T}{\mathcal{T}}
\renewcommand{\u}{\mathbf{u}}
\renewcommand{\v}{\mathbf{v}}
\newcommand{\w}{\mathbf{w}}
\renewcommand{\P}{\mathcal{P}}

In the rest of the section we define a distribution on input instances for our problem of approximating maximum matchings in a single pass in the streaming model. 
 Our construction follows its simple version presented in Section~\ref{sec:toy-construction}. A major difference is that we replace coordinate directions with an exponential size family of nearly orthogonal vectors, thereby achieving a lower bound of $n^{1+\Omega(1/\log\log n)}$ on the space complexity of obtaining a better than $1-1/e$ approximation in a single pass. This approach is inspired by techniques for constructing \rs graphs pioneered in \cite{montest} and extensions developed in \cite{gkk:streaming-soda12}.

%\section{Hard input distribution}\duallabel{sec:packing}
\subsection{Construction of host graphs $G(\u_1,\ldots, \u_k)$}\label{sec:host-graph-construction}
We first introduce  notation.  Each graph in our family of host graphs will be indexed by a $k$-tuple of vectors $(\u_1,\ldots, \u_k)\in \F_1\times \ldots \times \F_k$, where $\F_j, j=1,\ldots, k$ are families of vectors in $\bool^m$. We choose $\F_j$ so that vectors in $\bigcup_{j=1}^k \F_j=:\F$ are of equal Hamming weight and nearly orthogonal. Specifically, the following lemma guarantees the existence of a large family $\F$ such that for every $\u, \v \in \F$ it holds that $|\u|=|\v|=w$ and 
$(\u, \v)\leq \e w$, where $w=\Theta(\e^2 m)$. Here for a vector $\u\in \bool^m$ we write $|\u|$ to denote the Hamming weight of $\u$.
We assume from now on that $|\F_1|=|\F_2|=\ldots=|\F_k|=d$ for a parameter $d$. The lemma below shows that we can have $d=2^{\Omega (\e^2 m)}$. The specific form of the dependence of the exponent on $\e$ will not be important for the qualitative nature of our results, however, as we will ultimately set $\e$ to be a small constant.

We will use the following standard lower bound on the size of such families:
\begin{lemma}\label{lem:code}
For any $\e\in (0, 1)$, any integers  $m\geq 1$ and $w=(\e/2)m$, there exists a collection $\mathcal F_{m, w, \e} \subset \bool^m$  of vectors of Hamming weight $w$ with $\log |\mathcal{F}_{m, w,\e}| = \Omega(\e^2 m)$ such that for all $\u\neq \u'\in \mathcal F_{w,\e}$, $(\u, \u') < \e w$.
\end{lemma}
\begin{proof}
The proof is via the probabilistic method. Partition $[m]$ into $w$ subsets $I_1, \ldots, I_w$, with $|I_s|=m/w$ for $s=1,\ldots, w$. We pick $\u_1,\ldots,\u_N$ independently as follows. For every $j=1,\ldots, N$, the vector $\u_j$ includes exactly one random element of $I_s$ for each $s=1,\ldots, w$. This ensures that the Hamming weight of each $\u_j$ is exactly $w$.  

We now show that the vectors have small intersection size with high probability. Fix $i\neq j\in[N]$. Imagine $\u_i$ being fixed and picking the $w$ elements of $\u_j$ one by one. Let $X_s$ denote the indicator random variable for the event that the $s$th element of $\u_j$ (picked from $I_s$) is also in $S_i$. Then $(\u_i, \u_j) = \sum_{s=1}^w X_k$, and we set $\mu:= \expect[(\u_i, \u_j)]$. Note that $\mu=(w/m)\cdot w$, since for every $s=1,\ldots, w$ the vector $\u_i$ has exactly one nonzero coordinate in $I_s$, and the probability that $\u_j$ chooses the same coordinate is $1/|I_s|=w/m$. 
We have $\prob[(\u_i, \u_j) \ge \e w] = \prob[\sum_{s=1}^w X_s \ge 2\mu]$  The random variables $X_s$ are independent and thus the Chernoff bound yields
$$
\prob[(\u_i, \u_j) \geq  2\mu) \le \left(\frac{e}{4}\right)^\mu \le e^{-\Omega((w/m) w)}\leq e^{-c\e^2 m}
$$
for a constant $c>0$. 
Setting $N = 2^{(\ln_2 e) c \e^2 m/2}$ so that ${N \choose 2}<N^2=2^{(\ln_2 e) c \e^2 m}=e^{c \e^2 m}$, by a union bound with positive probability $|\u_i\cap \u_j| < \e w$ for all $i\neq j$, simultaneously, as desired. Note for this choice of $N$, we have $\log|\mathcal \F_{m, w,\e}| = \log N = \Theta(\e^2 m)$.
\end{proof}

 We also associate with each $\u\in \F_j, j=1,\ldots, k$ a random variable $U_{\u}$ that is uniformly distributed over the integers 
\begin{equation}\label{eq:shifts-range} 
\{0,1,\ldots, k/\theta-1\}\cdot W\cdot (\theta/k),
\end{equation} 
where $\theta\in (0, 1)$ is a parameter that we will set to a small constant times $1/\text{poly}(k)$, and $W$ is a parameter that will later set to $\text{poly}(k)\cdot w$ (where $w$ is the Hamming weight of the vectors in the collection $\F$). The variables $U_{\u}$ and $U_{\v}$ are independent for $\u\neq \v$.

%We will use the following lemma from \cite{gkk:streaming-soda12}, which is a convenient formulation of the construction of error correcting codes with fixed weight in \cite{lev71}
%\begin{lemma}\cite{gkk:streaming-soda12}\duallabel{lm:many-vectors}
%For sufficiently large $m>0$, any constant $\e\in (0, 1)$ and constant $\gamma>2$  there exists a family $\mathcal{F}$ of subsets of $[m]$ of size $\e m$ with intersection at most $\gamma \e^2 m$ such that  $\frac1{m}\log |\mathcal{F}|\geq c_{\e, \gamma}-o(1)$.
%\end{lemma}

 As before, the sides of the bipartition of the graph $G(\u_1,\ldots, \u_k)$ that we need to construct are denoted by $T$ and $S=S_0\cup\ldots\cup S_{k}$, where $S_0\cup S_1\cup \ldots \cup S_k$ is a partition of $S$. We  use the notation $[a]=\{1, \ldots, a\}$ for integer $a\geq 1$. In our construction the $T=T^0$ side of the graph is identified with  a hypercube $[m^4]^m$ for a value of $m$ to be chosen later, and each set $S_i, i=0,\ldots, k-1$ is identified with a subsampled hypercube $[m^4]^m$. The vertices of the last set $S_k$ do not have any special structure. 
 Vertices $x\in T$ or $y\in S_i$ will often be treated as points $x, y\in [m^4]^m$.  For $x\in T$ and $\u\in \F$ we use the dot product notation $ (x, \u)=\sum_{i=1}^m x_i\cdot \u_i \in \mathbb{Z}$. For an interval $[a, b]$ and a number $W$ we will write $[a, b]\cdot W$ to denote the set of integers belonging to the interval $[a\cdot W, b\cdot W]$. Finally, for an integer $i$ and an integer $W$ we will write $i \mod W$ to denote the residue of $i$ modulo $W$ that belongs to $[0, W-1]$.

\subsubsection{Defining the vertex set of the graph $G(\u_1,\ldots, \u_k)$}
We will use
\begin{definition}[Ground sets $X, Y$]\label{def:xy}
Let $X^*=Y=[m^4]^m$ for some integer $m>0$. Let $X$ be a random subset of $X^*$ where each point of $X^*$ appears independently with probability $1/k$.
\end{definition}
We will refer to vertices in $X$ and $Y$ as points in $[m^4]^m$. The host graph $G$ is generated by first selecting a $k$-tuple $(\u_1,\ldots, \u_k)\in \F_1\times \ldots \times \F_k$, and then defining the vertex and edge set as we describe below.  Before proceeding with the construction, we list relevant parameters here. 

\paragraph{Parameters of the construction} 
\begin{itemize}
\item $k$ -- the number of phases in the hard input distribution;
\item $w$ --- Hamming weight of binary vectors in $\F$;
\item $\e$ -- upper bound on the maximum dot product of any pair of distinct vectors from $\F$, normalized by their Hamming weight $w$;
\item $\eta$ -- small constant governing separation of red and blue vertices  -- see~\eqref{eq:strips-1} and~\eqref{eq:strips-2}.
\end{itemize}

\paragraph{Sets of red, white and blue vertices $R, W, B$.} Consider fixed $\w\in \bool^m$, and let
\begin{equation}\duallabel{eq:strips-1}
\begin{split}
R^Y(\w)&=\{y\in Y: ((y,  \w)+U(\w))\mod  W\in [0, 1/k)\cdot W\}\\
&\text{(red vertices with respect to $\w$)}\\
W^Y(\w)&=\{y\in Y: ((y,  \w)+U(\w)) \mod W\in ([1/k,1/k+\eta)\cup [1-\eta, 1))\cdot W\}\\
&\text{(white vertices with respect to $\w$)}\\
B^Y(\w)&=\{y\in Y: ((y,  \w)+U(\w))\mod W\in [1/k+\eta, 1-\eta)\cdot W\}\\
&\text{(blue vertices with respect to $\w$)}\\
\end{split}
\end{equation}
It is convenient to also define

\begin{equation}\duallabel{eq:strips-2}
\begin{split}
R^{X^*}(\w)&=\{x\in X^*: ((x,  \w)+U(\w))\mod  W\in [0, 1/k)\cdot W\}\\
W^{X^*}(\w)&=\{x\in X^*: ((x,  \w)+U(\w)) \mod W\in ([1/k,1/k+\eta)\cup [1-\eta, 1))\cdot W\}\\
B^{X^*}(\w)&=\{x\in X^*: ((x,  \w)+U(\w))\mod W\in [1/k+\eta, 1-\eta)\cdot W\},\\
\end{split}
\end{equation}
as well as let
\begin{equation}\duallabel{eq:strips-3}
\begin{split}
R^X(\w)=R^{X^*}(\w)\cap X, ~~W^X(\w)=W^{X^*}(\w)\cap X, ~~~B^X(\w)=B^{X^*}(\w)\cap X.
\end{split}
\end{equation}

The intuition for these sets is simple: ideally, we would like to partition vertices $y\in T_i$ into two classes, depending on whether their dot product with $\w$ is in $[0, 1/k)\cdot W$ (red points) or $[1/k, 1)\cdot W$(blue points; we ignore the shift $U(\w)$ for this intuitive discussion), and then match points in one color class in $X$ to points in the other color class in $Y$. This would work fine if the set $\F$ of vectors that we use contained orthogonal vectors only, as in our toy construction in Section~\ref{sec:toy-construction}. Since the family of vectors $\F$ that we use consists of vectors with small (constant) dot products, we need a `buffer' between the two classes above, provided by the set of white vertices $W^Y$.

\paragraph{Nested sequence of sets $T_i$ and sets $S_i$.} For all $i=0,\ldots, k$  let
\begin{equation}\duallabel{eq:ti-def}
\begin{split}
T_i&=\{y\in Y : ((y, \u_j)+U(\u_j)) \mod W\in [1/k, 1)\cdot W,\text{~for all~} j\in [1:i]\}\\
S_i&=\{x\in X : ((x, \u_j)+U(\u_j)) \mod W\in [1/k, 1)\cdot W,\text{~for all~} j=[1:i]\},
\end{split}
\end{equation}
so that $T_0=Y$ and $S_0=X$. For every $i=0,\ldots, k-1$ and $\w\in \F_{i+1}$ we let 
\begin{equation}\duallabel{eq:tiw-def}
\begin{split}
T_i^\w&=\{y\in Y : ((y, \u_j)+U(\u_j)) \mod W\in [1/k, 1)\cdot W,\text{~for all~} j\in [1:i]\\
&~~~~~~~~~~~~~~~~~~\text{~~~~and~}\\
&~~~~~~~~~~~~~~~~~~~~~((y, \w)+U(\w)) \mod W\in [1/k, 1)\cdot W\}\\
S_i^\w&=\{x\in X : ((x, \u_j)+U(\u_j)) \mod W\in [1/k, 1)\cdot W,\text{~for all~} j=[1:i]\\
&~~~~~~~~~~~~~~~~~~\text{~~~~and~}\\
&~~~~~~~~~~~~~~~~~~~~~((x, \w)+U(\w)) \mod W\in [1/k, 1)\cdot W\}
\end{split}
\end{equation}
Also, let 
\begin{equation}\duallabel{eq:siu-def}
\begin{split}
S_i^*=\{x\in S_i &: ((x, \u_j)+U(\u_j)) \mod W\in [1/k, 1)\cdot W,\\
&\text{~for all~} j\in [1:k]\}.\\
\end{split}
\end{equation}
Note that the sets $S_i^*$ are obtained from $S_i$ by adding extra constraints on dot products with vectors $\u_j$, namely for $j=i+1,\ldots, k$ (this is because all vertices in $S_i$ already satisfy the constraints above for $j=1:i$ by definition of $S_i$). 

\paragraph{Vertex set of $G(\u_1,\ldots, \u_k)$.} The graph $G(\u_1, \ldots, \u_k)$ whose edges we define shortly will be a bipartite graph with the sides of the bipartition given by $T=T_0$ and $S=S_0 \cup \ldots \cup S_k$, where $T_0$ and $S_i, i=0,1,\ldots, k$ are as defined above.  Note that the union of $S_i$'s in the definition of $S$ is understood as a disjoint union. In other words, vertices in both $S$ and $T$ are naturally labelled with points on the hypercube in $[m^4]^m$. These labels are distinct for vertices in $T$, but not for vertices in $S$. However, such labels are distinct for vertices in $S_i$ for every $i=0,1,\ldots, k$. We denote the number of vertices on the $Q$ side of the bipartition, i.e., in $T$, by $n$. We will have $|S|=O(|T|)$, so that the total number of vertices in our instance is $O(n)$.

\newcommand{\ind}{\mathbf{1}}
\renewcommand{\w}{\mathbf{w}}

\if 0
Consider fixed $i$ between $0$ and $k-1$. For all $\w\in \F_{i+1}$ let 
\begin{equation}\duallabel{eq:strips-1}
\begin{split}
R^Y_i(\w)&=\{y\in T^{i}: ((y,  \w)+U(\w))\mod  W\in [0, 1/k]\cdot W\}\\
&\text{(red vertices with respect to $\w$)}\\
W^Y_i(\w)&=\{y\in T^{i}: ((y,  \w)+U(\w)) \mod W\in ([1/k,1/k+\eta]\cup [1-\eta, 1))\cdot W\}\\
&\text{(white vertices with respect to $\w$)}\\
B^Y_i(\w)&=\{y\in T^{i}: ((y,  \w)+U(\w))\mod W\in [1/k+\eta, 1-\eta]\cdot W\}\\
&\text{(blue vertices with respect to $\w$)}\\
\end{split}
\end{equation}
Define $R^X_i(\w), W^X_i(\w), B^X_i(\w)$ for $\w \in \F_{i+1}$ similarly (note that these sets are defined only for $S_i$):
\begin{equation}\duallabel{eq:strips-2}
\begin{split}
R^X_i(\w)&=\{x\in S_{i}: ((x,  \w)+U(\w))\mod  W\in [0, 1/k]\cdot W\}\\
W^X_i(\w)&=\{x\in S_{i}: ((x,  \w)+U(\w)) \mod W\in ([1/k,1/k+\eta]\cup [1-\eta, 1))\cdot W\}\\
B^X_i(\w)&=\{x\in S_{i}: ((x,  \w)+U(\w))\mod W\in [1/k+\eta, 1-\eta]\cdot W\}\\
\end{split}
\end{equation}

Note that the sets $R^Y_i(\w), W^Y_i(\w)$ and $B^Y_i(\w)$ depend on the index $i$, but we will omit the subscript to simplify notation when $i$ is fixed.
\fi

\renewcommand{\S}{\mathcal{S}}
\paragraph{Estimates on the size of $T_i, S_i, T_i^\w, S_i^\w, R, B, W$.} 

We will need the following lemma, whose proof is given in Appendix~\ref{app:aux}
\begin{lemma}\label{lm:intersection-size}
For every $m\geq 2$, integer $W\geq 1$ and $\delta'\in (0, 1)$ such that $1/\delta'$ is an integer, if $Y=[m^4]^m$ and the set $\S$ is defined by 
$$
\S=\{y\in Y: (y, \u) +\Delta_\u \mod W\in [a_\u, b_\u)\cdot W, \text{~for all~}\u\in \mathcal{U}\},
$$ where $\mathcal{U}$ is a collection of binary vectors of fixed length $w$ and $a_\u, b_\u\in [0, 1]$  are constant integer multiples of $1/L$ for an integer $L$, the following conditions hold if $W$ is an integer multiple of $w\cdot \text{lcm}(L, 1/\delta')$, $\Delta_\u/W$ are multiples of $1/L$  and $m$ is sufficiently large. 

 If $\max_{\substack{\u\in \mathcal{U}, \v\in \mathcal{U}\\\u\neq \v}} (\u,\v)/|\v|\leq \delta'$, then 
$$
\left||\S|-|Y|\cdot \prod_{\u\in \mathcal{U}} (b_\u-a_\u)\right|\leq |\mathcal{U}|^2(6L\delta'  +4/m)\cdot |Y|.
$$
\end{lemma}

We now apply Lemma~\ref{lm:intersection-size} to bound the size of various relevant subsets of $Y$ and $X$. We gather the resulting bound in the following 
\begin{lemma}\label{lm:bounds-on-set-sizes}
There exists an event $\E$ over $X$ that occurs with probability at least $99/100$ such that the following bounds hold conditioned on $\E$.

For every $\e\in (0, 1)$, every integer $k\geq 2$, every choice of shifts $U(\u) \in \F$, every $(\u_1,\ldots, \u_k)\in \F_1\times \F_2\times\ldots\times \F_k$, if $W$ (see~\eqref{eq:ti-def} and~\eqref{eq:siu-def}) is an integer multiple of $w\cdot k/(\e\cdot \theta)$ (see~\eqref{eq:shifts-range}), then for sufficiently large integer $m$
\begin{description}
\item[(1)] $|T_i|=(1-1/k)^i |Y|+\Delta_i$, $|\Delta_i|=O(k^3 \e/\theta)\cdot |Y|$ for every $i\in \{0, 1, 2, \ldots, k\}$;
\item[(2)] $|S_i|=\frac1{k}((1-1/k)^i |Y|+\Delta_i$, $|\Delta_i|= O(k^3 \e/\theta)\cdot |Y|$ for every $i\in \{0, 1, 2,\ldots, k\}$;
\item[(3)] $|S_i^*|=\frac1{k}(1-1/k)^k |Y|+\Delta_i$, $|\Delta_i|=O(k^3 \e/\theta)\cdot |Y|$;
\item[(4)] $|T_i^\w|=(1-1/k)^{i+1} |Y|+\Delta_i$, $|\Delta_i|=O(k^3 \e/\theta)\cdot |Y|$  for every $i\in \{0, 1, 2, \ldots, k-1\}$ and $\w \in \F_{i+1}$;
\item[(5)] $|S_i^\w|=\frac1{k}(1-1/k)^{i+1} |Y|+\Delta_i$, $|\Delta_i|=O(k^3 \e/\theta)\cdot |Y|$  for every $i\in \{0, 1, 2, \ldots, k-1\}$ and $\w \in \F_{i+1}$.
\end{description}
\end{lemma}
\begin{proof}
We fix the values of the shifts $U(\u), \u\in \F$ as well as the sequence $(\u_1,\ldots, \u_k)\in \F_1\times \F_2\times \ldots \times \F_k$, and take a union bound over such fixings later (this is important for establishing the bounds on various subsets of $S$, as those depend on the random choice of $X$; see~\dualeqref{eq:ti-def}, ~\dualeqref{eq:tiw-def} and~\dualeqref{eq:siu-def}).  By~\eqref{eq:ti-def} we have
\begin{equation*}
\begin{split}
T_i&=\{y\in Y : ((y, \u_j)+U(\u_j)) \mod W\in [1/k, 1)\cdot W,\text{~for all~} j\in [1:i]\}\\
S_i&=\{x\in X : ((x, \u_j)+U(\u_j)) \mod W\in [1/k, 1)\cdot W,\text{~for all~} j=[1:i]\}.
\end{split}
\end{equation*}

We start with $T_i$, where we apply Lemma~\ref{lm:intersection-size} with ${\mathcal U}=\{\u_1,\ldots, \u_i\}$. We thus have $a_\u=\{1/k+U(\u)/W\}$, $b_\u=\{1+U(\u)/W\}$ for all $\u\in {\mathcal U}$ (where $\{\cdot \}$ stands for the fractional part of the argument). Since $U(\w)$ are integer multiples of $\theta W/k$ by definition (see ~\eqref{eq:shifts-range}), we get that setting $L=k/\theta$ ensures that $a_\u, b_\u$ are multiples of $1/L$. Recall that vectors 
$\u_1,\ldots, \u_k\in \{0, 1\}^m$ have Hamming weight $w$ and $\max_{s\neq t} (\u_s, \u_t)/w\leq \e$ by assumption of the lemma. Since further $W$ is an integer multiple of $w\cdot k/(\e\cdot \theta)$, we get that indeed $W$ is an integer multiple of $w\cdot \text{lcm}(L, 1/\e)$, and hence the preconditions of Lemma~\ref{lm:intersection-size} are satisfied.
We now get by Lemma~\ref{lm:intersection-size}, using the fact that $\prod_{\u\in \mathcal{U}} (b_\u-a_\u)=(1-1/k)^i$ by our setting of parameters, that 
\begin{equation*}
\begin{split}
\left||T_i|-(1-1/k)^i |Y|\right|&\leq |\mathcal{U}|^2(6L\e +4/m)\cdot |Y|\\
&\leq k^2(6(k/\theta)\e +4/m)\cdot |Y|\\
&=O(k^3 \e/\theta)\cdot |Y|,
\end{split}
\end{equation*}
where in the last step we used the assumption of our lemma that $m$ if sufficiently large as a function of  $k$ and $\e$.
This proves {\bf (1)}. The proof of {\bf (4)} is similar, with $\mathcal{U}=\{\u_1,\ldots, \u_i, \w\}$.

Similarly, since every element of $X^*=Y$ appears in $X$ independently with probability $1/k$ (see Definition~\ref{def:xy}), an application of Lemma~\ref{lm:intersection-size} as above shows that
$$
\expect_{X}[|S_i|]=\frac1{k}((1-1/k)^i |Y|\pm O(k^3 \e/\theta)\cdot |Y|).
$$
We thus also get by an application of Chernoff bounds (Theorem~\ref{thm:chernoff-bounds}) we get that for every $i=0,\ldots, k$
\begin{equation}\label{eq:ei-prob}
\prob\left[\left||S_i|-\expect[|S_i|]\right|>\e\cdot |Y|\right]<2e^{-\e^2 |Y|/(3k)}.
\end{equation}
The bound above is for a fixed choice of the shifts $U(\u), \u\in \F$. The number of such choices is bounded by $(m^5)^{2^m}\leq e^{|Y|^{1/2}}$ when $m$ is larger than a constant. The number of choice of $(\u_1,\ldots, \u_k)\in \F_1\times \F_2\times \ldots\times \F_k$ is bounded by $2^{k 2^m}\leq e^{|Y|^{1/2}}$ as well, and therefore we have  $\left||S_i|-\expect[|S_i|]\right|\leq \e\cdot |Y|$ for every choice of shifts  with probability at least 
$1-2e^{-\e^2 |Y|/(6k)}$. Denote the success event by $\E_i$.  Conditioned on $\bigcap_{i=0}^{k-1} \E_i$ one has $\left||S_i|-\frac1{k}((1-1/k)^i |Y|\right|=O(k^2 \e/\theta)\cdot |Y|$, proving {\bf (2)}. 

Similarly, since every element of $X^*=Y$ appears in $X$ independently with probability $1/k$ (see Definition~\ref{def:xy}), an application of Lemma~\ref{lm:intersection-size} as above shows that for every $\w\in \F_{i+1}$
$$
\expect_{X}[|S_i^\w|]=\frac1{k}((1-1/k)^{i+1} |Y|\pm O(k^3 \e/\theta)\cdot |Y|).
$$
We thus also get by an application of Chernoff bounds (Theorem~\ref{thm:chernoff-bounds}) we get that for every $i=0,\ldots, k-1$ and $\w\in \F_{j+1}$
\begin{equation}\label{eq:eiw-prob}
\prob\left[\left||S_i^\w|-\expect[|S_i^\w|]\right|>\e\cdot |Y|\right]<2e^{-\e^2 |Y|/(3k)}.
\end{equation}
Similarly to the above, we take a union bound over all fixings of shifts $U(\u), \u\in \F$ and choices of $(\u_1, \u_2,\ldots, \u_k)\in \F_1\times \F_2\times\ldots\times \F_k$, getting an upper bound of $2e^{-\e^2 |Y|/(6k)}$ on the probability of the failure event. Denote the success event by $\E_{i, \w}$.  Conditioned on $\bigcap_{i=0}^{k-1}\bigcap_{\w \in \F_{i+1}} \E_{i, \w}$ one has $\left||S_i^\w|-\frac1{k}((1-1/k)^{i+1} |Y|\right|=O(k^3 \e/\theta)\cdot |Y|$, proving {\bf (5)}.

Finally, recall that by~\eqref{eq:siu-def} one has for every $i\in 0, 1, \ldots, k-1$
\begin{equation*}
\begin{split}
S_i^*&=\{x\in S_i : ((x, \u_l)+U(\u_l)) \mod W\in [1/k, 1)\cdot W,\text{~for all~} l\in [1:k]\}\\
&=\{x\in [m^4]^m\cap X : ((x, \u_l)+U(\u_l)) \mod W\in [1/k, 1)\cdot W,\text{~for all~} l\in [1:k]\}\\
&=T_k\cap X.
\end{split}
\end{equation*}

We now have by {\bf (1)} that 
\begin{equation*}
\begin{split}
\left||T_k|-(1-1/k)^k |Y|\right| &=O(k^3 \e/\theta)\cdot |Y|.
\end{split}
\end{equation*}

We thus get, since $X$ contains every element of $[m^4]^m$ independently with probability $1/k$ by Definition~\ref{def:xy}, that $\expect_X[|S_i^*|]=|T^k|/k$, and thus by Chernoff bounds (Theorem~\ref{thm:chernoff-bounds})
\begin{equation}\label{eq:eij-prob}
\prob\left[\left||S_i^*|-\expect[|S_i^*|]\right|>\e\cdot |Y|\right]<2e^{-\e^2 |Y|/(3k)}.
\end{equation}
Similarly to the above, we take a union bound over all fixings of shifts $U(\u), \u\in \F$ and choices of $(\u_1, \u_2,\ldots, \u_k)\in \F_1\times \F_2\times\ldots\times \F_k$, getting an upper bound of $2e^{-\e^2 |Y|/(6k)}$ on the probability of the failure event.  Denote the success event by $\E_{i, *}$. Conditioned on $\bigcap_{i=0}^{k-1}\E_{i, *}$ one has 
$\left||S_i^*|-\frac1{k}((1-1/k)^* |Y|\right|=O(k^3 \e/\theta)\cdot |Y|$, as required.

We now let $\E=\left(\bigcap_{i=1}^k \E_i\right)\cap \left(\bigcap_{i=1}^k \E_{i, *}\right)\cap \left(\bigcap_{i=0}^{k-1}\bigcap_{\w \in \F_{i+1}} \E_{i, \w}\right)$. Using a union bound together with~\eqref{eq:ei-prob},~\eqref{eq:eiw-prob} and~\eqref{eq:eij-prob} we get
\begin{equation*}
\begin{split}
\prob[\E]&\geq 1-\sum_{i=0}^{k-1} \prob[\bar \E_i]-\sum_{i=0}^{k-1} \prob[\bar \E_{i, *}]-\sum_{i=0}^{k-1}\sum_{\w \in \F_{i+1}} \prob[\bar \E_{i, \w}]\\
&\geq 1-2k 2^m 2e^{-\e^2 |Y|/(6k)}\\
&\geq 1-2k 2^m 2e^{-\e^2 m^{4m}/(6k)}\\
&\geq 99/100\\
\end{split}
\end{equation*}
as long as $m$ is sufficiently large as a function of $\e$ and $k$.
\end{proof}

\newcommand{\A}{{\mathcal{A}}}
\renewcommand{\B}{{\mathcal{B}}}

\subsubsection{Defining the edge set of the graph $G(\u_1,\ldots, \u_k)$}

%We now define the edges of the $((k-1)\gamma, k\gamma, O(\delta))$-almost regular induced subgraph $H^w_i$, for a constant $\gamma>0$ (the induced property will be shown later). The subgraph $H^w_i$ will consist of disjoint copies of small complete bipartite graphs.

First, the only edges incident on vertices in $S_k$ are the edges of a perfect matching between $S_k$ and $T^k$. In the rest of the section we define edges incident on $S_i$, $i=0,\ldots, k-1$. The following definition will be useful in the analysis. Let $\text{Bad}\subseteq [m^4]^m$ be defined by 
\begin{equation}\label{eq:bad-set}
\begin{split}
\text{Bad}:=\{y\in [m^4]^m: \exists i\text{~s.t.~}y_i<m^2\text{~or}~y_i>m^4-m^2\}.
\end{split}
\end{equation}
Note that $|\text{Bad}|\leq 2m^3/m^4=o(1)$ by a union bound.

For each $i=0,\ldots, k-1$ we will have $\Gamma_G(S_i)\subseteq T_i$ (see~\eqref{eq:ti-def} for the definitions of $S_i$ and $T_i$). The edges incident to $S_i$ can be partitioned into an induced union of nearly regular constant degree subgraphs. Each such subgraph is indexed by a vector $\w \in \F_{i+1}$, and is denoted by $H^{\w}_i$. We now give the construction of these subgraphs.

%\paragraph{Constructing $H^\w_i$.}
%Fix $\w\in \F_{i+1}$. In what follows we omit the parameter $\w$ when referring to sets $R^Y(\w), W^Y_i(\w), B^Y_i(\w)$ and $R^X_i(\w), W^X_i(\w), B^X_i(\w)$, as well as omit the subscript $i$, since both $\w$ and $i$ are fixed in this construction. 

The graph $H^\w_i$ is a disjoint union of constant size complete bipartite graphs, where each such constant size graph corresponds to a set of points on the integer lattice that lie on a short line segment in direction $\w$ (recall that $\w\in \bool^m$). In what follows we first define the relevant lines (Sets $\mathcal{L}^Y$ and $\mathcal{L}^X$), and then define the edges of $H^\w_i$.

\paragraph{Defining sets of lines $\mathcal{L}^Y$ and $\mathcal{L}^X$.}  For an arbitrary  $y\in R^Y$
let 
$$
\ell(y):=\left\lfloor \frac{(y, \w)+U(\w)}{W}\right\rfloor,
$$
and 
define 
\begin{equation}\label{eq:line-y-def}
\text{line}^Y(y, \w):=\left\{y'\in R^Y: y'=y+\lambda\cdot \w \text{~~such that~~}\ell(y')=\ell(y)\right\}.
\end{equation}

Similarly, for $x\in B^{X^*}$ let
\begin{equation*}
\text{line}^{X^*}(x, \w):=\left\{x'\in B^{X^*}: x'=x+\lambda\cdot \w\text{~~such that~~}\ell(x')=\ell(x)\right\}.
\end{equation*}
and  for $x\in B^X$ let
\begin{equation}\label{eq:line-x-def}
\text{line}^X(x, \w):=\left\{x'\in B^X: x'=x+\lambda\cdot \w\text{~~such that~~}\ell(x')=\ell(x)\right\}.
\end{equation}

Note that for every fixed $\w$ and every pair $x_1, x_2\in X\setminus \text{Bad}$ one has either $\text{line}^{X^*}(x_1, \w)=\text{line}^{X^*}(x_2, \w)$ or $\text{line}^{X^*}(x_1, \w)\cap \text{line}^{X^*}(x_2, \w)=\emptyset$. Analogous properties hold for $\text{line}^X$ and $\text{line}^Y$. Let 
$$
\mathcal{L}^X(\w)=\bigcup_{x\in X\setminus \text{Bad}} \text{line}^X(x, \w)\text{~~and~~}\mathcal{L}^Y(\w)=\bigcup_{y\in Y\setminus \text{Bad}} \text{line}^Y(y, \w).
$$

Note that for every $y\in Y\setminus \text{Bad}$ and every $\w$ there exists a unique line $L^X\in \mathcal{L}^X(\w)$ such that for every $x\in L^X$ and $y\in \text{line}(y, \w)=:L^Y$ one has $x-y=\lambda \w$ for some integer $\lambda$ and $\ell(x)=\ell(y)$. We call $L^X$ the pair of $L^Y$. We denote the function mapping $y$-lines to their corresponding pair $x$-lines by $\pi_\w: \mathcal{L}^Y(\w)\to \mathcal{L}^X(\w)$. Let $\mathcal{L}^{X^*}(\w)$ and $\pi_\w^*$ be defined analogously.

We now give bounds on the size of lines. We start with
\begin{claim}[Size of $\text{line}^Y(y, \w)$]\label{cl:y-line-size}
If $m\geq W/|\w|$ and $W/|\w|$ is an integer multiple of $k$, then for all $y \in R^Y\setminus \text{Bad}$, $\w \in \bigcup_{i=1}^k \F_i$  one has $|\text{line}^Y(y, \w)|=W/(k|\w|)$.
\end{claim}
\begin{proof}
First note that for every integer $\lambda$, $\w \in \bigcup_{i=1}^k \F_i$ and every $y\in Y$ one has, letting $y'=y+\lambda \w$,
$$
(y', \w)+U(\w)=((y+\lambda\cdot \w,  \w)+U(\w))=((y,  \w)+U(\w))+\lambda |\w|.
$$
Note that every $\lambda$ that results in $\ell(y+\lambda \w)=\ell(y)$ satisfies $|\lambda|\leq W/|\w|$, and thus by our assumption $y\in Y\setminus \text{Bad}$ we have 
$y+\lambda \w\in Y$
since $m\geq W/|\w|$ by assumption of the claim.

Recall that $y'\in \text{line}^Y(y, \w)$ amounts to  two conditions: $y'\in R^Y$ and $\ell(y')=\ell(y)$, where the former constraint is 
\begin{equation}\duallabel{eq:923t32ddsf}
((y', \w)+U(\w)) \pmod{W}\in [0, 1/k)\cdot W.
\end{equation}
Letting $z:=\left\lfloor\frac{((y,  \w)+U(\w))\text{~mod~} W}{|\w|}\right\rfloor$, we note that the set of values  of $\lambda$ results in~\dualeqref{eq:923t32ddsf} being satisfied at the same time as $\ell(y')=\ell(y)$ is exactly  $\{-z, -z+1,\ldots, -z+W/(k|\w|)-1\}$. We thus have that
$|\text{line}^Y(y, \w)|=W/(k|\w|)\text{~~~for all~}y\in Y\setminus \text{Bad}$, as required.  
\end{proof}

We have 
\begin{claim}[Size of $\text{line}^X(x, \w)$]\label{cl:x-line-size}
For every $\e\in (0, 1)$, every $\eta \in (0, 1/10)$ such that $1/\eta$ is an integer, if $m\geq W/|\w|$ is sufficiently large and $W/|\w|$ is an integer multiple of $\text{lcm}(k, 1/\eta)$, $W/|\w|\geq \frac{48k^2\ln(1/\eta)}{\e^2}$, the following conditions hold.

With probability at least $99/100$ over the choice of $X\subseteq X^*$  for every setting of $U(\w), \w\in \bigcup_{i=1}^k \F_i,$  for all but $2\eta^2 |Y|$ points $x \in B^X\setminus \text{Bad}$, every $\w \in \bigcup_{i=1}^k \F_i$  one has $|\text{line}^X(x, \w)|\geq \frac1{k}(1-1/k)(1-4\eta-\e) W/|\w|$ for sufficiently large $m$ as a function of $\eta, k, W/|\w|$ and $\e$. 
\end{claim}
\begin{proof}
For every integer $\lambda$ and every $x\in X^*\setminus \text{Bad}$ one has 
$$
((x+\lambda\cdot \w,  \w)+U(\w))=((x,  \w)+U(\w))+\lambda |\w|.
$$

Letting $z:=\left\lfloor\frac{((y,  \w)+U(\w))\text{~mod~} W}{|\w|}\right\rfloor$, we note that the set of values  of $\lambda$ results in the equation above being satisfied at the same time as $\ell(y')=\ell(y)$ is exactly  $\{-z+\frac{W}{|\w|}\left(\frac1{k}+\eta\right), \ldots, -z+\frac{W}{|\w|}(1-\eta)-1\}$. 
We thus get that the set of values of $\lambda$ that result in $x'=x+\lambda\cdot \w\in \text{line}^{X^*}(x, \w)$ has size $(1-1/k-2\eta)W/|\w|$, as required. 

Note that every $\lambda$ that results in $\ell(x+\lambda \w)=\ell(x)$ satisfies $|\lambda|\leq W/|\w|$, and thus by our assumption $x\in X^*\setminus \text{Bad}$ we have 
$x'+\lambda \w\in X^*$ since $m\geq W/|\w|$ by assumption of the claim. In order to establish the claim, it suffices to analyze the sampling process involved in constructing $X$ from $X^*$.

Since for every $L^{X^*}$ the set $L^X$ is a random subsampling of  $L^{X^*}$, where each element of $X^*$ is included in $X$ independently with probability $1/k$, we have 
$\expect\left[|L^X|\right]=\frac1{k}|L^{X^*}|=\frac1{k}(1-1/k-2\eta)W/|\w|$ for every $y\in Y\setminus \text{Bad}$. Since $X$ is obtained from $X^*$ by independent sampling at rate $1/k$, we get by the Chernoff bound (Theorem~\ref{thm:chernoff-bounds}) 
\begin{equation*}
\prob\left[|L^X|\not \in (1\pm \e) \frac1{k}\left(1-1/k-2\eta\right) |L^{X^*}|\right]\leq 2e^{-\e^2 |L^{X^*}|/(12k)},
\end{equation*}
where we used the fact that $1-1/k-2\eta\geq 1-1/2-1/5\leq 1/4$, since $\eta<1/10$ by assumption of the claim. Since 
\begin{equation*}
\begin{split}
(1-\e)\left(1-1/k-2\eta\right)-(1-\e-4\eta)(1-1/k)&=(1-\e)(1-1/k)-(1-\e)2\eta\\
&-(1-\e)(1-1/k)+4\eta(1-1/k)\\
&=2\eta(-1+\e+2(1-1/k))\\
&\geq 2\eta \e\geq 0\text{~~~~(since $k\geq 2$)},
\end{split}
\end{equation*}
we in particular have 
\begin{equation}\duallabel{eq:prob-bad}
\prob\left[|L^X|< (1-\e-4\eta) \frac1{k}\left(1-1/k\right) |L^{X^*}|\right]\leq 2e^{-\e^2 |L^{X^*}|/(12k)}.
\end{equation}

Now as long as $|L^{X^*}|>48k\ln(1/\eta)/\e^2$, we have that the rhs above is upper bounded by $\eta^2$. Since $|L^{X^*}|=W/(k|\w|)$ by~Claim~\ref{cl:y-line-size}, this follows since $W\geq \frac{48k^2\ln(1/\eta)}{\e^2}\cdot |\w|$ by assumption of the claim.

Finally, note that for every $\w$ we just showed that a single line $L^X$ deviates from expectation with probability at most $\eta^2$. Since distinct lines do not overlap, an application of Chernoff bounds shows that for every $\w$ the probability that the number of lines $L^X\in \mathcal{L}^X(\w)$ that deviate from expectation is at most $2\eta^2 |\mathcal{L}^X(\w)|$ with probability at least 
$$
1-2e^{-\Omega(\eta^2 |\mathcal{L}^{X^*}|)}=1-2 e^{-\Omega(\eta^2 m^{4m} /(W/|\w|))}.
$$
A union bound over at most $2^m$ vectors $\w\in \F$ and at most $(m^5)^{2^m}$ choices for the shifts $U(\w), \w\in \F,$ yields failure probability at most $2 \cdot 2^m \cdot (m^5)^{2^m} \cdot e^{-\Omega(\eta^2 m^{4m} /(k W/|\w|))} <1/100$ as long as $m$ is sufficiently large as a function of $\eta, k$ and $W/|\w|$. 
\end{proof}

% We now define a set of edges of a $((k-1)\gamma, k\gamma, \delta)$-almost regular subgraph between (a subset of) $L^Y$ and $\pi^*(L^Y)$. 

\paragraph{Defining bipartite cliques induced by $L^Y\cup \pi^*(L^Y)$.}  We start with
\begin{definition}[Typical lines]\label{def:typical-line}
For $\w \in \bigcup_{i=1}^k \F_i$ and $L^Y\in \mathcal{L}^Y(\w)$ we say that $L^Y$ and its pair $\pi_\w(L^Y)$ are {\em typical} if $|\pi_\w(L^Y)|\geq \frac1{k}(1-1/k)(1-4\eta-\e) W/|\w|$.
\end{definition}

If $L^Y$ is typical as per Definition~\ref{def:typical-line}, let $\tilde L^Y$ denote an arbitrary subset of $L^Y$ of cardinality $(1-4\eta-\e)|L^Y|$. Similarly, let $\tilde L^X$ denote an arbitrary subset of $\pi_\w(L^Y)$ of cardinality $(1-4\eta-\e)(1-1/k)|L^Y|$. Our parameter setting will ensure that $|L^Y|=W/(k|\w|)$ is an integer multiple of $lcm(1/\eta, 1/\e, k)$, so this is feasible. For convenience let $\tilde L^X=\tilde L^Y=\emptyset$ for lines that are not typical. We thus have $|\tilde L^X|=(1-1/k)|\tilde L^Y|$. Now let
\begin{equation}\label{eq:e-ly}
\tilde E(L^Y):=\tilde L^Y\times \tilde L^X.
\end{equation}

Note that for a typical line $L^Y\in \mathcal{L}^Y(\w)$ the degree of a vertex $y\in L^Y$ in $\tilde E(L^Y)$ is either zero or $(1-4\eta-\e)(1-1/k)|L^Y|=(1-4\eta-\e)(1-1/k)\cdot \frac{W}{|\w|}=:(1-1/k)\gamma$, where we let 
\begin{equation}\label{eq:gamma-def}
\gamma:=(1-4\eta-\e)\frac{W}{|\w|},
\end{equation}
 and the degree of a vertex $x\in \pi(L^Y)$ is either zero or $(1-4\eta-\e)|L^Y|=(1-4\eta-\e)\frac{W}{|\w|}=\gamma$.
Also note that all edges in the graph that we just defined are of the form $(c, d)$, where 
\begin{equation}\duallabel{eq:edges}
c=d+\lambda \cdot \w, |\lambda|\leq \frac{W}{|\w|}.
\end{equation}

\paragraph{Defining the edges of $H^\w_i, \w\in \F_{i+1}$.} Let 
$$
\tilde E^\w_i:=\bigcup_{y\in T_i\setminus T_i^\w} \tilde E(\text{line}^Y(y, \w)),
$$
where $\tilde E(\text{line}^Y(y, \w))$ is defined in~\eqref{eq:e-ly}, $T_i$ is defined in~\eqref{eq:ti-def} and $T_i^\w$ is defined in~\eqref{eq:tiw-def}. Note that since our vectors are only nearly orthogonal, for a $y\in T_i$ we do not necessarily have $\text{line}^Y(y, \w)\subseteq T_i$. We now let 
\begin{equation}\label{eq:ewi}
E^\w_i:=\tilde E^\w_i\cap (T_i\times S_i).
\end{equation}

\subsubsection{Induced property of subgraphs $H^{\w}_i$}
We now show that the graphs $H^\w_i$ constructed above are induced for each $i$ and $\w\in \F_i$. The argument is similar to \cite{montest,gkk:streaming-soda12}.  

\begin{claim}\label{cl:induced}
For every $\eta\in (0, 1)$, integer $k\geq 2$, if $\e\in (0, 1)$ is smaller than $\eta$, then for every $i=0,\ldots, k-1$, the edge set $E^\w_i$ (defined in~\eqref{eq:ewi}) is an induced union of subgraphs $H_i^\w$.
\end{claim}
\begin{proof}
Recall that $\F_{i+1}$ was chosen as a family of binary vectors of fixed weight with small intersections, namely
for every $\w, \w'\in \F_{i+1}, \w\neq \w'$ one has
\begin{equation}\duallabel{eq:orthogonal}
(\w, \w')\leq \e |\w|.
\end{equation} 
Suppose that an edge $(c, d)\in E(H^\w_i), c\in X, d\in Y$  is induced by 
$H^{\w'}_i$ for $\w'\neq \w$. 
Since edges of $H^{\w'}_i$ connect red points in $Y$ with respect to $\w'$ to blue points in $X$ with respect to $\w'$ (see~\eqref{eq:line-y-def}, ~\eqref{eq:line-x-def} and the definition of edges in $H^\w_i$ in ~\eqref{eq:e-ly} and~\eqref{eq:ewi}; see also~\eqref{eq:line-y-def} and~\eqref{eq:line-x-def}), it must be that $d\in R^Y(\w')$ and $c\in B^X(\w')$,  so
\begin{equation}\duallabel{eq:induced-1}
|(c-d, \w')|\geq \eta\cdot W.
\end{equation}
However, by \dualeqref{eq:orthogonal} together with \dualeqref{eq:edges} one has 
$$
|(c-d, \w')|=|\lambda| \cdot (\w, \w')\leq \frac{W}{|\w|} (\w, \w')\leq \frac{W}{|\w|} \e |\w|=  \e W<\eta W, 
$$
since $\e<\eta$ by assumption of the claim. This yields a contradiction with \dualeqref{eq:induced-1}, and hence $H^\w_i$ are induced. 
\end{proof}

\renewcommand{\ind}{\mathbf{1}}

\subsubsection{Existence of a large matching in the host graph}

We now show that with high probability over the choice of the random shifts $U(\v), \v\in \bigcup_{i=1}^k \F_i$, for any $i=0,\ldots, k-1$ any collection $\u_1,\ldots, \u_{i-1}, \u_s\in \F_s, s=1,\ldots, i-1$ and $\w\in \F_{i+1}$ there exists a matching of $1-O(k^3 \e/\eta)$ fraction of $S_{i}$ to $T_i\setminus T_i^{\w}$. Formally we prove

\begin{claim}\label{cl:large-matching}
For every integer $k\geq 2$, sufficiently small $\eta\in (0, 1)$ such that $1/\eta$ is an integer,  $\e \in (0, c\cdot \eta^2/k^6)$ for a sufficiently small constant $c>0$ such that $1/\e^{1/2}$ is an integer, if $\theta=\eta$ (see~\dualeqref{eq:shifts-range})  and $W/w$ is an integer multiple of $k/(\e\cdot \theta)$, the following conditions hold for sufficiently large $m$.

There exists an event $\E_{balanced-degrees}$ that occurs with probability at least $99/100$ over the choice of random shifts $U(\v), \v\in \bigcup_{j=1}^k \F_j$, for every $i=0,\ldots, k-1$, every collection $\u_1,\ldots, \u_s\in \F_s, s=1,\ldots, i$ and every $\w\in \F_{i+1}$ such that conditioned on $\E_{balanced-degrees}$ and the event $\E$ from Lemma~\ref{lm:bounds-on-set-sizes} there exists a matching of $1-O(k^3 \e^{1/2}/\eta)$ fraction of $S_{i}$ to $T_i\setminus T_i^{\w}$.  
\end{claim}
\begin{proof}
We will do this by exhibiting a fractional matching of appropriate size. Recall that a fractional matching is an assignment of non-negative weights $z_e$ to edges $e$ of the graph such that for every vertex $v$ of the graph one has $\sum_{e\in \delta(v)} z_e\leq 1$.  We now exhibit a fractional matching in the graph in three steps. 

First, for every typical line $L^Y\in \mathcal{L}^Y(\w)$ that touches $T_i\setminus T^\w_i$ we assign weights to every edge of $\tilde E(L^Y)$ in such a way that every vertex in $L^Y$ that has nonzero degree in $\tilde E(L^Y)$ receives $1-1/k$ fractional mass, and every vertex in $\pi(L^Y)$ that has a nonzero degree receives mass $1$. Then we assign fractional mass uniformly to edges incident on vertices in $S_i\setminus S_i^\w$ to ensure that these vertices contribute the missing $1/k$ fraction of mass to vertices in $L^Y$, up to a small error term that is independent of $k$, the number of rounds in the game, and can be made arbitrarily small by choosing the maximum dot product $\e$ between vectors in $\bigcup_{j=1}^k \F_j$ small, and making the `buffer' between red and blue vertices appropriately small (this mass is assigned to edges in lines $L^Y\in \mathcal{L}^Y(\v)$ for $\v\in \mathcal{F}_{i+1}\setminus \{\w\}$). This ensures that the matching supported by the lines that touch  $T_i\setminus T^\w_i$ is about the size of $S_i$. The only problem is that this matching uses edges outside of $T_i\setminus T^\w_i$ and $S_i$. We then show that pruning to edges contained in $(T_i\setminus T^\w_i)\times S_i$ only affects matching size by a small error term, completing the proof.

\paragraph{Step 1: weights on edges of $H^\w_i$.} Recalling that for a typical line $L^Y\in \mathcal{L}^Y(\w)$ the degree of every vertex in $L^Y$ is either zero or $(1-1/k)\gamma$ (where $\gamma$ is defined in~\eqref{eq:gamma-def}), we put weight $1/\gamma$ on every edge of $\tilde E(L^Y)$. This way every vertex of nonzero degree in $L^Y$ gets fractional mass $1-1/k$, and every vertex of nonzero degree in $\pi(L^Y)$ gets fractional mass $1$.

\paragraph{Step 2: weights on edges of $H^\v_i$ for $\v \neq \w$.} We start by showing that for a fixed $\v$ and for every $y\in Y$ one has that $\prob_{U(\v)}[y\in R^Y(\v)]$ is very close to $\frac1{k}$. Indeed, recall that $U(\v)$ is uniformly random over the set 
$$
\{0, \ldots, k/\theta-1\}\cdot W\cdot (\theta/k),
$$
where $\theta\in (0, 1)$ is a parameter that by assumptions of the lemma is equal to $\eta$ (see~\eqref{eq:shifts-range}).   Using the definition of $R^Y(\v)$ (see~\eqref{eq:strips-1}) we can now bound
\begin{equation*}
\prob_{U(\v)}[y\in R^Y(\v)]=\prob_{U(\v)}[((y,  \v)+U(\v))\mod  W\in [0, 1/k)\cdot W].
\end{equation*}
Writing $(y, \v)=\left\lfloor \frac{(y, \v)}{W\cdot (\theta/k)}\right\rfloor \cdot W\cdot (\theta/k)+((y, \v) \text{~mod~} (W\cdot (\theta/k)))$ and recalling that $U(\v)$ is uniformly random in $\{0, \ldots, k/\theta-1\}\cdot W\cdot (\theta/k)$ by definition as well as that $1/\theta$ is an integer, we get that 
$$
\prob_{U(\v)}[((y,  \v)+U(\v))\mod  W\in [0, 1/k)\cdot W]=(1/\theta)/(k/\theta)=1/k,
$$
as required. A similar argument shows that for every $x\in S_i$ one has $\prob_{U(\v)}[x\in B^{X^*}(\v)]=1-1/k-2\eta$. Indeed, this is because
$$
\prob_{U(\v)}[((x,  \v)+U(\v))\mod  W\in [1/k+\eta, 1-\eta)\cdot W]=(k/\theta)(1-1/k-2\eta)/(k/\theta)=1-1/k-2\eta,
$$
where we used the assumption that $\theta=\eta$.

 Next note that each vertex $y\in R^Y(\w)\setminus \text{Bad}$ has degree $(1-1/k)\gamma$ or $0$ in $H^\v_i$, and for every $\v$ the fraction of vertices that have degree $0$ in $H^\v_i$ is at most $2\eta^2 |Y|$ by Claim~\ref{cl:x-line-size}. Furthermore, since the random shifts $U(\v)$ are independent for distinct $\v$, we obtain using Chernoff bounds (Theorem~\ref{thm:chernoff-bounds}) for $\delta\in (0, 1)$ that for  every $\w\in \F_{i+1}$ and every $y\in T_i\setminus T^\w_i$
\begin{equation*}
\prob_{\{U(\v)\}_{\v\in \F_{i+1}, \v\neq \w}}\left[\sum_{\v\in \F_{i+1}, \v\neq \w} {\bf I}[y\in R^Y(\v)]\not \in (1\pm \delta)d/k\right]\leq 2e^{-\delta^2 d/(3k)}.
\end{equation*}
Similarly we have for every $\w\in \F_{i+1}$ and  $x\in S_i\setminus S^\w_i$ and $\delta\geq 4\eta$
\begin{equation*}
\prob_{\{U(\v)\}_{\v\in \F_{i+1}, \v\neq \w}}\left[\sum_{\v\in \F_{i+1}, \v\neq \w} {\bf I}[x\in B^X(\v)]\not \in (1\pm \delta)d(1-1/k)\right]\leq 2e^{-\delta^2 d/24}.
\end{equation*}

Let $\E_{balanced-degrees}$ denote the event that for every every collection $\u_1,\ldots, \u_{i-1}, \u_s\in \F_s, s=1,\ldots, i-1$, every $\w\in \F_{i+1}$, every $y\in T_i$ one has 
$\sum_{\v\in \F_{i+1}, \v\neq \w} {\bf I}[y\in R^Y(\v)] \in (1\pm \delta)d/k$ (i.e. $y$ is a red vertex with respect to about the expected number of vectors $\v$) and $\sum_{\v\in \F_{i+1}, \v\neq \w} {\bf I}[x\in B^X(\v)]\not \in (1\pm \delta)d(1-1/k)$ (i.e. $x$ is a blue vertex with respect to about the expected number of vectors $\v$).
Since there are only $O(m^{4m})$ vertices in $Y$ and $X$, and $\left|\bigcup_i \F_i\right|\leq 2^{m}$, and $d=2^{\Omega(m)}$, for any constant $k,\eta, \delta$ and sufficiently large $m$ a union bound shows that $\E_{balanced-degrees}$ occurs with probability at least $99/100$.

The assignment of fractional weights on edges incident to vertices in $S_i\setminus S_i^\w$ is as follows: we put weight $\frac1{(1-1/k)\gamma\cdot  (1+\delta)d}$ on each edge of $\tilde E(L^Y)$ for $L^Y\in \mathcal{L}^Y(\v)$ that is incident on $y\in T_i\setminus T_i^\w$. We now verify feasibility of this solution in the presence of weights assigned in step 1, and then compute the size of the matching.

To verify feasibility, note that, conditioned on $\E_{balanced-degrees}$, the contribution of this assignment to any vertex in $S^i\setminus S^i(\w)$ is at most
$$
\frac1{ (1-1/k)\gamma\cdot  (1+\delta)d}\cdot (1+\delta) d (1-1/k)\cdot \gamma=1,
$$
where we used the fact that the degree of a vertex in $S_i\setminus S_i^\w$ in a subgraph induced by a typical line is at most $\gamma$. Contribution to any vertex in $T_i\setminus T_i^\w$ is at most 
$$
\frac1{(1-1/k)\gamma\cdot  (1+\delta)d}\cdot (1+ \delta) d/k \cdot (1-1/k) \gamma=1/k,
$$
where we used the fact that the degree of a vertex in $T_i\setminus T_i^\w$ in a subgraph induced by a typical line is at most $(1-1/k)\gamma$. Thus, the total mass assigned to vertices in $T_i\setminus T_i^\w$ as well as $S_i\setminus S_i^\w$ is upper bounded by $1$.

To lower bound the value of the fractional solution, first note that by Claim~\ref{cl:x-line-size} with high probability over the choice of $X$ for every $\w$ at most $2\eta^2 |Y|$ points belong to atypical lines (see Definition~\ref{def:typical-line}), which corresponds to a loss of at most $2\eta^2 |Y|$ in matching size. Now recall that a line is called typical (see Definition~\ref{def:typical-line}) if $|\pi_\w(L^Y)|\geq (1-1/k-4\eta-\e) W/|\w|$. Thus, at most a $4\eta+\e$ fraction of mass assigned is lost due to this. Since this applies to every $\v \in \F_{i+1}, \v\neq \w$, as well, we get that the constructed fractional matching is feasible, and its size is at least $(1-O(\delta+\e)) |S_i|-O(\eta^2) |Y|=(1-O(\eta k))|S_i|$, where we set $\delta=4\eta$ and used the fact that conditioned on the event $\E$ from Lemma~\ref{lm:bounds-on-set-sizes} one has $|S_i|=\Omega(1/k)|Y|$.

\paragraph{Step 3: bounding effect of truncation of $\tilde E^\w_i$ to $E^\w_i$} 
In steps 1 and 2 we showed that the mass we assigned to $\tilde E^\w_i$ corresponds to a feasible matching of size at least $(1-O(\eta k))|S_i|$. We now show that truncating $\tilde E^\w_i$ to $E^\w_i$ (see~\eqref{eq:ewi}) does not lead to a significant loss in matching size. Recall that by~\eqref{eq:ti-def}
\begin{equation*}
\begin{split}
T_i&=\{y\in Y : ((y, \u_j)+U(\u_j)) \mod W\in [1/k, 1)\cdot W,\text{~for all~} j\in [1:i]\}\\
S_i&=\{x\in X : ((x, \u_j)+U(\u_j)) \mod W\in [1/k, 1)\cdot W,\text{~for all~} j=[1:i]\}\\
\end{split}
\end{equation*}

For every $y\in T_i$ one has by~\eqref{eq:ti-def} that $((y, \u_j)+U(\u_j)) \mod W\in [1/k, 1)\cdot W$ for all  $j=1,\ldots, i$, and hence for every $\lambda\in (0, W/|\w|]$ and every $j=1,\ldots, i$
$$
|(y+\lambda \w, \u_j)-(y, \u_j)|\leq \lambda (\w, \u_j)\leq \e \lambda |\w|\leq \e W
$$
by choice of the family $\F_1,\ldots, \F_k$. We thus get that $y+\lambda \w$ belongs to the set
$$
\hat T_i:=\{y\in Y : ((y, \u_j)+U(\u_j)) \mod W\in ([0, \e)\cup [1/k-\e, 1))\cdot W,\text{~for all~} j\in [1:i]\},
$$
i.e. the result of relaxing the constraints that define $T_i$ by a $\e$ in every direction. At the same time
\begin{equation*}
\begin{split}
\hat T_i\setminus T_i&\subseteq \bigcup_{j=1}^i \{y\in Y : ((y, \u_j)+U(\u_j)) \mod W\in ([0, \e)\cup [1/k-\e, 1/k))\cdot W\}\\
&\subseteq \bigcup_{j=1}^i \{y\in Y : ((y, \u_j)+U(\u_j)) \mod W\in ([0, \e^{1/2})\cup [1/k-\e^{1/2}, 1/k))\cdot W\}
\end{split}
\end{equation*}
and thus 
\begin{equation*}
\begin{split}
|\hat T_i\setminus T_i|&\leq \sum_{j=1}^i |\{y\in Y : ((y, \u_j)+U(\u_j)) \mod W\in ([0, \e^{1/2})\cup [1/k-\e^{1/2}, 1/k))\cdot W\}|\\
&\leq 2k(6(k/\theta)\cdot \e^{1/2}+4/m)|Y|\text{~~~~~~(by Lemma~\ref{lm:intersection-size})}
\end{split}
\end{equation*}
To obtain the last bound, we applied Lemma~\ref{lm:intersection-size} for each $j=1,\ldots, i$ with $\mathcal{U}=\{\u_j\}$, $\delta'=\e$ and $L=(k/\theta)\cdot \e^{-1/2}$.
Now recalling that by~\eqref{eq:ewi} one has $E^\w_i:=\tilde E^\w_i\cap (T_i\times S_i)$ and that $S_i$ is obtained by intersecting $T_i$ with $X$, 
we get that the edges pruned from $\tilde E^\w_i$ by restricting to $T_i\times S_i$ as above are incident on as set of vertices of size at most 
\begin{equation}
2|\hat T_i\setminus T_i|\leq 4k(6(k/\theta)\cdot \e^{1/2}+4/m)|Y|.
\end{equation}
Since every vertex received at most $1$ unit of fractional mass, the size of the matching supported by $E^\w_i$ is thus at least the size of the matching supported by $\tilde E^\w_i$ minus $2k\cdot (6(k/\theta)\cdot \e^{1/2}+4/m)|Y|=O(k^2 \e^{1/2}/\eta) |Y|=O(k^3\e^{1/2}/\theta)|S_i|$ giving the result. In the last transition we used the fact that $|S_i|=\Omega(|Y|/k)$ when $\e<c\cdot \theta^2/k^6$ for a sufficiently small constant $c>0$ (by Lemma~\ref{lm:bounds-on-set-sizes}), as well as the assumption that $\theta=\eta$.

\end{proof}

\subsubsection{Existence of a sparse directed cut}

 Define
\begin{equation}\duallabel{eq:z-def}
 Z:=\{y\in Y:((y, \u_j)+U(\u_j)) \mod W\in ([1/k-\e, 1/k)\cup[0, \e))\cdot W\text{~for some $j\in [1:k]$}\}.
 \end{equation}
 
 We prove 
 \begin{claim}\label{cl:small-cut}
For every integer $k\geq 2$, $\e\in (0, 1)$ such that $1/\e^{1/2}$ is an integer, if $W/w$ is an integer multiple of $k/(\e\cdot \theta)$, the following conditions hold for sufficiently large $m$.

For every $i=0,1,\ldots, k-1$ the subgraph $H^*$ induced by $(T_i\setminus (T_k\cup Z))\cup S_i^*$ only contains the edges of $E^{\u_{i+1}}_i$.  In addition, one has $|Z|\leq 2k^2(6\e^{1/2}/\theta+4/m)|Y|$.
 \end{claim}
 \begin{proof}
Recall that the sets $S_i^*$ are defined in~\eqref{eq:siu-def}. First note that if an edge $(c, d), c\in P, d\in Q$ belongs to $H^*$, then $c\in S_i^*$ and $d\in T^i$, so $(c, d)$ necessarily belongs to some graph $H^\w_i$, where $\w\in \F_{i+1}$. Then  we have by \dualeqref{eq:edges} that 
 $$
d-c=\lambda\cdot \w, \text{~where~} |\lambda|\leq W/|\w|.
 $$
 
Thus, we have for all $j=1,\ldots, k$ using the orthogonality condition \dualeqref{eq:orthogonal}
\begin{equation}\duallabel{eq:orth}
\left|(c-d, \u_j)\right|\leq \frac{W}{|\w|}\left|(\w, \u_j)\right|\leq \e W.
\end{equation}
Now recall that $c\in S_i^*$ by assumption, so by \dualeqref{eq:ti-def} and \dualeqref{eq:siu-def}
\begin{equation*}
(c, \u_j)+U(\u_j) \mod W\in [1/k, 1)\cdot W,\forall j=1,\ldots, k.
\end{equation*}
Thus, by \dualeqref{eq:orth} one has
\begin{equation*}
(d, \u_j)+U(\u_j) \mod W\in ([1/k-\e, 1]\cup [0, \e))\cdot W,\text{~~for all~} j=1,\ldots, k,
\end{equation*}
i.e. $d\in Z\cup T^k$. Thus, the subgraph $H^*$ contains only edges of $E^{\u_{i+1}}_i$ for every $i=0,\ldots, k-1$, as required.

It remains to bound the size of $Z$. We first note that 
\begin{equation}\label{eq:z-union-bound}
\begin{split}
 Z&=\{y\in Y:((y, \u_j)+U(\u_j)) \mod W\in ([1/k-\e, 1/k)\cup[0, \e))\cdot W\text{~for some $j\in [1:k]$}\}\\
 &\subseteq \bigcup_{j=1}^k \{y\in Y:((y, \u_j)+U(\u_j)) \mod W\in ([1/k-\e, 1/k)\cup[0, \e))\cdot W\}\\
  &\subseteq \bigcup_{j=1}^k \{y\in Y:((y, \u_j)+U(\u_j)) \mod W\in ([1/k-\e^{1/2}, 1/k)\cup[0, \e^{1/2}))\cdot W\}.
\end{split}
\end{equation}

For each $j=1,\ldots, k$ we now use Lemma~\ref{lm:intersection-size} with $\mathcal{U}=\{\u_j\}$, $a_{\u_j}=\frac1{k}-\e, b_{\u_j}=\frac1{k}$, and then again with $a_{\u_j}=0, b_{\u_j}=\e$. In both cases we set $L=(k/\theta)\cdot 1/\e^{1/2}$. We thus get 
\begin{equation*}
\begin{split}
&|\{y\in Y:((y, \u_j)+U(\u_j)) \mod W\in ([1/k-\e^{1/2}, 1/k)\cup[0, \e^{1/2}))\cdot W\}|\\
&\leq 2(6L\e +4/m)\cdot |Y|\\
&\leq 2(6 k \e^{1/2}/\theta +4/m)\cdot |Y|.
\end{split}
\end{equation*}

Using this together with ~\eqref{eq:z-union-bound} yields $|Z|\leq 2k^2(6\e^{1/2}+4/m)|Y|$, as required.

\end{proof}

\subsection{Distribution over inputs}

\newcommand{\dtree}
{
\begin{center}
\tikzstyle{vertex}=[circle,draw, fill=none, minimum size=15pt,inner sep=1pt]
\tikzstyle{svertex}=[circle,fill=none, draw, minimum size=15pt,inner sep=2pt]
\tikzstyle{evertex}=[circle,draw=none, minimum size=25pt,inner sep=1pt]
\tikzstyle{edge} = [draw,-, color=red!100, very  thick]
\tikzstyle{bedge} = [draw,-, color=green!100, very  thick]
\begin{tikzpicture}[scale=0.4, auto,swap]

    \draw (0, 9)  node {$d$-ary tree $\mathcal{T}$};
    
    \node[vertex](u) at (0, 7) {};
    \node[svertex](u) at (0, 7) {$u_0$};    
    
    \node[vertex](v1) at (-6, 4) {};    
    \node[svertex](v1) at (-6, 4) {$u_1$};
    \node[vertex](v2) at (0, 4) {};
    \node[vertex](v3) at (6, 4) {};

    \node[vertex](v11) at (-4, 1) {};
    \node[svertex](v11) at (-4, 1) {$u_2$};        
    \node[vertex](v12) at (-6, 1) {};
    \node[vertex](v13) at (-8, 1) {$w$};
    
    \node[evertex](h) at (-10, 4) {$H^{w}_2$};    
    \node[evertex](hp) at (-6.7, 2) {};        
    \path[draw, line width=1pt, ->] (h) -- (hp);

    \node[vertex](v21) at (-2, 1) {};
    \node[vertex](v22) at (0, 1) {};
    \node[vertex](v23) at (+2, 1) {};
    
    \node[vertex](v31) at (4, 1) {};
    \node[vertex](v32) at (6, 1) {};
    \node[vertex](v33) at (8, 1) {};

    \path[draw, line width=4pt, -] (u) -- (v1);    
    \path[draw, line width=2pt, -, dashed] (u) -- (v2);    
    \path[draw, line width=2pt, -, dashed] (u) -- (v3);    

    \path[draw, line width=4pt, -] (v1) -- (v11);    
    \path[draw, line width=2pt, -, dashed] (v1) -- (v12);    
    \path[draw, line width=2pt, -, dashed] (v1) -- (v13);    
    
    \path[draw, line width=1pt, -] (v2) -- (v21);    
    \path[draw, line width=1pt, -] (v2) -- (v22);    
    \path[draw, line width=1pt, -] (v2) -- (v23);    
    
    \path[draw, line width=1pt, -] (v3) -- (v31);    
    \path[draw, line width=1pt, -] (v3) -- (v32);    
    \path[draw, line width=1pt, -] (v3) -- (v33);

\end{tikzpicture}    
\end{center}
 }

\newcommand{\packing}[1]
{
     \draw[fill=gray!50,opacity=1] (23,7) ellipse (7 and 1);
     \draw(18, 9) node {$T^{u_0}$};     
     
     \draw[fill=white,opacity=1] (24.5,7) ellipse (5 and 1);     
     \draw(22, 9) node {$T^{u_1}$};          

     \draw[fill=none,opacity=1] (20,1) ellipse (3 and 1);          
     \ifthenelse{\lengthtest{#1 pt < 4 pt}} {     \draw(18, -1) node {$S_0$}; }{};
     
     \draw[fill=gray!50,opacity=1] (21,1) ellipse (1.8 and 0.8); 
     
     \draw[fill=none,opacity=1] (22,1) ellipse (0.8 and 0.5);      
     
    \node[evertex](h) at (15, 4) {$H^{u_1}_0$};    
    \node[evertex](hp) at (20.5, 3) {};        
    \path[draw, line width=1pt, ->] (h) -- (hp);

     \draw(21, -1) node {$S_0^{u_1}$}; 
     \draw(22.5, -1) node {$S_0^{u_2}$};      
     
     \draw[color=black!50, line width=1 pt] (18, 1) -- (20, 7);
%     \draw[color=black!50, line width=1 pt] (18, 1) -- (22, 7);     
     \draw[color=black!50, line width=1 pt] (18, 1) -- (17, 7);          

     \draw[color=black!100, line width=3 pt] (20, 1) -- (18, 7);
%     \draw[color=black!50, line width=1 pt] (20, 1) -- (21, 7);     
     \draw[color=black!50, line width=1 pt] (20, 1) -- (25, 7);          

     \draw[color=black!100, line width=3 pt] (21, 1) -- (19, 7);
%     \draw[color=black!50, line width=1 pt] (21, 1) -- (21, 7);     
     \draw[color=black!50, line width=1 pt] (21, 1) -- (25, 7);

     \draw[fill=none,opacity=1] (26,1) ellipse (2 and 1);              
\ifthenelse{\lengthtest{#1 pt < 4 pt}} {    \draw(25, -1) node {$S_1$};                  }{};

     \draw[fill=none,opacity=1] (26.8,1) ellipse (1.1 and 0.7);              
     \draw(27, -1) node {$S_1^{u_2}$};

     \draw[color=black!50, line width=1 pt] (25, 1) -- (21, 7);
%     \draw[color=black!50, line width=1 pt] (25, 1) -- (24, 7);     
     \draw[color=black!50, line width=1 pt] (25, 1) -- (26, 7);          

     \draw[color=black!50, line width=1 pt] (27, 1) -- (24, 7);
     \draw[color=black!50, line width=1 pt] (27, 1) -- (22, 7);

\ifthenelse {\lengthtest{#1 pt > 2 pt}}      
{
     \draw[fill=none,opacity=1] (26,7) ellipse (3 and 0.9);          
     \draw(26, 9) node {$T^{u_2}$};
     
     \draw[fill=none,opacity=1] (32, 1) ellipse (3 and 0.9);             
      \draw(32, -1) node {$S_2$}; 
    
     \draw[color=black!50, line width=1 pt] (30, 1) -- (24, 7);
     \draw[color=black!50, line width=1 pt] (32, 1) -- (26, 7);     
     \draw[color=black!50, line width=1 pt] (34, 1) -- (28, 7);

}
{};

\ifthenelse {\lengthtest{#1 pt > 3 pt}}      
{
     \draw[fill=none,opacity=1] (21,1) ellipse (2 and 0.9);      
     \draw[fill=none,opacity=1] (21.5,1) ellipse (1.3 and 0.8);      
          
     \draw[fill=none,opacity=1] (26.5,1) ellipse (1.3 and 0.9);

     \draw(18, -1) node {$S_0^{u_0}$};     
     \draw(20, -1) node {$S_0^{u_1}$};                    
     \draw(22, -1) node {$S_0^{u_2}$};

     \draw(25, -1) node {$S_1^{u_1}$};     
     \draw(27, -1) node {$S_1^{u_2}$};                    
     }{};
}

\newcommand{\dpacking}[1]
{
\begin{center}
\tikzstyle{vertex}=[circle,fill=blue!50, minimum size=15pt,inner sep=1pt]
\tikzstyle{svertex}=[circle,fill=red!100, minimum size=15pt,inner sep=1pt]
\tikzstyle{evertex}=[circle,draw=none, minimum size=25pt,inner sep=1pt]
\tikzstyle{edge} = [draw,-, color=red!100, very  thick]
\tikzstyle{bedge} = [draw,-, color=green!100, very  thick]
\begin{tikzpicture}[scale=0.4, auto,swap]

\packing{#1}
     
\end{tikzpicture}
\end{center}
}

%%%%%%%%%%%%%%%%%%

\newcommand{\dtreeT}
{
\tikzstyle{vertex}=[circle,draw, fill=none, minimum size=15pt,inner sep=1pt]
\tikzstyle{svertex}=[circle,fill=none, draw, minimum size=15pt,inner sep=2pt]
\tikzstyle{evertex}=[circle,draw=none, minimum size=25pt,inner sep=1pt]
\tikzstyle{edge} = [draw,-, color=red!100, very  thick]
\tikzstyle{bedge} = [draw,-, color=green!100, very  thick]
\begin{tikzpicture}[scale=0.35, auto,swap]

    \draw (0, 9)  node {$d$-ary tree $\mathcal{T}$};
    
    \node[vertex](u) at (0, 7) {};
    \node[svertex](u) at (0, 7) {$u_0$};    
    
    \node[vertex](v1) at (-6, 4) {};    
    \node[svertex](v1) at (-6, 4) {$u_1$};
    \node[vertex](v2) at (0, 4) {};
    \node[vertex](v3) at (6, 4) {};

    \node[vertex](v11) at (-4, 1) {};
    \node[svertex](v11) at (-4, 1) {$u_2$};        
    \node[vertex](v12) at (-6, 1) {};
    \node[vertex](v13) at (-8, 1) {$w$};
    
    \node[evertex](h) at (-10, 4) {$H^{w}_2$};    
    \node[evertex](hp) at (-6.7, 2) {};        
    \path[draw, line width=1pt, ->] (h) -- (hp);

    \node[vertex](v21) at (-2, 1) {};
    \node[vertex](v22) at (0, 1) {};
    \node[vertex](v23) at (+2, 1) {};
    
    \node[vertex](v31) at (4, 1) {};
    \node[vertex](v32) at (6, 1) {};
    \node[vertex](v33) at (8, 1) {};

    \path[draw, line width=4pt, -] (u) -- (v1);    
    \path[draw, line width=2pt, -, dashed] (u) -- (v2);    
    \path[draw, line width=2pt, -, dashed] (u) -- (v3);    

    \path[draw, line width=4pt, -] (v1) -- (v11);    
    \path[draw, line width=2pt, -, dashed] (v1) -- (v12);    
    \path[draw, line width=2pt, -, dashed] (v1) -- (v13);    
    
    \path[draw, line width=1pt, -] (v2) -- (v21);    
    \path[draw, line width=1pt, -] (v2) -- (v22);    
    \path[draw, line width=1pt, -] (v2) -- (v23);    
    
    \path[draw, line width=1pt, -] (v3) -- (v31);    
    \path[draw, line width=1pt, -] (v3) -- (v32);    
    \path[draw, line width=1pt, -] (v3) -- (v33);

\end{tikzpicture}    
 }

\newcommand{\dpackingT}[1]
{
\tikzstyle{vertex}=[circle,fill=blue!50, minimum size=15pt,inner sep=1pt]
\tikzstyle{svertex}=[circle,fill=red!100, minimum size=15pt,inner sep=1pt]
\tikzstyle{evertex}=[circle,draw=none, minimum size=25pt,inner sep=1pt]
\tikzstyle{edge} = [draw,-, color=red!100, very  thick]
\tikzstyle{bedge} = [draw,-, color=green!100, very  thick]
\begin{tikzpicture}[scale=0.35, auto,swap]

\packing{#1}
     
\end{tikzpicture}
}

\if 0
\begin{figure}
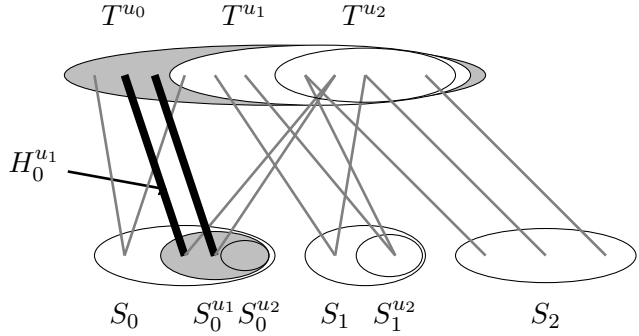

\dpacking{3}
\caption{Subgraphs $(T^{u_i}, S_i)$ that arrive in the stream. The edges of induced near-regular subgraph $H^{u_1}_0$  induced by $(T^{u_0}\setminus T^{u_1})\cup S_{0}^{u_1}$ are shown in bold.}\duallabel{fig:2}
\end{figure}
\fi

We now formally define our hard input distribution. The input graph $G'$ is generated as follows. First sample $X\subseteq \{0, 1\}^m$ as in Definition~\ref{def:xy}, then for every $\w \in \bigcup_{j=1}^k \F_j$ sample the shift $U(\w)$ independently as per~\eqref{eq:shifts-range}. Finally, sample $\u_s\in \F_s, s=1,\ldots, k$ independently and uniformly at random, and let $G:=G(\u_1,\ldots, \u_k)$ denote the host graph as constructed in Section~\ref{sec:host-graph-construction}. For every $\v \in \F_{i+1}$ and $y\in T_i\setminus T_i^\v$ let $b_y^\v\in \{0, 1\}$ denote a Bernoulli random variable with expectation $1-\xi$ for a small $\xi\in (0, 1)$ that we will set later. The variables $b_y^\v$ are independent conditioned on 
$$
\sum_{y\in T_i\setminus T_i^\v} b_y^\v=\lceil (1-\xi) |T_i\setminus T_i^\v|\rceil
$$ for every $\v\in \F_{i+1}$. In other words, $b^\v$ encodes a uniformly random subset of $T_i\setminus T_i^\v$ of size $\lceil (1-\xi)T_i\setminus T_i^\v\rceil$. For every $i=0,\ldots, k-1$ let $B_i:=\{b_y^\v\}_{\v \in \F_{i+1}, y\in T_i\setminus T_i^\v}$. 
\begin{definition}[Subsampling of the host graph]\label{def:g-prime}
For $i=0,\ldots, k-1$ let the graph $G'(\u_{1:i}; B_{0:i})$ be formed by including, for every $\v\in \F_{i+1}$ and $y\in T_i\setminus T_i^\v$ all edges incident on $y$ in $H_i^\v$ if $b_y^\v=1$ and none of these edges otherwise. For $i=k$, let $G'(\u_{1:k}; B_{0:k-1})$ contain all edges incident on $S_k$. Let 
$$
G':=\left(\bigcup_{i=0}^{k-1} G'(\u_{1:i}; B_{0:i})\right)\cup G'(\u_{1:k}; B_{0:k-1}).
$$
\end{definition} 
The stream  consists of $k+1$ {\em phases}: for each $i=0,\ldots, k$  the vertices and edges of $G'(\u_{1:i}, B_{0:i})$  incident on $S_i$ arrive in phase $i$ in an arbitrary order.

We start with the following  claim
\begin{claimp}\duallabel{claim:large-matching}
For every $\xi\in (0, 1)$, every integer $k\geq 2$, sufficiently small $\eta\in (0, 1)$ such that $1/\eta$ is an integer,  $\e \in (0, c\cdot \eta^2/k^6)$ for a sufficiently small constant $c>0$ such that $1/\e^{1/2}$ is an integer, if $\theta=\eta$ (see~\dualeqref{eq:shifts-range})  and $W/w$ is an integer multiple of $k/(\e\cdot \theta)$, the following conditions hold for sufficiently large $m$.

There exists an event $\E_{large-matching}$ that occurs with probability at least $97/100$ over the choice of random shifts $U(\v), \v\in \bigcup_{j=1}^k \F_j$ and choice of $\bigcup_{i=0}^{k-1} \{b_y^\v\}_{\v\in \F_{i+1}, y\in T_i\setminus T_i^\v}$,  the graph $G'$ contains a matching of size at least $(1-O(\xi+k^4\e^{1/2}/\eta))|Y|$.
\end{claimp}
\begin{proof}
By Claim~\ref{cl:large-matching} there exists an event $\E_{balanced-degrees}$ that depends only on the choice of $X$ and the shifts $U(\w), \w\in \bigcup_{j=1}^k \F_j$ such that conditioned on $\E_{balanced-degrees}$ and event $\E$ from Lemma~\ref{lm:bounds-on-set-sizes} by Claim~\dualref{cl:large-matching} for every $i=0,\ldots, k-1$ and $\u_1,\ldots,\u_s\in \F_s, s=1,\ldots, i$ and $\w\in \F_{i+1}$ there exists a matching $M_i$ of $1-O(k^3\e^{1/2}/\eta)$ fraction of $S_{i}$ to $T_i\setminus T_i^{\u_{i+1}}$.  Furthermore, the set $S_k$ can be perfectly matched to $T_k$ by definition. Let $M$ be a union of these matchings. Note that conditioned on $\E_{balanced-degrees}$ the matching $M$ satisfies
\begin{equation*}
\begin{split}
|M|&\geq (1-O(\xi+k^3\e^{1/2}/\eta))|S|+|T_k|\\
&\geq (1-O(\xi+k^3\e^{1/2}/\eta))\sum_{i=0}^{k-1} |S_i|+|T_k|\\
&\geq (1-O(\xi+k^3\e^{1/2}/\eta))\sum_{i=0}^{k-1} (1-1/k)^i|Y|/k+|T_k|-O(k^4\e/\theta)|Y|\\
&\geq (1-O(\xi+k^3\e^{1/2}/\eta))(1-(1-1/k)^k)|Y|+(1-1/k)^k|Y|-O(k^4\e/\theta)\\
&\geq (1-O(\xi+k^4\e^{1/2}/\eta))|Y|,
\end{split}
\end{equation*}
where we used Lemma~\ref{lm:bounds-on-set-sizes}, {\bf (2)}, in the third transition and Lemma~\ref{lm:bounds-on-set-sizes}, {\bf (1)}, in the forth transition.

 Now recall that $G'$ contains every edge of $M$ with probability at least $1-\xi$, and these events are negatively associated for different edges, since for every $(y, x)\in M\cap H^{\u_{i+1}}$ one has
$(y, x)\in G'$ if and only if $b_y^{\u_{i+1}}=1$, and $b_y^{\u_{i+1}}$ are negatively associated for different $y$ by construction. We thus get by an application of Chernoff bounds that
$$
\prob[|M\cap G'|<(1-2\xi)|M|]<e^{-\Omega(\xi^2 |M|)}<e^{-\Omega(\xi^2 m^{4m})}.
$$
We now define the event $\E_{large-matching}$ to be the intersection of $\E_{balanced-degrees}$ with the success events for $i=0,\ldots, k-1$ above, getting that $\prob[\E_{large-matching}]\geq 1-97/100$ by a union bound.

%In particular, this implies that the graph $G'$ contains a matching of size at least $(1-O(\xi+\eta k))|S|$ by applying the above to $\w=\u_{i+1}$ for every $i=0,\ldots, k-1$.

%Applying this claim for every $i=0,\ldots, k-1$ with $\w=\u_i$, we get that there exists a matching of at least a $1-O(\eta k)$ fraction of $\bigcup_{i=0}^{k-1} S_i$ to $\bigcup_{i=0}^{k-1} T_i\setminus T_i^{\u_{i+1}}=T_0\setminus T_k$. Adding the matching from $S_k$ to $T_k$ shows that there exists a matching $M$ in $G(\u_1,\ldots, \u_k)$ of size $(1-O(\eta k))|Y|$. 

%The claim for every $i=0,\ldots, k-1$ follows similarly. We now define the event $\E_{large-matching}$ to be the intersection of $\E_{balanced-degrees}$ with the success event above, getting that $\prob[\E_{large-matching}]\geq 1-98/100$ by a union bound.
\end{proof}

\subsection{Bounding performance of a small space algorithm}
\newcommand{\G}{\mathcal{G}}

\newcommand{\pe}{P_e}
\newcommand{\Xc}{{\mathcal X}}

By Yao's minimax principle it is sufficient to upper bound the performance of a deterministic small space algorithm 
that succeeds with probability at least $1/2$.
To do that, we bound the size of the matching that a small space algorithm can output at the end of the stream. Let $M_{ALG}$ denote the matching that the algorithm outputs.
We first upper bound the approximation ratio that the algorithm obtains in terms of the number of edges in $E(H^{\u_{i+1}}_i)\cap M_{ALG}$, for $i=0,1,\ldots, k-1$.

\begin{lemma}\duallabel{lm:match-upper}
For every integer $k\geq 2$, $\e \in (0, 1)$ such that $1/\e^{1/2}$ is an integer, if $W/w$ is an integer multiple of $k/(\e\cdot \theta)$ (see~\eqref{eq:shifts-range}), the following conditions hold for sufficiently large $m$.

If the graph $G'$ is generated as per Definition~\dualref{def:g-prime}, and $M_{ALG}$ is any matching in $G'$, then $|M_{ALG}|\leq (1-1/k)^k |Y|+\sum_{i=0}^{k-1} |E(H^{\u_{i+1}}_i)\cap M_{ALG}|+O(k^3 \e^{1/2}/\theta) \cdot |Y|$.
\end{lemma}
\begin{proof} Consider the cut $(A, B)$, where $A=\left(T\setminus (T_k\cup Z)\right)\cup \bigcup_{i=0}^{k-1} (S_i\setminus S_i^*)$ 
and $B=T_k\cup S_k\cup \bigcup_{i=0}^{k-1} S_i^*\cup Z$. Recall that the sets $T_i, S_i$ are defined in~\eqref{eq:ti-def}, $S_i^j$ is defined in~\eqref{eq:siu-def} and $Z$ is defined in~\eqref{eq:z-def}.

By the maxflow/mincut theorem, the size of the matching output by the algorithm is bounded by 
$|A\cap S|+|B\cap T|+|((A\cap T)\times (B\cap S))\cap M_{ALG}|$.

By Claim~\ref{cl:small-cut} the subgraph $M_{ALG}\cap (T_i\times S_i)$ induced by $(T_i\setminus (T_k\cup Z))\cup S_i^*$ only contains the edges of $H^{\u_{i+1}}_i$ for every $i=0,\ldots, k-1$.  Thus, 
$$
|((A\cap T)\times (B\cap S))\cap M_{ALG}|\leq  \sum_{i=0}^{k-1} |E(H^{\u_{i+1}}_i)\cap M_{ALG}|
$$  and we get
\begin{equation}\label{eq:e-start-ub}
\begin{split}
|M_{ALG}|&\leq \left|\bigcup_{i=0}^{k-1} (S_i\setminus S_i^*)\right|+|T_k|+|Z|+\sum_{i=0}^{k-1} |E(H^{\u_{i+1}}_i)\cap M_{ALG}|\\
&=\left|\bigcup_{i=0}^{k-1} S_i\right|+|T_k|-\left|\bigcup_{i=0}^{k-1} S_i^*\right|+|Z|+\sum_{i=0}^{k-1} |E(H^{\u_{i+1}}_i)\cap M_{ALG}|\\
\end{split}
\end{equation}

Furthermore, again by Claim~\ref{cl:small-cut} one has $|Z|\leq 2k^2(6\e^{1/2}/\theta+4/m)|Y|$. By Lemma~\ref{lm:bounds-on-set-sizes}, {\bf (1)} one has 
$|T_i|=(1-1/k)^i |Y|+\Delta_i$, $|\Delta_i|=O(k^3 \e/\theta)\cdot |Y|$ for every $i\in \{0, 1, 2, \ldots, k\}$.
By Lemma~\ref{lm:bounds-on-set-sizes}, {\bf (2)} one has  $|S_i|=\frac1{k}((1-1/k)^i |Y|+\Delta_i$, $|\Delta_i|= O(k^3 \e/\theta)\cdot |Y|$ for every $i\in \{0, 1, 2,\ldots, k\}$ and by Lemma~\ref{lm:bounds-on-set-sizes}, {\bf (3)} one has $|S_i^*|=\frac1{k}(1-1/k)^k |Y|+\Delta_i$, $|\Delta_i|=O(k^3 \e/\theta)\cdot |Y|$ for every $i\in \{0, 1, 2,\ldots, k\}$. Substituting these bounds into~\eqref{eq:e-start-ub}, we get
\begin{equation*}
\begin{split}
|M_{ALG}|&\leq \sum_{i=0}^{k-1} \frac1{k}(1-1/k)^i |Y|-\sum_{i=0}^{k-1}\frac1{k} (1-1/k)^k |Y|\\
&+(1-1/k)^k\cdot |Y|+\sum_{i=0}^{k-1} |E(H^{\u_{i+1}}_i)\cap M_{ALG}|+O(k^3 \e^{1/2}/\theta )\cdot |Y|.
\end{split}
\end{equation*}
We now use the fact that $\sum_{i=0}^{k-1} \frac1{k}(1-1/k)^i |Y|=(1-(1-1/k)^k)|Y|$ in the upper bound above to get
\begin{equation*}
\begin{split}
|M_{ALG}|&\leq (1-(1-1/k)^k)|Y|+\sum_{i=0}^{k-1} |E(H^{\u_{i+1}}_i)\cap M_{ALG}|+O(k^3 \e^{1/2}/\theta )\cdot |Y|,\\
\end{split}
\end{equation*}
as required. 
\end{proof}

\renewcommand{\b}{\mathbf{b}}
\renewcommand{\B}{\mathbf{B}}

We will use
\begin{lemma} \emph{(Data Processing Inequality)} \label{thm:dpi}
    For any random variables $(X,Y,Z)$ such that $X \to Y \to Z$ forms a Markov chain, we have $I(X;Z) \le I(X;Y)$.
\end{lemma}

\begin{lemma}\label{lm:space-lb}
For every integer $k\geq 2$, if $\e<\eta$, the following conditions hold for sufficiently large $m$. 

Let  $M_{ALG}$ denote the subset of edges of $G'$ output by a space $s$ streaming algorithm after a single pass over the edges of $G'$ presented in the order defined above. If 
$$
\left| M_{ALG}\cap  \bigcup_{i=0}^{k-1} E(H^{\u_{i+1}}_i)\right|>c n
$$ 
with probability more than $1/2$ over the randomness used to generate the graph $G'$, then $s=\Omega_{c, k}(n d)$.
\end{lemma}
\begin{proof}
Recall that we define $M_{ALG}$ to be the empty set if the matching output by ALG contains edges that are not in $G$, and that
\begin{equation}\label{eq:key-edges-reported}
\begin{split}
\prob\left[\left| M_{ALG}\cap  \bigcup_{i=0}^{k-1} E(H^{\u_{i+1}}_i)\right|>c n\right]&>1/2,
\end{split}
\end{equation}
where the probability is over the choice of $\u_{1:k}$ (the vectors defining the host graph $G$) and $\B_{0:k-1}$ (the random variables used to subsample the host graph $G$ to generate $G'$), subsampling $X$ and the shifts $\{U(\w)\}$. Define, for $\v \in \F_{i+1}$, 
$$
M_{ALG}^\v=M_{ALG}\cap E(H_i^\v).
$$
Note that $M_{ALG}^\v$ depends on $i$, since $\v$ uniquely determines $i$ (since it belongs to $\F_{i+1}$ and no other $\F_j, j=0,1,\ldots, k-1$). Also define
$$
\E_{many-edges}(i):=\left\{|M_{ALG}^\v|>cn/(6k)\right\}.
$$ 

We have, using~\eqref{eq:key-edges-reported}, that there exists an index $0\leq i\leq k-1$ such that 
$$
\prob_{\u_{1:k}, \B_{0:k-1}, \mathbf{X}, \{\mathbf{U}(\w)\}}[\E_{many-edges}(i)]\geq 1/(6k).
$$ Fix one such index $i$ in what follows. We have by an averaging argument applied to~\eqref{eq:key-edges-reported} (using Claim~\ref{cl:induced} to conclude that the edge sets of $H_i^\w$ are disjoint for distinct $\w\in \F_{i+1}$) that 
there exists a fixing $F=(u_{1:i}, B_{0:i-1}, X, \{U(\w)\})$ of $\u_{1:i}, \B_{0:i-1}, \mathbf{X}, \{\mathbf{U}(\w)\}$ such that 
$$
\prob_{\u_{i+1:k-1}, \B_{i:k-1}}[\E_{many-edges}(t)| F]\geq 1/(6k).
$$
For every $\v\in \F_{i+1}$ define 
$$
\E_{many-edges}(i, \v):=\E_{many-edges}(i)\wedge \{\u_{i+1}=\v\}.
$$
We have
\begin{equation}\label{eq:many-edges-v}
\begin{split}
1/(6k)\leq &\prob_{\u_{i+1:k-1}, \B_{i:k-1}}[\E_{many-edges}(i)| F]\leq \expect_{\u_{i+1}}[\prob_{\u_{i+2:k}, \B_{i:k-1}}[\E_{many-edges}(i)| F]]\\
&=\frac1{|\F_{i+1}|}\sum_{\v\in \F_{i+1}}\prob_{\u_{i+2:k}, \B_{i:k-1}}[\E_{many-edges}(i, \v)| F].\\
\end{split}
\end{equation}
%which means that for at least $1/(12k)$ fraction of  $\v \in \F_{i+1}$ one has $\prob[\E_{many-edges}(t, \v)| F]\geq 1/(12k)$. Indeed, otherwise the total success probability would be less than $(1/(12k))\cdot 1+1\cdot 1/(12k)<1/(6k)$, contradicting the equation above. Denote this subset of indices by $\tilde \F_t\subseteq \F_t$.

Let $\Pi$ denote the state of the algorithm after processing 
$$
\left(\bigcup_{j=0}^{i-1} G'(u_{1:j}; B_{0:j})\right)\cup G'(u_{1:i}; B_{0:i-1}, \B_i),
$$
where $G'(u_{1:j}; B_{0:j})$ is given by Definition~\ref{def:g-prime}. We now lower bound $I(\Pi; \B_i)$ (which then gives a lower bound on the entropy of $\Pi$, and therefore on the space $s$). First note that 
$$
\B_i=\{b_i^\v\}_{\v \in \F_{i+1}}\to \Pi \to M_{ALG}
$$ 
forms  a Markov chain, and thus by the data processing inequality (Lemma~\ref{thm:dpi})  we have
\begin{equation}\label{eq:sigma-t-vs-sigma-i}  
I(\Pi; \B_i) \ge I(M_{ALG}; \B_i).
\end{equation}
It thus suffices to lower bound $I(M_{ALG}; \B_i)=H(\B_i)-H(\B_i| M_{ALG})$. We upper bound the second term. First let $E$ denote the indicator random variable of $\E_{many-edges}(i)$ conditioned on $\{\u_{1:i}=u_{1:i}\text{~and~}\B_{0:i-1}=B_{0:i-1}\}$ (note that by choice of the index $i$ we have $\expect_{\u_{i+1:k-1}, \B_{i:k-1}}[E]\geq 1/(6k)$.  Further, for every $\v\in \F_{i+1}$ let $E^\v$  denote the indicator random variable of $\E_{many-edges}(i, \v)$ conditioned on $\{\u_{1:i}=u_{1:i}, \B_{0:i-1}=B_{0:i-1}\}$.  Note that by~\eqref{eq:many-edges-v} we have 
\begin{equation}\duallabel{eq:prob-lb}
\expect_{\v\sim UNIF(\F_{i+1})} [E^\v]\geq 1/(6k). 
\end{equation}

\begin{equation}\label{eq:ent-bi-ub}
\begin{split}
H(\B_i| \Pi)&\leq \sum_{\v \in \F_{i+1}} H(b^\v| \Pi)\text{~~~~~~~~~~~~~~~~~~~~~~~(by subadditivity of entropy)}\\
&\leq \sum_{\v \in \F_{i+1}} H(b^\v, E^\v| \Pi)\\
&= \sum_{\v \in \F_{i+1}} (H(E^\v)+H(b^\v| \Pi, E^\v))\\
&\leq \sum_{\v \in \F_{i+1}} (1+H(b^\v| \Pi, E^\v))\text{~~~~~~~(since $E^\v\in \bool$)}\\
&=d+\sum_{\v \in \F_{i+1}} H(b^\v| \Pi, E^\v=1)\cdot \prob[E^\v=1]+\sum_{\v\in \F_{i+1}}  H(b^\v| \Pi, E^\v=0)\cdot \prob[E^\v=0]\\
%&\leq d+\sum_{\v \in \F_{i+1}} H(b^\v| \Pi, E^\v=1)\cdot \prob[E^\v]+\sum_{\v\in \F_{i+1}}  H(b^\v)\cdot (1-\prob[E^\v])\\
\end{split}
\end{equation}

We now bound the terms on the rhs. First, since for every $\v$ one has that $b^\v$ has a fixed number of nonzeros in uniformly random positions by definition,
\begin{equation}
H(b^\v| \Pi, E^\v=0)\leq H(b^\v).
\end{equation}

\newcommand{\wh}[1]{\widehat{#1}}
\renewcommand{\i}{\mathbf{i}}

We now bound the second sum on the last line of~\eqref{eq:ent-bi-ub}. First recall that if $E^\v=1$, then $|M_{ALG}^\v|>cn/(6k)$, and $M_{ALG}^\v$ is a subset of the edges of $G'(u_{1:i-1}, \v; B_{1:i})$. Let $b:=B_i$. Recall that for every vertex $y\in T_i\setminus T_i^\v$ we include edges incident on it in $H^\v_i$ if $b_y^\v=1$ and do not otherwise.  We thus have, for every $y\in T_i\setminus T_i^\v$ such that $\delta(y)\cap M_{ALG}^\v \neq \emptyset$  that $b_y^\v=1$. Define for $\v \in \F_{i+1}$
$$
\gamma^\v:=\frac{|(T_i\setminus T_i^\v)\cap M_{ALG}^\v|}{|T_i\setminus T_i^\v|},
$$
and note that whenever $|M_{ALG}^\v|>cn/(6k)$, we get using Lemma~\ref{lm:bounds-on-set-sizes}, {\bf (1)} and {\bf (4)}, as well as the fact that $\e<1/10$ by our choice of parameters (since $\e\leq (c/k)^C$ for a sufficiently large constant $C\geq 1$), 
\begin{equation}\duallabel{eq:gamma-v-lb}
\gamma^\v=\frac{cn/(6k)}{|T_i\setminus T_i^\v|}\geq c/30.
\end{equation}
We may assume that $\gamma^\v\leq 1/3$ (remove some edges from $M^\v_{ALG}$ otherwise). We have
\begin{equation}
\begin{split}
H(b^\v| M_{ALG}^\v, E^\v=1)&\leq \log_2 {|T_i\setminus T_i^\v|-|M_{ALG}^\v| \choose \lceil (1-\xi)|T_i\setminus T_i^\v|\rceil-|M_{ALG}^\v|}\\
&\leq \log_2 {(1-\gamma^\v)|T_i\setminus T_i^\v| \choose  (1-\xi-\gamma^\v)|T_i\setminus T_i^\v|},
\end{split}
\end{equation}
where we used the assumption that $\gamma^\v\leq 1/3$ and the fact that $\xi<1/3$ by our setting of parameters. We thus get
$$
H(b^\v| M_{ALG}^\v, E^\v=1)\leq (1-\gamma^\v) |T_i\setminus T_i^\v| H_2\left(1-\frac{\xi}{1-\gamma^\v}\right),
$$
since
\begin{equation*}
\begin{split}
\log_2 { (1-\gamma^\v)|T_i\setminus T_i^\v| \choose (1-\xi-\gamma^\v) |T_i\setminus T_i^\v|}&=\log_2 { (1-\gamma^\v) |T_i\setminus T_i^\v| \choose (1-\frac{\xi}{1-\gamma^\v}) (1-\gamma^\v) |T_i\setminus T_i^\v|}\\
&\leq (1-\gamma^\v) |T_i\setminus T_i^\v| H_2\left(1-\frac{\xi}{1-\gamma^\v}\right),
\end{split}
\end{equation*}
where the last transition is by subadditivity of entropy. 
At this point we also note that 
\begin{equation*}
\begin{split}
(1-\gamma^\v) H_2\left(1-\frac{\xi}{1-\gamma^\v}\right)&=\xi\log_2 (1/\xi)+\xi \ln 2-\xi \log\frac1{1-\gamma^\v}+O(\xi^2)\\
&\leq H_2(1-\xi)-\xi \log\frac1{1-\gamma^\v}+O(\xi^2)\\
&\leq H_2(1-\xi)-\Omega(\xi\cdot c)+O(\xi^2)\text{~~~~~~~(by~\dualeqref{eq:gamma-v-lb} and Claim~\dualref{cl:monotonicity-entropy})}.
\end{split}
\end{equation*}
since $H_2(1-\xi)=\xi\log_2 (1/\xi)+\xi \ln 2+O(\xi^2)$ and $\xi$ is smaller than a constant.  Putting the above bounds together, and noting that by subadditivity of entropy
$$
H(b^\v)\leq |T_i\setminus T_i^\v| H_2\left(1-\xi\right)=|T_i\setminus T_i^\v| \cdot (\xi\log_2 (1/\xi)+\xi \ln 2+O(\xi^2)),
$$
we get, since $\xi$ is smaller than $c$ by a large constant factor by our choice of $\xi$, that
 \begin{equation*}
\begin{split}
H(b^\v| M_{ALG}, E^\v=1)&\leq H(b^\v)-\Omega(c\cdot \xi)\cdot |T_i\setminus T_i^\v|\\
&\leq H(b^\v)-\Omega(c\cdot \xi/k)\cdot |T|
\end{split}
\end{equation*}

Using this upper bound in ~\eqref{eq:ent-bi-ub}, we get
\begin{equation*}
\begin{split}
H(\B_i| \Pi)&\leq d+\sum_{\v \in \F_{i+1}} H(b^\v| \Pi, E^\v=1)\cdot \prob[E^\v]+\sum_{\v\in \F_{i+1}}  H(b^\v)\cdot (1-\prob[E^\v])\\
&\leq d+\sum_{\v \in \F_{i+1}} (H(b^\v)-\Omega(c\cdot \xi/k)\cdot |T|)\cdot \prob[E^\v]+\sum_{\v\in \F_{i+1}}  H(b^\v)\cdot (1-\prob[E^\v])\\
&\leq d+\sum_{\v \in \F_{i+1}} H(b^\v)-\Omega(c\cdot \xi/k)\cdot |T| \cdot \sum_{\v\in \F_{i+1}}  \prob[E^\v]\\
&=H(\B_i)-\Omega(c\cdot \xi/k)\cdot |T|\cdot  |\F_{i+1}| \cdot \expect_{\v\sim UNIF(\F_{i+1})} [E^\v]+d\\
&\leq H(\B_i) - \Omega(c\cdot \xi/k^2)\cdot |T| \cdot |\F_{i+1}| +d\text{~~~~~~~~~~~~~~~~~~~(by~\dualeqref{eq:prob-lb})}\\
&\leq H(\B_i)- \Omega(c\cdot \xi/k^2)\cdot |T| \cdot |\F_{i+1}|\\
&\leq H(\B_i)- \Omega_k(c)\cdot dn.\\
\end{split}
\end{equation*}
Using this bound in~\eqref{eq:sigma-t-vs-sigma-i}, we get $I(\B_i: \Pi)\geq \Omega_k(c)\cdot dn$, and therefore
$$
s\geq H(\Pi)\geq I(\B_i; \Pi) \ge \Omega_{c, k}(n d),
$$
as required.
\end{proof}

\begin{claim}\duallabel{cl:monotonicity-entropy}
For every $\xi>0$ the function $(1-\gamma)H_2(1-\frac{\xi}{1-\gamma})$ is decreasing in $\gamma$ for all $\gamma\in (0, 1-\xi)$.
\end{claim}
\begin{proof}
We have
\begin{equation*}
\begin{split}
(1-\gamma)H_2\left(1-\frac{\xi}{1-\gamma}\right)&=\frac1{\ln 2}\cdot(1-\gamma)\left[\left(1-\frac{\xi}{1-\gamma}\right)\ln \frac1{1-\frac{\xi}{1-\gamma}}+\frac{\xi}{1-\gamma}\ln \frac{1-\gamma}{\xi}\right]\\
&=\frac1{\ln 2}\cdot\left[(1-\gamma-\xi)\ln \frac{1-\gamma}{1-\gamma-\xi}+\xi\ln \frac{1-\gamma}{\xi}\right]\\
&=\frac1{\ln 2}\cdot\left[(1-\gamma-\xi)\ln \left(1+\frac{\xi}{1-\gamma-\xi}\right)+\xi\ln \frac{1-\gamma}{\xi}\right]\\
\end{split}
\end{equation*}

Since $\xi\ln \frac{1-\gamma}{\xi}$ is decreasing in $\gamma\in (0, 1)$, it suffices to show that $(1-\gamma-\xi)\ln \left(1+\frac{\xi}{1-\gamma-\xi}\right)$ is decreasing in $\gamma$ for $\gamma\in (0, 1-\xi)$.
Letting $x=1-\gamma-\xi$, it suffices to show that $x \ln (1+\frac{\xi}{x})$ is increasing in $x$ for $x\in (0, 1-\xi)$. Rescaling $x$ by $\xi$, it suffices to show that $x \ln (1+\frac1{x})$ is increasing in $x$ for all $x>0$.
The derivative with respect to $x$ is $\ln (1+\frac1{x})-\frac1{x+1}$, which approaches $0$ as $x\to \infty$. The derivative of this function is  $-\frac{1}{x(x+1)^2}$, which is negative for all $x>0$, and thus $\ln (1+\frac1{x})-\frac1{x+1}>0$ for all $x>0$.
\end{proof}

We can now give

\begin{proofof}{Theorem~\dualref{thm:main}}
Since ALG provides a better than $(1-1/e+c)$-approximation for some constant $c>0$ by assumption, there exists integer $k$  such that 
$1-(1-1/k)^k+c/2\leq 1-1/e+c$ (we assume that $2/c$ is an integer, which can be ensured by reducing $c$ by at most a factor of $2$). 

\paragraph{Setting parameters.} Let $G'$ be generated as per Definition~\ref{def:g-prime} with parameters selected as follows. First let $\xi=\eta=(c/k)^A$ for a sufficiently large integer $A>1$ (recall that $\xi$ is the rate at which we subsample edges of $G$ to obtain $G'$). Then let $\e=(\eta/k)^{2B}$ for a sufficiently large integer  $B>1$ (note that $1/\e^{1/2}$ is an integer). Finally let $m$ be an integer multiple of $1/\e$, let $w=\e m$ and let $W=w\cdot k/(\e\cdot \theta)$, where $\theta=\eta$.

We have by Claim~\ref{claim:large-matching} that the graph $G'$ (as per Definition~\ref{def:g-prime})  contains matching of size at least $(1-O(\xi+k^3\e/\eta))|S|$ with probability at least $97/100$. We also note that
\begin{equation*}
\begin{split}
|S|&=\sum_{i=0}^{k-1} |S_i|+|S|=\sum_{i=0}^{k-1}(1-1/k)^i |Y|/k+(1-1/k)^k|Y|\pm O(k^4 \e/\theta)|Y|\\
&=(1\pm O(k^4 \e/\theta))|Y|\\
&=(1\pm c/100)|Y|\\
\end{split}
\end{equation*}
by Lemma~\ref{lm:bounds-on-set-sizes}, {\bf (1)} and {\bf (2)}, since $O(k^4 \e/\theta)=O(k^4 \e/\eta)<c/100$ when $A$ and $B$ above are larger than an absolute constant (as we verify below in~\eqref{eq:283hg832hg832g}).  Thus, the algorithm must output a matching of size at least 
\begin{equation}\label{eq:msize-lb}
(1-(1-1/k)^k+c/2)(1-c/100)^2|Y| \geq (1-(1-1/k)^k+c/4)|Y|
\end{equation}
with probability at least $1/2$. The inequality above uses the fact that
\begin{equation}\label{eq:283hg832hg832g}
\begin{split}
O(\xi+ k^4 \e/\eta)&=O((c/k)^A+k^4 (\eta/k)^{2B}/\eta)\\
&=O((c/k)^A+(\eta/k)^{2B-4})\\
&=O((c/k)^A+(c/k)^{2B-4})\\
&<c/100
\end{split}
\end{equation}
as long as $A$ and $B$ are larger than an absolute constant.

Now let $M_{ALG}$ be the matching output by a single pass streaming algorithm ALG on the graph $G'$ presented in the order prescribed by our input distribution. For convenience we define $M_{ALG}$ to be the empty set if ALG outputs an edge that was not in $G'$. By Lemma~\ref{lm:match-upper} we have 
\begin{equation}\label{eq:9023823gfsdASD}
|M_{ALG}|\leq (1-1/k)^k |Y|+\sum_{i=0}^{k-1} |E(H^{\u_{i+1}}_i)\cap M_{ALG}|+O(k^3 \e^{1/2}/\theta )|Y|.
\end{equation}
Using this together with~\eqref{eq:9023823gfsdASD}, as well as the fact that $O(k^3 \e^{1/2}/\theta )<c/100$ as long as $A$ and $B$ above are larger than an absolute constant, we get
$$
\prob\left[\sum_{i=0}^{k-1} |E(H^{\u_{i+1}}_i)\cap M_{ALG}|>(c/8)\cdot n\right]> 1/2
$$
by assumption of the theorem. We  now have by Lemma~\ref{lm:space-lb} that the space complexity $s$ of the algorithm satisfies $s=\Omega_{c, k}(nd)$. Since $n=m^{4m}$ and $d=2^{\Omega(\e^2 m)}=2^{\Omega(m)}$ for any fixed $\e$ by Lemma~\ref{lem:code}, we get that $nd=n^{1+\Omega(1/\log\log n)}$, as required.

\end{proofof}

%!TEX root = ./doc-t.tex
\section{Multipass approximation for matchings}\duallabel{sec:multipass}

In this section we present our algorithm for approximating matchings in multiple passes in the vertex arrival setting, proving Theorem~\ref{thm:main-ubound}.

\subsection{The algorithm}
Let $G=(P, Q, E)$ denote a bipartite graph. We assume that vertices in $P$ arrive in the stream together with all their edges. At each step the algorithm maintains a fractional matching $\{f_e\}_{e\in E}$, where the capacity of each vertex in $Q$ is infinite and the capacity of each vertex $u\in P$ is equal to the number of times it has appeared in so far (i.e. always between $1$ and $k$). The capacity of an edge $e=(u, v), u\in P, v\in Q$ is equal to the capacity of $u$. For a vertex $u\in P$ we write $\delta(u)$ to denote the set of neighbors of $u$ in $G$. 

The fractional matching $f_e$ is initialized at zero, and upon arrival of a vertex $u\in P$ the algorithm continuously assigns a single unit of water to its least loaded neighbors. At the end of the $k$ passes we obtain a bona-fide matching by reducing the load of vertices on the $Q$ side that were assigned more than $k$ units of fractional mass down to $k$ units (simply reduce the load on neighboring edges). Scaling the resulting allocation by $1/k$ gives a feasible fractional matching, which can then be rounded to an integral matching using standard techniques in nearly linear time in the support size of the matching. The algorithm for processing a vertex $u\in P$ upon arrival is summarized in Algorithm~\ref{alg:main} below.

\begin{algorithm}[H]\duallabel{alg:main}
\caption{\textsc{ProcessVertex}($G$, $u$, $\delta(u)$)}
\begin{algorithmic}[1]
\STATE \textsc{WaterFilling}($G', u, \delta(u)$) {~~~~~~~~~~~~~~~~~~~~~~~~~~~~~$\rhd$ Assign one unit of water to least loaded neighbors}
\STATE \textsc{RemoveCycles}($G', f$). 
\end{algorithmic}
\end{algorithm}

 The function \textsc{WaterFilling}($G', u, \delta(u)$) increases the load of the least loaded neighbors of $u$ simultaneously (with other neighbors joining if the load reaches their level) until one unit of water in total is dispensed out of $u$. Here the support of the fractional matching $\{f_e\}_{e\in E}$ maintained by the algorithm is denoted by $G'$. The function \textsc{RemoveCycles}($G', f$) reroutes flow among cycles that could have emerged in the process, ensuring that the flow is supported on at most $|P|+|Q|-1$ edges.  

\paragraph{Efficient implementation.} First note that \textsc{WaterFilling}$(G', u ,\delta(u))$ can be implemented to run in time $O(|\delta(u)| \log n)$. Indeed, we need to find $\theta$ such that 
$$
\sum_{e=(u, v)\in \delta(u)} \max\{\theta-c_v, 0\}=1,
$$
where $c_v$ is the load of $v\in Q$ in the current fractional allocation. The function on the lhs is non-decreasing for all $\theta\geq \min_{e=(u, v)\in \delta(u)} c_v$, so the root can be found to within polynomial precision in $O(\log n)$ time using binary search.

Similarly, the function \textsc{RemoveCycles} can be implemented to run in nearly linear time at the expense of a loss of an $O(\log n)$ factor in space complexity. To achieve this we first buffer incoming vertices until the number of edges received is $\Theta(n)$ and only perform cycle removal after such a batch has been received. Let $f:E\to \mathbb{R}$ denote the allocation corresponding to one such batch. Write $f=\sum_{i=0}^{O(\log n)} 2^{-i} f_i$, where $f_i: E\to \{0, 1\}$ encode the sets of edges whose $i$-th bit in the allocation $f$ is set to $1$. Denote the corresponding edge sets by $E_i\subseteq E$, $i=0,1,\ldots, O(\log n)$.  Now for every $E_i$ run DFS to find cycles, and note that every time a cycle in $E_i$ is found, we zero out half of the edges on the cycle while rerouting flow in $f_i$.  Thus, the amount of work on $E_i$ is indeed linear in its size, resulting in a nearly linear runtime bound overall.

We now turn to analyzing the approximation ratio. We first give a sketch of the proof under additional assumptions on the graph $G$, and then proceed to give the relevant definitions and the complete argument.

\subsection{Analysis in a simple case (when $G$ has a perfect matching)} 

In this section we assume that $G=(P, Q, E)$ has a perfect matching $M$ in order to illustrate the main idea behind our analysis.

We start with 
\begin{definition}[Level sets $b^k$]\duallabel{def:levelsets}
For each $k\geq 1$ and all $x\geq 0$ denote by $b^k(x)$ the number of vertices in $Q$ that have load {\em at least $x$} after $k$ passes in Algorithm~\ref{alg:main}. 
\end{definition}

Note that  $b^k(x)$ is non-increasing in $x$ and $b^k(x)-b^{k-1}(x)\geq 0$ for all $x$. 
Furthermore, we have 
\begin{equation}\duallabel{eq:int-1}
b^k(0)=|M| \text{~~and~~}\int_0^\infty b^k(x)dx=k|M|.
\end{equation}
The first equality holds since  $G$ is assumed to contain a perfect matching, and the second  holds since every vertex $u\in P$ contributed $1$ unit of water, amounting to $|M|=|P|$ amount of water overall, and \dualeqref{eq:int-1} calculates the sum of loads on all $v\in Q$. 
Furthermore, note that the size of the matching constructed by the algorithm after $k$ passes is exactly equal to 
\begin{equation}\duallabel{eq:int-2}
\frac{1}{k}\int_0^k b^k(x)dx,
\end{equation}
since every vertex $v\in Q$ with load $x$ contributes $\frac1{k}\cdot \min\{k, x\}$ to the matching. Hence the approximation ratio after $k$ passes is at least 
\begin{equation}\duallabel{eq:int-3}
1-\frac1{|M|}\cdot \frac1{k}\int_k^\infty b^k(x)dx,
\end{equation}
where we used \dualeqref{eq:int-1} to convert \dualeqref{eq:int-2} into \dualeqref{eq:int-3}.  Thus, it is sufficient to lower bound $\int_0^k b^k(x)dx$ in order to analyze the approximation ratio, and we turn to bounding this quantity.

First consider the case $k=1$. 
For each such vertex $u$ consider its match $M(u)$. Since $u$ ended up at level at least $x$ after the first pass, its match $M(u)$ must be at level at least $x$ after the first pass as well, as levels are non-decreasing. Hence, we have 
\begin{equation}\duallabel{eq:int-4}
\begin{split}
b^1(x)&=\left|\{u\in P: u\text{~is at level~}\geq x\text{~after first pass}\}\right|\\
&\geq \left|\{u\in P: u\text{~allocated some water at level~}\geq x\text{~during first pass}\}\right|\\
&\geq \int_x^\infty b^1(s)ds
\end{split}
\end{equation}
for all $x\geq 0$. This, however, together with \dualeqref{eq:int-1} can be shown to imply that $\int_x^\infty b^1(s)ds\leq |M|\cdot e^{-x}$ for all $x$. We thus get using \dualeqref{eq:int-3} that the approximation ratio after one pass is at least $1-1/e$.

Now suppose that $k>1$ and consider vertices $v\in Q$ that are at level at least $x$ after $k$-th pass, but were at a lower level after $(k-1)$-th pass. There are exactly $b^k(x)-b^{k-1}(x)$ such vertices. Since these vertices $u$ were at level at least $x$ after $k$-th pass, their matches $M(u)$ must have also been at level at least $x$ after the $k$-th pass, implying similarly to the above that 
\begin{equation}\duallabel{eq:int-5}
b^k(x)\geq \int_x^{\infty} (b^k(s)-b^{k-1}(s))ds
\end{equation}
for all $x\geq 0$. The above equation implies that for all $k\geq 1$ 
\begin{equation}\duallabel{eq:int-6}
\int_x^\infty b^{k}(s)ds\leq |M| \cdot \int_x^\infty F^k(s)ds,
\end{equation}
where $1-F^k(x)$ is the cdf of the Gamma distribution with scale $1$ and shape $k$, i.e. 
$F^k(x)=\int_x^\infty e^{-s}s^{k-1}/(k-1)!ds$.
Using this in \dualeqref{eq:int-3} yields the desired bound on the approximation ratio, i.e. $1-e^{-k}k^{k-1}/k!$. 

\subsection{Analysis in a general case}\label{sec:mp-analysis}
The proof sketch we gave in the previous subsection works under the assumption that $G$ has a perfect matching. The general case is more involved. While the analysis above proceeds by showing that not too much mass will be in the tail $\int_k^\infty b^k(x)dx$, here we find it more convenient to show that substantial mass will be in the head of the distribution, i.e. bound $\int_{0}^k b^k(x)dx$ from below.  We extend the argument using a careful reweighting of vertices and scaling of levels guided by the structure of the {\em canonical decomposition} of $G$ introduced in \cite{gkk:streaming-soda12}, which we now define. 

Let $G=(P, Q, E)$ denote a bipartite graph. For a set $S\subseteq P$ we denote the set of neighbors of $S$ by $\Gamma(S)$. For a number $\alpha>0$ the graph $G$ is said to have vertex expansion at least $\alpha$ if $|\Gamma(S)|\geq \alpha |S|$ for all $S\subseteq P$. The canonical decomposition of $G$ is defined as follows:
\begin{definition}[Canonical decomposition]\duallabel{def:decomposition}
Let $G=(P, Q, E)$ denote a bipartite graph. A partition of $Q=\bigcup_{j\in \I} T_j, T_j\cap T_i=\emptyset, j\neq i$ and $P=\bigcup_{j\in \I} S_j, S_j\cap S_i=\emptyset, j\neq i$ together with numbers $\alpha_j>0$, where $\alpha_j\leq 1$ for $j\leq 0$ and $\alpha_j>1$ for $j>0$  is called a {\em canonical partition} if 
\begin{enumerate}
\item for all $i$ one has $\Gamma\left(\bigcup_{j\in \I, j\leq i} S_j\right)\subseteq \bigcup_{j\in \I, j\leq i} T_j$;
\item $|\Gamma(S)\cap T_j|\geq \alpha_j|S|$ for all $S\subseteq S_j$ for all $j\in \I$;
\item $|T_j|/|S_j|=\alpha_j$, for all $j\in \I$.
\end{enumerate}
Here $\I\subset \mathbb{Z}$ is a set of indices. 
\end{definition}
\newcommand{\fmatching}
{
\begin{center}
\tikzstyle{vertex}=[circle,fill=blue!100, minimum size=15pt,inner sep=1pt, shading=ball,ball color      = blue!100]
\tikzstyle{svertex}=[circle,fill=black!100, minimum size=5pt,inner sep=1pt]
\tikzstyle{evertex}=[circle,draw=none, minimum size=25pt,inner sep=1pt]
\tikzstyle{edge} = [draw,-, color=red!100, very  thick]
\tikzstyle{bedge} = [draw,-, color=green!100, very  thick]
\begin{tikzpicture}[scale=0.4, auto,swap]
   
     \draw[fill=none,opacity=1] (14,7) ellipse (1 and 1);
     \draw[fill=none,opacity=1] (14,1) ellipse (0.5 and 0.5);     
     \node[evertex](uiii) at (14, -1) {};
     \draw(uiii) node {$S_2$};     
     \node[evertex](viii) at (14, 9) {};
     \draw(viii) node {$T_2$};
     
      \node[evertex](w) at (12, -3) {};
     \draw(w) node {$\alpha_2=\frac{|T_2|}{|S_2|}$};

     \draw[fill=none,opacity=1] (18,7) ellipse (2 and 1);
     \draw[fill=none,opacity=1] (18,1) ellipse (1.5 and 1);          
     \node[evertex](uii) at (18, -1) {};
     \draw(uii) node {$S_1$};     
     \node[evertex](vii) at (18, 9) {};
     \draw(vii) node {$T_1$};
     
      \node[evertex](w) at (18, -3) {};
     \draw(w) node {$\alpha_1=\frac{|T_1|}{|S_1|}$};

     \draw[fill=none,opacity=1] (23,1) ellipse (2 and 1);
     \draw[fill=none,opacity=1] (23,7) ellipse (2 and 1);
     \node[evertex](ui) at (23, -1) {};
     \draw(ui) node {$S_0$};
     \node[evertex](vi) at (23, 9) {};
     \draw(vi) node {$T_0$};     
     
      \node[evertex](w) at (23, -3) {};
     \draw(w) node {$\alpha_0=\frac{|T_0|}{|S_0|}$};

     \draw (14, 1) -- (14.5, 7);
     \draw (14, 1) -- (14, 7);
     \draw (14, 1) -- (13.5, 7);     
     
     \draw (14, 1) -- (18, 7);     
     
     \draw (17.5, 1) -- (17, 7);
     \draw (17.5, 1) -- (18, 7);
     \draw (18.5, 1) -- (18, 7);     
     \draw (18.5, 1) -- (19, 7);          
     
    \draw (18.5, 1) -- (22, 7);     
    \draw (18.5, 1) -- (24, 7);

     \draw (22, 1) -- (22, 7);
     \draw (23, 1) -- (23, 7);
     \draw (24, 1) -- (24, 7);     
     
     \draw (23, 1) -- (32, 7);     

     \draw (27.5, 7) -- (27, 1);
     \draw (27.5, 7) -- (28, 1);
     \draw (28.5, 7) -- (28, 1);     
     \draw (28.5, 7) -- (29, 1);

     \draw (32, 7) -- (32.5, 1);
     \draw (32, 7) -- (32, 1);
     \draw (32, 7) -- (31.5, 1);

%%%%%%%%%%%%%%%%%%%%%%%

     \draw[fill=none,opacity=1] (28,1) ellipse (2 and 1);
     \draw[fill=none,opacity=1] (28,7) ellipse (1.5 and 1);          
     \node[evertex](uii) at (28, -1) {};
     \draw(uii) node {$S_{-1}$};     
     \node[evertex](vii) at (28, 9) {};
     \draw(vii) node {$T_{-1}$};
     
      \node[evertex](w) at (28, -3) {};
     \draw(w) node {$\alpha_{-1}=\frac{|T_{-1}|}{|S_{-1}|}$}; 

%%%%%%%%%%%%%%%%%%%%%%%

     \draw[fill=none,opacity=1] (32,1) ellipse (1 and 1);
     \draw[fill=none,opacity=1] (32,7) ellipse (0.5 and 0.5);     
     \node[evertex](uiii) at (32, -1) {};
     \draw(uiii) node {$S_{-2}$};     
     \node[evertex](viii) at (32, 9) {};
     \draw(viii) node {$T_{-2}$};
     
      \node[evertex](w) at (34, -3) {};
     \draw(w) node {$\alpha_{-2}=\frac{|T_{-2}|}{|S_{-2}|}$};

\end{tikzpicture}
\end{center}
}

\begin{figure}
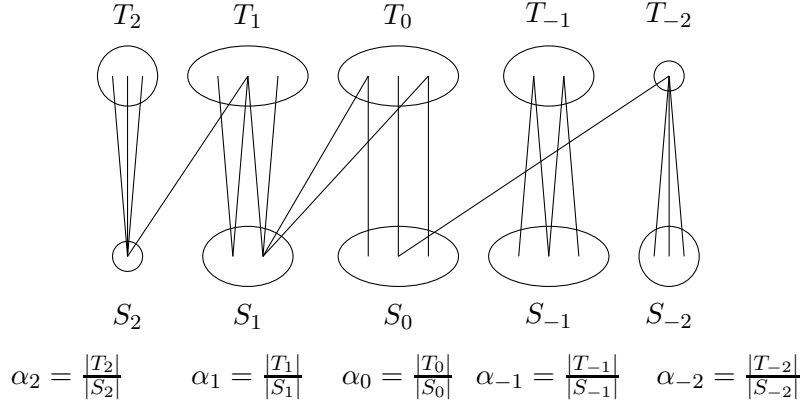

\fmatching
\caption{Canonical decomposition of a bipartite graph. Note that edges from $S_i$ only go to $T_j$ with $j\leq i$ (property (1)).}
\duallabel{fig:cp}
\end{figure}

\if 0
\begin{remark}
For $k=1$, the analysis is inspired by the analysis of the round-robin algorithm in \cite{mpx06}.  We note that the difference in our case is that we essentially consider a fractional version of their process, and obtain significantly better bounds on the quality of approximation. In particular, the best approximation factor that follows from the result of \cite{mpx06} is $1/8$ {\em even after any $k$ passes}, while here we get the optimal $1-1/e$ factor for $k=1$, and an approximation of the form $1-O(1/k^{1/2})$ for all $k>0$.
\end{remark}
\fi

See Fig.~\dualref{fig:cp} for an illustration.  

\paragraph{Vertex capacities and canonical matching.} 
First,  define vertex capacities as follows. For $u\in P$ let $j$ be such that $u\in S_j$(see Fig.~\dualref{fig:cp}), and let $c(u):=\min\{1, \alpha_j\}$ . Similarly, for $v\in Q$ let $j$ be such that $v\in T_j$ (see Fig.~\dualref{fig:cp}) and let $c(v):=\min\{1, 1/\alpha_j\}$. 
We will also use
\begin{claimp}[Monotonicity of capacities]\duallabel{cl:monotonicity}
For every $i\leq j$ and every $v\in T_i, w\in T_j$ one has $c(v)\geq c(w)$. Similarly, for every $i\leq j$ and every $v\in S_i, w\in S_j$ one has $c(v)\leq c(w)$.
\end{claimp}
\begin{proof}
Follows by monotonicity of $\alpha_j$'s.
\end{proof}

\begin{definition}[Canonical matching]\duallabel{def:canonical-matching}
Let $M:E\to [0, 1]$ be a (possibly fractional) matching in $G$ such that  $\sum_{e\in \delta(u)} x_e=c(u)$ for all $u\in P=\bigcup_j S_j$ and $\sum_{e\in \delta(v)} x_e=c(v)$ for all $v\in Q=\bigcup_j T_j$. 
\end{definition}
Such a matching exists by properties (2) and (3) of the canonical decomposition. Furthermore, any such $M$ is a maximum matching in $G$, since $|M|=|C|$, where $C=\left(\cup_{j: \alpha_j\geq 1} S_j\right)\bigcup \left(\cup_{j:\alpha_j< 1} T_j\right)$ forms a vertex cover in $G$ by property (1) of the canonical decomposition. For every integer $j=1,\ldots, k$ and $e\in E$ we let $\wt{M}^j_e\in [0, 1]$ denote the load assigned by our algorithm in the $j$-th pass to edge $e$. Note that $\wt{M}$ does not necessarily form a matching, but for every $u\in P$ and every $j$ one has $\sum_{e=\delta(u)} \wt{M}^j(e)=1$, since every vertex on the $P$ side dispenses one unit of water in every pass. We note that
\begin{claimp}\duallabel{cl:capacities-sum}
For every graph $G$, if $(S_j, T_j)$ is the canonical decomposition (as per Definition~\dualref{def:decomposition}),  $M$ a canonical matching in $G$ (as per Definition~\dualref{def:canonical-matching}), and vertex capacities as defined above, then $\sum_{u\in P} c(u)=\sum_{v\in Q} c(v)=|M|$.
\end{claimp}

%Define $w_{in}(u), u\in P$ by setting $w_{in}(u)=\min\{1, \alpha_j\}$ for $u\in S_j$ (see Fig.~\dualref{fig:cp}). Note that one has $\sum_{u\in P} w_{in}(u)=|M|$. 
%\item[(Sink capacities)] Define $w_{out}(v), v\in Q$ by setting $w_{out}(v)=\min\{1, 1/\alpha_j\}$ for $v\in T_j$ (see Fig.~\dualref{fig:cp}). Note that one has $\sum_{v\in Q} w_{out}(v)=|M|$.  
%\end{description}

\paragraph{Shadow allocation and density function $\phi^k_v(x)$.}  We will use the concept of a {\em shadow allocation}, in which whenever $a$ units of water are added to a vertex $v\in Q$ in the original allocation, $a/c(v)$ units of water are added to $v$ in the shadow allocation.  Now whenever water from a vertex $u\in P$ is added to vertex $v\in Q$ at level $x$ during the $j$-th pass {\em in the shadow allocation}, we let $\phi^j_v(x):=c(u)$, where $\phi$ is the {\em density function}.

The following claim is crucial for our analysis: 

\begin{claimp}\duallabel{cl:phi-sum}
For every graph $G$, if $M$ is a maximum matching in $G$, vertex capacities $c$ and density function $\phi$ are defined as above, one has $\sum_{v\in Q} c(v) \int_0^{\infty} \phi^j_v(x)dx=|M|$
for all $j=1,\ldots, k$. 
\end{claimp}
\begin{proof}
\begin{equation*}
\begin{split}
\sum_{v\in Q} c(v) \int_0^{\infty} \phi^j_v(x)dx&=\sum_{v\in Q} c(v) \sum_{e=(u, v)\in \delta(v)} c(u)\cdot \wt{M}^j(e)/c(v)\\
&=\sum_{v\in Q} \sum_{e=(u, v)\in \delta(v)} c(u)\cdot \wt{M}^j(e)\\
&=\sum_{u\in P} c(u)\sum_{e=(u, v)\in \delta(v)} \wt{M}^j(e)\\
&=\sum_{u\in P} c(u)\\
&=|M|,\\
\end{split}
\end{equation*}
where the first equality is by definition of the shadow allocation, the fourth is by definition of $\wt{M}^j$ and the last is by Claim~\dualref{cl:capacities-sum}.
\end{proof}

\paragraph{Load of a vertex and level of an edge.}  The core of our analysis will consist of bounding the distribution of water levels among vertices in $Q$ in the shadow allocation, showing that there cannot be too many highly overloaded vertices. For a vertex $v\in Q$ let $l^k(v)$ denote the load of $v$ in the shadow allocation after the $k$-th pass. For an edge $e=(u, v)$ let $l^k(e)$ denote the load of $v$  in the shadow allocation after $u$ is processed in the $k$-th pass. The key property of $l^k(e)$ that we need is given by
\begin{lemma}\label{lm:shadow-mon}
For every $k\geq 1$, every $e=(u, v)$ such that $\wt{M}^k(e)>0$ and $f=(u, w)$ such that $M(f)>0$ one has $l^k(f)\geq l^k(e)$.
\end{lemma}
\begin{proof}
Denote the load of $v$ in the original (as opposed to shadow) allocation after $u$ is processed during the $k$-th pass by $x$, and denote the load of $w$ in the original (as opposed to shadow) allocation after $u$ is processed during the $k$-th pass by $y$. We have $y\geq x$ by the definition of the waterfilling algorithm. Also note that $l^k(e)=x/c(v)$ and $l^k(f)=y/c(w)$ by definition of the shadow allocation. By the properties of the canonical decomposition one has $v\in T_i, w\in T_j$ for some $i\leq j$, and hence $c(v)\geq c(w)$ by Claim~\dualref{cl:monotonicity}. We therefore have
$$
l^k(f)=y/c(w)\geq y/c(v)\geq x/c(v)=l^k(e),
$$
as required.
\end{proof}

\paragraph{Reweighted level set sizes $b^k$.} For every $x\geq 0$, integer $k\geq 1$ we let $b^k(x)$ denote the (weighted) number of vertices with load at least $x$ in the shadow allocation, defined as follows:
$$
b^k(x)=\sum_{v\in Q} c(v)\cdot {\mathbbm 1}_{l^k(v)\geq x}.
$$
Note that $b^k(0)=\sum_{v\in Q} c(v)=|M|$ for every $k$. We have
\begin{lemma}\label{lm:alg-lb}
Algorithm~\dualref{alg:main} constructs a matching of size at least
$\frac1{k}\int_0^k b^k(x)dx.$
\end{lemma}
\begin{proof}
For a vertex $v\in Q$ let $l_{org}^k(v)$ denote the water level of at $v$ in the original allocation after $k$ passes. Then $v$ contributes $\frac1{k}\min\{k, l_{org}(v)\}$ to the matching. 
At the same time $l^k(v)=l_{org}^k(v)/c(v)$, so 
\begin{equation*}
\begin{split}
\frac1{k}\int_0^k b^k(x)dx&=\frac1{k}\int_0^k \sum_{v\in Q} c(v)\cdot \mathbbm{1}_{l^k(v)\geq x}dx\\
&=\frac1{k}\sum_{v\in Q} c(v)\cdot \min\{k, l^k(v)\}\\
&=\frac1{k}\sum_{v\in Q} c(v)\cdot \min\{k, l^k_{org}(v)/c(v)\}\\
&=\frac1{k}\sum_{v\in Q} \min\{c(v)\cdot k, l^k_{org}(v)\}\\
&\leq \sum_{v\in Q} \frac1{k}\min\{k, l^k_{org}(v)\},\\
\end{split}
\end{equation*}
where we used the fact that $c(v)\leq 1$ for all $v$ in the last step. This completes the proof of the lemma.
\end{proof}

\paragraph{Bounding the evolution of $b^k(x)$.} In what follows we derive bounds on the reweighted level set sizes $b^k(x)$, which then allow us to lower bound $\frac1{k}\int_0^k b^k(x)dx$. We start with
\begin{lemma}\duallabel{lm:rec}
One has for all $x\geq 0$ and all $k\geq 1$
$$
b^k(x)\geq \int_x^\infty \sum_{v\in Q} c(v) \phi^k_v(s)ds.
$$
\end{lemma}
\begin{proof}
First note that
\begin{equation}\duallabel{eq:wgog2g2g23g}
\begin{split}
\int_x^\infty \sum_{v\in Q} c(v)\cdot  \phi^k_v(s)ds&= \sum_{v\in Q} c(v)\cdot  \sum_{\substack{e=(u, v)\in \delta(v)\\l^k(e)\geq x}} c(u)\cdot \wt{M}^k(e)/c(v)\\
&= \sum_{v\in Q}  \sum_{\substack{e=(u, v)\in \delta(v)\\l^k(e)\geq x}} c(u)\cdot \wt{M}^k(e)\\
&=\sum_{u\in P}  c(u) \sum_{\substack{e=(u, v)\in \delta(u)\\l^k(e)\geq x}} \wt{M}^k(e)
\end{split}
\end{equation}
by definition of the shadow allocation and density function $\phi$.  Recall that for an edge $e=(u, v)$ we let $l^k(e)$ denote the load of vertex $v$ right after $u$ arrives in the $k$-th pass. 
At the same time, 
\begin{equation}\duallabel{eq:923hg9hg23g3g}
\begin{split}
b^k(x)&=\sum_{v\in Q} c(v)\cdot {\mathbbm 1}_{l^k(v)\geq x}\\
&= \sum_{\substack{v\in Q\\l^k(v)\geq x}} \sum_{e=(u, v)\in E} M(e) \text{~~~~~~~(since $\sum_{e=(u, v)\in E} M(e)=c(v)$ for every $v$)}\\
&\geq \sum_{u\in P} \sum_{\substack{e=(u, v)\in E\\ l^k(e)\geq x}} M(e).\\
\end{split}
\end{equation}

Now recall that by Lemma~\ref{lm:shadow-mon} for every $u\in P$ if  $u$ dispensed some water at level at least $x$ {\em in the shadow allocation} during the $k$-th pass, i.e. if
$$
\sum_{e=(u, v)\in \delta(u):l^k(e)\geq x} \wt{M}^k(e)>0,
$$ 
then its canonical matches, namely vertices $v$ such that $M_{(u, v)}>0$, were at level at least $x$ in the shadow allocation after $u$ was processed during $k$-th pass. In particular, in that case we have
$$
\sum_{e=(u, v)\in E: l^k(e)\geq x} M(e)=c(u).
$$ 
Since $\sum_{e=(u, v)\in \delta(u):l^k(e)\geq x} \wt{M}^k(e)\leq 1$ always, we thus get for every $u\in P$
\begin{equation*}
\begin{split}
c(u) \sum_{e=(u, v)\in \delta(u):l^k(e)\geq x} \wt{M}^k(e)\leq  \sum_{e=(u, v)\in E: l^k(v)\geq x} M(e).\\
\end{split}
\end{equation*}
Indeed, if the sum on the lhs is positive, then the sum on the rhs equals $c(u)$ (which suffices since the lhs is bounded by $c(u)$), and if the sum in the lhs is zero, then the inequality holds trivially since the rhs is nonnegative. Summing over $u\in P$, we get
\begin{equation*}
\begin{split}
\sum_{u\in P}  c(u) \sum_{e=(u, v)\in \delta(u):l^k(e)\geq x} \wt{M}^k(e)\leq  \sum_{u\in P}  \sum_{e=(u, v)\in E: l^k(e)\geq x} M(e).\\
\end{split}
\end{equation*}

This, together with~\dualeqref{eq:wgog2g2g23g} and~\dualeqref{eq:923hg9hg23g3g} yields $b^k(x)\geq \int_x^\infty \sum_{v\in Q} c(v)\cdot  \phi^k_v(s)ds$, as required.

%By definition of the canonical decomposition $(S_j, T_j)_{j\in \I}$ for each $j\in \I$ there exists a (possibly fractional) matching $M_j$ in $G$ that matches each $u\in S_j$ exactly $\alpha_j$ times and each $v\in T_j$ exactly once. Let $M_j(u, v)\in [0, 1]$ denote the extent to which $u$ is matched to $v$, so that $\sum_{v\in T_j} M_j(u, v)=\alpha_j$ for all $u\in S_j$ and $\sum_{u\in S_j} M_j(u, v)=1$ for all $v\in T_j$. 

\end{proof}

We now get, letting $b^0\equiv 0$ for convenience,
\begin{lemma}\duallabel{lm:rec-1}
For all $x\geq 0$ and all $k\geq 1$ one has
$|M|-b^k(x)\leq \int_0^x (b^k(s)-b^{k-1}(s))ds.$
\end{lemma}
\begin{proof}
By Lemma~\dualref{lm:rec} we have
$b^k(x)\geq \int_x^\infty \sum_{v\in Q} c(v) \phi^k_v(s)ds.$
Putting this together with Claim~\dualref{cl:phi-sum} we get $
|M|-b^k(x)\leq \int_0^x \sum_{v\in Q} c(v) \phi^k_v(s)ds$
for all $x\geq 0$ and $k\geq 1$.
To complete the proof, we note that, since $\phi_v^k(s)\leq 1$ for all $v, k, s$,
\begin{equation*}
\begin{split}
\int_0^x \sum_{v\in Q} c(v) \phi^k_v(s)ds&\leq \int_0^x \sum_{v\in Q} c(v)\cdot \mathbbm{1}[v\text{~is allocated water at level $s$ in pass $k$}]ds\\
&= \int_0^x \sum_{v\in Q} c(v) \cdot (\mathbbm{1}_{l^k(v)\geq s}-\mathbbm{1}_{l^{k-1}(v)<s}) ds\\
&= \int_0^x ( b^k(s)-b^{k-1}(s))ds
\end{split}
\end{equation*}
for all $k\geq 1$ and $x\geq 0$, where we let $ b^0\equiv 0$ for convenience.
\end{proof}

We now prove lower bounds on $b^k(x)$. Recall that for integer $k\geq 1$ 
\begin{equation}\duallabel{eq:923hg9h3g}
\begin{split}
F^k(x)&=\int_x^\infty e^{-s}s^{k-1}/(k-1)!ds=\sum_{i=0}^{k-1} e^{-x} x^i/i!,
\end{split}
\end{equation}
so that $1-F^k(x)$ is the cdf of the Gamma distribution with scale $1$ and shape $k$.  We now prove our main lower bound on $b^k$:

\begin{lemma}\duallabel{lm:kpass-rec}
For every $k\geq 1$  for all $x\geq 0$ one has $\int_0^x b^k(s)ds\geq |M|\cdot \int_0^x F^k(s)ds.$
\end{lemma}
\begin{proof}
We prove the claim of the lemma by induction on $k$. 
\begin{description}
\item[Base: $k=1$]  Recall that by Lemma~\dualref{lm:rec-1} one has 
\begin{equation}\duallabel{eq:ineq-1}
|M|-b^1(x)\leq \int_0^x b^1(s)ds.
\end{equation}
Letting $f(x)=\int_0^x b^1(s)ds$, we get by rearranging~\eqref{eq:ineq-1} and noting that $f'(x)=b^1(x)$ that $f'(x)\geq |M|-f(x)$ for all $x\geq 0$. We also have $f(0)=0$. Thus, we have $f(x)\geq |M|\cdot (1-e^{-x})=|M|\cdot \int_0^x F^1(s)ds$, as required.
\if 0
Since $b^1(0)=|M|$, this implies that $\int_0^x b^1(s)ds\geq |M| (1-e^{-x})$. 
\fi
\if 0
Indeed, since $b(x)=|M|(1-e^{-x})$ satisfies \dualeqref{eq:ineq-1} with equality, we get after subtracting
\begin{equation}
(b^1(x)-|M|(1-e^{-x}))\geq -\int_0^x (b^1(s)-|M|(1-e^{-s}))ds, 
\end{equation}

Suppose that $\int_0^{x_0} b^1(s)ds<|M| (1-e^{-x_0})$.
Let $g(x)=\left(\int_0^{x_0} b^1(s)ds\right)\cdot e^{-x+x_0}$ for $x\in [0, x_0]$. Then $g(x)$ satisfies \dualeqref{eq:ineq-1} with equality, and hence $ b^1(x)\geq g(x)$ for all $x\in [0, x_0]$. But $g(0)=\left(\int_{x_0}^\infty  b^1(s)ds\right)\cdot e^{x_0}>|M|$, a contradiction with $ b^1(0)=|M|$.
\fi

\item[Inductive step: $k-1\to k$] We need to prove that 
\begin{equation}
\int_0^x  b^{k}(s)ds\geq |M|\cdot \int_0^x F^{k}(s)ds.
\end{equation}

Using Lemma~\dualref{lm:rec-1} and the inductive hypothesis, we get for all $x\geq 0$ 
\begin{equation}\duallabel{eq:ineq-k}
\begin{split}
 b^k(x)&\geq |M|-\int_0^x (b^{k}(s)-b^{k-1}(s))ds\\
&= |M|-\int_0^x b^{k}(s)ds+\int_0^x b^{k-1}(s)ds\\
 &\geq |M|-\int_0^x b^{k}(s)ds+|M|\cdot \int_0^x F^{k-1}(s)ds.\text{~~~~~~(by the inductive hypothesis)}
 \end{split}
\end{equation}
%where the second inequality we used the inductive hypothesis to lowe $\int_0^x b^{k-1}(s)ds$ with $|M|\cdot \int_0^x F^{k-1}(s)ds$.

%Rearranging the terms, we get
%\begin{equation}\duallabel{eq:ineq-k}
%\int_0^x b^{k}(s)ds\geq  |M|-b^k(x)+|M|\cdot \int_0^x F^{k-1}(s)ds.
%\end{equation}
Let $f(x)=\int_0^x b^{k}(s)ds$ (note that $f(0)=0$). We have from \dualeqref{eq:ineq-k} that
\begin{equation*}
\begin{split}
f'(x)\geq |M|-f(x)+|M|\cdot \int_0^x F^{k-1}(s)ds.
\end{split}
\end{equation*}
Thus, for all $x\geq 0$ one has $f(x)\geq  g(x)$, where $g(x)$ is given by the solution of 
\begin{equation*}
\begin{split}
g'(x)=|M|-g(x)+|M|\cdot \int_0^x F^{k-1}(s)ds,
\end{split}
\end{equation*}
which we now solve. The latter implication holds by Claim~\ref{cl:monotone} applied to $|M|-f(x)$. Note that since $g(0)=0$, we have by the above that $g'(0)=|M|$. Thus, $h(x)=g'(x)$ satisfies
\begin{equation}\duallabel{eq:k-pass-g}
h'(x)=-h(x)+|M|\cdot F^{k-1}(x), h(0)=|M|.
\end{equation}
The solution to \dualeqref{eq:k-pass-g} is given by 
\begin{equation}\duallabel{eq:g-eq}
\begin{split}
h(x)=|M|\cdot e^{-x}\left(\int_0^x e^s F^{k-1}(s)ds+1\right).\\
\end{split}
\end{equation}
Calculating the integral in \dualeqref{eq:g-eq} using the expression for $F^{k-1}(s)$ given by ~\dualeqref{eq:923hg9h3g} yields 
\begin{equation}
\begin{split}
\int_0^x e^s F^{k-1}(s)ds=\int_0^x e^s \int_s^\infty \frac1{(k-2)!}z^{k-2}e^{-z}dzds=\int_0^x \sum_{j=0}^{k-2} \frac1{j!}s^{j}ds=\sum_{j=1}^{k-1}\frac1{j!}x^j, 
\end{split}
\end{equation}
and hence 
$$
h(x)=|M|\cdot e^{-x}\left(\int_0^x e^s F^{k-1}(s)ds+1\right)=|M|\cdot e^{-x}\left(\sum_{j=1}^{k-1}\frac1{j!}x^j+1\right)=|M|\cdot F^{k}(x)
$$ 
by~\dualeqref{eq:923hg9h3g}. We thus get $g(x)=\int_0^x h(s)ds=|M|\cdot \int_0^x F^k(s)ds$, and therefore  $\int_0^x b^k(s)ds\geq f(x)\geq |M|\cdot \int_0^x F^k(s)ds$ as required.

\end{description}

\end{proof}

Given Lemma~\dualref{lm:kpass-rec}, we immediately obtain 
\begin{theorem}
Algorithm~\dualref{alg:main} achieves a $(1-e^{-k}\frac{k^{k-1}}{(k-1)!})$-approximation to maximum matchings in $k$ passes over the input stream.
\end{theorem}
\begin{proof}
By Lemma~\ref{lm:alg-lb} together with Lemma~\dualref{lm:kpass-rec} the approximation ratio is at least 
$$
\frac1{|M|}\cdot \frac1{k}\int_0^k  b^k(x)dx\geq \frac1{k}\int_0^k F^k(x)dx=1-\frac1{k}\int_k^\infty F^k(x)dx.
$$

We now recall (by~\dualeqref{eq:923hg9h3g}) that $F^k(x)=\sum_{j=0}^{k-1} e^{-x} x^j/j!$.  Integrating by parts, we have 
$$
\int_k^\infty e^{-x} x^j/j!dx=\left.-e^{-x} x^j/j!\right|_k^\infty+\int_k^\infty e^{-x} x^{j-1}/(j-1)!dx,
$$
and hence 
$$
\frac1{k}\int_k^\infty F^k(x)dx=\frac1{k}\int_k^\infty\sum_{j=0}^{k-1} e^{-x} x^j/j!dx=\frac1{k}\sum_{j=0}^{k-1} (k-j)e^{-k}k^{j}/j!.
$$
Since 
$$
\frac1{k}\sum_{j=0}^{k-1} (k-j)e^{-k}k^{j}/j!=\sum_{j=0}^{k-1} e^{-k}k^j/j!-\sum_{j=1}^{k-1}e^{-k}k^{j-1}/(j-1)!=e^{-k}k^{k-1}/(k-1)!,
$$
we thus get
\begin{equation*}
\begin{split}
\frac1{|M|}\cdot \frac1{k}\int_0^k  b^k(x)dx\geq1-e^{-k}k^{k-1}/(k-1)!,
\end{split}
\end{equation*}
as required.
\end{proof}

\begin{remark}
We note that the approximation ratio satisfies $\frac{e^{-k}k^{k-1}}{(k-1)!}=\frac{1}{\sqrt{2\pi k}}+O(k^{-3/2})$.
\end{remark}

\section{Gap-existence}\duallabel{sec:gap}

In this section we show how our techniques yield an efficient algorithm for Gap-existence, thereby proving Theorem~\dualref{thm:gap}.

We now describe \textsc{DiscretizedWaterfilling}, which is a version of Algorithm~\dualref{alg:main}. We will explicitly maintain a subset $I^*\subset I$ of size $O(|A|/\e)$ while relying on an oracle \textsc{NewNeighbor}$(a, I^*)$ that, given any set $I^*\subseteq I$, outputs any node $i\in I\setminus I^*$ that 
$a$ is connected to or $\emptyset$ if all neighbors of $a$ are in $I^*$.  The difference 

\begin{algorithm}[H]\duallabel{alg:sparse}
\caption{\textsc{DiscretizedWaterfilling}($G, a, \e, k$)}
\begin{algorithmic}[1]
\STATE  $I^*\leftarrow \emptyset$~~~~~~~~~~~~~~~~~~~~~~~~~~~~~~~~~~~~~~~~~~~~~~~~~~~$\rhd N(a)\subseteq I$ stands for the vertex neighborhood of $a\in A$
\WHILE{$\leq 1$ unit of water allocated}
\WHILE{$\exists$ $i\in N(a)\cap I^*$ with level $<(\e/4) k$ and $\leq 1$ unit of water has been allocated}
\STATE Allocate water to $i$ until it is at level $(\e/4) k$
\IF{one unit of water has been allocated from $a$}
\STATE {\bf return}
\ENDIF
\STATE $i\gets \textsc{NewNeighbor}(a, I^*)$ ~~~~~~~~~$\rhd \textsc{NewNeighbor}(a, I^*)$ returns $\emptyset$ if all neighbors of $a$ are in $I^*$
\IF{$i\neq \emptyset$}
%\STATE {\bf return}
%\ELSEIF
\STATE $I^*\leftarrow I^*\cup \{i\}$
\ELSE
\STATE {\bf break} from both loops
\ENDIF

\ENDWHILE
\ENDWHILE
\STATE Perform water filling on neighbors in $I^*$.
\end{algorithmic}
\end{algorithm}

First we prove
\begin{lemma}\duallabel{lm:space}
The space used by Algorithm~\dualref{alg:sparse} is $O(|A|/\e)$.
\end{lemma}
\begin{proof}
Call a vertex {\em saturated} if the amount of water in it is at least $\e k$. The number of saturated vertices is $O(|A|/\e)$ since there are $k|A|$ units of water in the system, and each saturated vertex accounts for at least $\e k$. We say that an unsaturated vertex $i$ belongs to $a\in A$ if $i$ was added to $I^*$ when \textsc{NewNeighbor} was called from $a$. Note that for each $a\in A$ only one $i\in I$ belongs to $a$. Thus, this amounts to at most $|A|$ additional vertices.
\end{proof}

Our algorithm for Gap-Existence is as follows:

\begin{algorithm}[H]\duallabel{alg:gap}
\caption{\textsc{GapExistence}($G$, $\e$)}
\begin{algorithmic}[1]
\STATE Run \textsc{DiscretizedWaterfilling}($G$) with $k=O(\log (\sum_{a\in A}B_a/\e)/\e^2)$.
%\STATE Let $G'$ denote the support of the fractional solution.
%\STATE Output {\bf YES} if a complete matching with budgets $\lfloor (1-\e)B_a\rfloor$ exists in $G'$, {\bf NO} otherwise.
\STATE Output {\bf YES} if at most $\e/2$ water is allocated above level $k/(1-\e/2)$, and {\bf NO} otherwise.
\end{algorithmic}
\end{algorithm}

We now assume that we are in the {\bf YES} case, i.e. there exists a matching with budgets $B_a$, and prove that the algorithm will find a matching with budgets $\lfloor (1-\e) B_a\rfloor$.

We recall definitions of levels and level set sizes below.
\begin{definition}
Define $l^k(i)$ to be the level of water at vertex $i$ after the $k$-th pass  (here we refer to the level of water in the actual allocation constructed by waterfilling, not the shadow allocation used for analysis purposes in Section~\ref{sec:mp-analysis}). 
\end{definition}

\begin{definition}[Level set sizes]
For each $k\geq 1$ and all $x\geq 0$ denote by $b^k(x)$ the number of vertices in $I$ that have load {\em at least $x$} after $k$ passes, i.e. the number of vertices $i\in I$ with $l^k(i)\geq x$.
\end{definition}

Note that $b^k(x)$ is non-increasing in $x$ and $b^k(x)-b^{k-1}(x)\geq 0$ for all $x$, and $\int_0^\infty b^k(x)dx=k|A|$, since every vertex dispenses one unit of water in every pass.

We now note that the allocation constructed by \textsc{DiscretizedWaterfilling} can be used to obtain a matching as follows: we first scale the allocation by a factor of $1-\e/2$, then take all water allocated below level $k$. Dividing by $1/k$ gives a matching where every vertex in $A$ is assigned at least 
\begin{equation}\duallabel{eq:int-2r}
(1-\e/2)\cdot B_a-\frac{1}{k}\int_{k/(1-\e/2)}^\infty b^k(x)dx
\end{equation}
fractional mass. Thus, if the second term is bounded by $\e/2$, then the graph contains a matching with budgets $\lfloor (1-\e)B_a\rfloor, a\in A$, i.e. if the algorithm outputs {\bf YES}, it is correct. In what follows we show that in the {\bf YES} case, i.e. when the input graph admits a matching with budgets $B_a, a\in A$, the second term is indeed bounded by $\e/2$.

For simplicity of notation we assume from now on that every $a\in A$ is replaced with $B_a$ unit demand copies (and we use $A$ to denote the set of those copies, abusing notation somewhat). We assume that we are in the {\bf YES} case, i.e. the original graph contains a matching with budgets $B_a$, and thus the new graph admits a perfect matching of the $A$ side -- denote this matching by $M$. For every edge $e$ of $G$, every $k$ we let $\wt{M}^k(e)$ denote the amount of fractional mass allocated along edge $e$ in the $k$-th pass. For an edge $e=(a, i)$ we let $l^k(e)$ denote the load of vertex $i$ right after $a$ arrives in the $k$-th pass, and let $l^k(i)$ denote the load of $i$ after the $k$-th pass.

\begin{lemma}\duallabel{lm:gap-rec}
One has for all $k\geq 1$  and $x\geq (\e/4) \cdot k$
\begin{equation}\duallabel{eq:int-5r}
b^k(x)\geq \int_x^{\infty} (b^k(s)-b^{k-1}(s))ds.
\end{equation}
where $b^0\equiv 0$.
\end{lemma}
\begin{proof}
Intuitively, the lemma follows since if a vertex $a\in A$ ended up allocating water at level at least $x$ during  the $k$-th pass, its match must have been at level at least $x$ when $a$ arrived. Together with the fact that levels are monotone increasing this gives the result. We now give the details.

First note that
\begin{equation}\duallabel{eq:wgog2g2g23g-uebrg}
\begin{split}
\int_x^\infty \sum_{i\in I} (b^k(s)-b^{k-1}(s)) ds&= \sum_{i\in I} \sum_{\substack{e=(a, i)\in \delta(i)\\l^k(e)\geq x}} \wt{M}^k(e)\\
&=\sum_{a\in A}  \sum_{\substack{e=(a, i)\in \delta(a)\\l^k(e)\geq x}} \wt{M}^k(e)
\end{split}
\end{equation}

At the same time
\begin{equation}\duallabel{eq:923hg9hg23g3g-uebrg}
\begin{split}
b^k(x)&=\sum_{i\in I} {\mathbbm 1}_{l^k(i)\geq x}\\
&\geq \sum_{\substack{i\in I\\l^k(i)\geq x}} \sum_{e=(a, i)\in E} M(e) \text{~~~~~~~(since $\sum_{e=(a, i)\in E} M(e)\leq 1$ for every $i\in I$)}\\
&\geq \sum_{a\in A} \sum_{\substack{e=(a, i)\in E\\ l^k(e)\geq x}} M(e).\\
\end{split}
\end{equation}

Now note that if  $a\in A$ dispensed some water at level at least $x$ during the $k$-th pass, i.e. if
$$
\sum_{e=(a, i)\in \delta(a):l^k(e)\geq x} \wt{M}^k(e)>0,
$$ 
then vertices $i$ such that $M_{(a, i)}>0$ were at level at least $x$ after $a$ was processed during $k$-th pass. In particular, in that case we have
$$
\sum_{e=(a, i)\in E: l^k(e)\geq x} M(e)=1.
$$ 
Since $\sum_{e=(a, i)\in \delta(a):l^k(e)\geq x} \wt{M}^k(e)\leq 1$ always, we thus get for every $a\in A$
\begin{equation*}
\begin{split}
\sum_{e=(a, i)\in \delta(a):l^k(e)\geq x} \wt{M}^k(e)\leq  \sum_{e=(a, i)\in E: l^k(i)\geq x} M(e).\\
\end{split}
\end{equation*}
Indeed, if the sum on the lhs is positive, then the sum on the rhs equals $1$ (which suffices since the lhs is bounded by $1$), and if the sum in the lhs is zero, then the inequality holds trivially since the rhs is nonnegative. Summing over $a\in A$, we get
\begin{equation*}
\begin{split}
\sum_{a\in A}   \sum_{e=(a, i)\in \delta(a):l^k(e)\geq x} \wt{M}^k(e)\leq  \sum_{a\in A}  \sum_{e=(a, i)\in E: l^k(e)\geq x} M(e).\\
\end{split}
\end{equation*}

This, together with~\dualeqref{eq:wgog2g2g23g-uebrg} and~\dualeqref{eq:923hg9hg23g3g-uebrg} yields $b^k(x)\geq \int_x^\infty \sum_{i\in I} (b^k(s)-b^{k-1}(s))ds$, as required.
\end{proof}

We now get, letting $\Delta=(\e/4)\cdot k$ to simplify notation, 
\begin{lemma}\duallabel{lm:gap-q}
For all $k\geq 1$  and all $x\geq \Delta$, then
\begin{equation}\duallabel{eq:int-6r}
\int_x^\infty b^{k}(s)ds\leq |A|\cdot \int_{x-\Delta}^\infty F^k(s)ds.
\end{equation}
\end{lemma}
\begin{proof}
We prove the lemma by induction on $k$. 
\begin{description}
\item[Base: $k=1$]   Recall that by Lemma~\dualref{lm:gap-rec} one has 
\begin{equation}\duallabel{eq:gap-ineq-1}
 b^1(x)\geq \int_x^\infty  b^1(s)ds, 
\end{equation}
for all $x\geq \Delta$. Let $f(x)=\int_x^\infty b^1(s)ds$, so that $f(x)\leq |A|$ for every $x$, and note that for $x\geq \Delta$
$$
f'(x)=-b^1(x)\leq -\int_x^\infty  b^1(s)ds=-f(x).
$$
Let $g(x)$ be a function such that $g(\Delta)=|A|$ and $g'(x)=-g(x)$ for all $x\geq \Delta$. Then we have $f(x)\leq g(x)$ for all $x\geq \Delta$, and therefore for all $x\geq \Delta$
$$
\int_x^\infty b^1(s)ds=f(x)\leq g(x)=|A|e^{-x+\Delta}=|A|\cdot \int_{x-\Delta}^\infty e^{-s}ds=|A|\cdot \int_{x-\Delta}^\infty F^1(s)ds.
$$

\item[Inductive step: $k-1\to k$] We need to prove that 
\begin{equation}
\int_x^\infty  b^{k}(s)ds\leq |A|\cdot \int_x^\infty F^{k}(s)ds.
\end{equation}

Using Lemma~\dualref{lm:gap-rec} we get for all $x\geq \Delta$ 
\begin{equation}\duallabel{eq:gap-ineq-k}
 b^k(x)\geq \int_x^{\infty} ( b^{k}(s)- b^{k-1}(s))ds=\int_x^{\infty}  b^{k}(s)ds-|A|\cdot \int_{x-\Delta}^\infty F^{k-1}(x),
\end{equation}
where we used the inductive hypothesis to upper bound $\int_x^{\infty}  b^{k-1}(s)ds$ with $|A|\cdot \int_{x-\Delta}^\infty F^{k-1}(s)ds$. We thus have that the function $f(x)=\int_x^\infty b^k(s)ds$ satisfies 
$$
f'(x)\leq -f(x)+|A|\cdot \int_{x-\Delta}^\infty F^{k-1}(x), f(\Delta)\leq k|A|,
$$
where the last condition comes from the fact that every vertex in $A$ dispenses one unit of water in every pass overall, so the total amount of water dispensed at level $x$ or above is bounded by $k|A|$. 

Let $g$ satisfy 
\begin{equation}\duallabel{eq:gap-k-pass-g}
g'(x)=-g(x)+|A|\cdot \int_{x-\Delta}^\infty F^{k-1}(s)ds, g(\Delta)=k|A|,
\end{equation}
so that $g(x)\geq f(x)=\int_x^\infty b^k(s)ds$ for $x\geq \Delta$ (by Claim~\ref{cl:monotone} below).  Let $h(x)=g'(x)$, so that 
\begin{equation}\duallabel{eq:gap-k-pass-h}
h'(x)=-h(x)-|A|\cdot F^{k-1}(x-\Delta)
\end{equation}
and $h(\Delta)=g'(\Delta)=-g(\Delta)+|A|\cdot \int_0^\infty F^{k-1}(s)ds=-k|A|+(k-1)|A|=-|A|$. The second to last equality holds since $\int_0^\infty F^{k-1}(s)ds$ equals the expectation of the sum of $k-1$ exponentially distributed variables of unit scale, which is $k-1$.

The solution to \dualeqref{eq:gap-k-pass-h} is given by
\begin{equation}\duallabel{eq:gap-h-eq}
\begin{split}
h(x)=e^{-x+\Delta}\left(-|A|\cdot\int_0^{x-\Delta} e^s F^{k-1}(s)ds-|A|\right).\\
\end{split}
\end{equation}
Calculating the integral in \dualeqref{eq:gap-h-eq} yields 
\begin{equation}
\begin{split}
\int_0^{x-\Delta} e^s F^{k-1}(s)ds&=\int_0^{x-\Delta} e^s \int_s^\infty \frac1{(k-2)!}z^{k-2}e^{-z}dzds\\
&=\int_0^{x-\Delta} \sum_{j=0}^{k-2} \frac1{j!}s^{j}ds\\
&=\sum_{j=1}^{k-1}\frac1{j!}(x-\Delta)^j, 
\end{split}
\end{equation}
and hence 
\begin{equation*}
\begin{split}
h(x)&=e^{-x+\Delta}\left(-|A|\cdot\sum_{j=1}^{k-1}\frac1{j!}(x-\Delta)^j-|A|\right)\\
&=e^{-x+\Delta}\left(-|A|\cdot\sum_{j=0}^{k-1}\frac1{j!}(x-\Delta)^j\right)\\
&=-|A|\cdot F^k(x-\Delta)\\
\end{split}
\end{equation*}
by~\dualeqref{eq:923hg9h3g}.  Therefore 
\begin{equation*}
\begin{split}
g(x)&=g(\Delta)+\int_\Delta^x h(s)ds\\
&=k|A|+\int_\Delta^x h(s)ds\\
&=-\int_x^\infty h(s)ds\\
&=|A|\cdot \int_{x-\Delta}^\infty F^k(s-\Delta)ds,
\end{split}
\end{equation*}
and $\int_x^\infty b^k(s)ds=f(x)\leq g(x)=|A|\cdot \int_{x-\Delta}^\infty F^k(s-\Delta)ds$, as required.

\end{description}
\end{proof}

\begin{claim}\label{cl:monotone}
For every $g: \mathbb{R}\to \mathbb{R}$, if $f:\mathbb{R}\to \mathbb{R}$ satisfies $f'(x)\leq -f(x)+g(x), f(0)=a$ for some $a=0$, then 
 $f(x)\leq h(x)$ for $h: \mathbb{R}\to \mathbb{R}$ that satisfies $h(x)=-h'(x)+g(x), h(0)=a$ and is pointwise non-decreasing in $a$.
\end{claim}
\begin{proof}
Let $q(x)=e^x f(x)$, so that
$$
q'(x)=e^x f(x)+e^x f'(x)\leq e^x f(x)+e^x (-f(x)+g(x))=e^x g(x).
$$
Integrating from $0$ to $x$, we get $q(x)\leq q(0)+\int_0^x e^s g(s)ds=f(0)+\int_0^x e^s g(s)ds$. Letting $h(x):=e^{-x}(\int_0^x e^s g(s)ds+f(0))$, we note that
$$
f(x)=e^{-x} q(x)\leq e^{-x} (f(0)+\int_0^x e^s g(s)ds)=h(x).
$$
It remains to note that $h'(x)=-h(x)+g(x)$ for all $x\geq 0$, $h(0)=f(0)=a$, and $h(x)$ is non-decreasing in $f(0)=a$, as required.
\end{proof}

We will need
\begin{lemma}\duallabel{lm:z}
For all $k\geq 1$ and $\delta\geq 0$ 
\begin{equation*}
\frac1{k}\int_{k(1+\delta)}^\infty F^k(x)dx\leq k\cdot e^{-\delta k}(1+\delta)^{k}
\end{equation*}
\end{lemma}
\begin{proof}
Recalling that $F^k(x)=\sum_{j=0}^{k-1} e^{-x} x^j/j!$ and using integration by parts
$$
\int_{k(1+\delta)}^\infty e^{-x} x^j/j!dx=\left.-e^{-x} x^j/j!\right|_{k(1+\delta)}^\infty+\int_{k(1+\delta)}^\infty e^{-x} x^{j-1}/(j-1)!dx,
$$
we get
\begin{equation}
\begin{split}
\int_{k(1+\delta)}^\infty F^k(x)dx&=\int_{k(1+\delta)}^\infty\sum_{j=0}^{k-1} e^{-x} x^j/j!dx\\
&=\sum_{j=0}^{k-1} (k-j)e^{-k(1+\delta)}(k(1+\delta))^{j}/j!\\
&\leq e^{-\delta k}(1+\delta)^{k} \cdot e^{-k}\sum_{j=0}^{k-1} k^{j+1}/j!\\
&\leq e^{-\delta k}(1+\delta)^{k} \cdot k e^{-k}\sum_{j=0}^\infty k^j/j!\\
&=k\cdot e^{-\delta k}(1+\delta)^{k}.
\end{split}
\end{equation}
\end{proof}

We now use Lemma~\dualref{lm:gap-q} to upper bound the second term in~\dualeqref{eq:int-2r} by $\e/2$, as required. By Lemma~\ref{lm:gap-q} we have
\begin{equation}\duallabel{eq:gap-bound}
\frac1{k}\int_{k/(1-\e/2)}^{\infty} b^{k}(x)dx\leq |A|\cdot  F^{k}(k/(1-\e/2)-(\e/4)k).
\end{equation}
Since
\begin{equation*}
\begin{split}
k/(1-\e/2)-(\e/4)k&= k\cdot \frac{1-\e/2+\e/2}{1-\e/2}-k\cdot \frac{(\e/4)(1-\e/2)}{1-\e/2}\\
&= k\left(1+\frac{\e/2-(\e/4)(1-\e/2)}{1-\e/2}\right)\\
&\geq k\left(1+\e/4\right)\\
\end{split}
\end{equation*}
when $\e$ is smaller than an absolute constant, we get, letting $\delta=\e/4$ for convenience of notation, that  by Lemma~\dualref{lm:z}
\begin{equation}
\begin{split}
\frac1{k}\int_{k(1+\delta)}^{\infty}  b^k(x)dx&\leq k\cdot e^{-k(\delta-\ln (1+\delta))}\leq k\cdot e^{-\Omega(\e^2 k)}
\end{split}
\end{equation}
as long as $\delta=\Theta(\e)$ is smaller than an absolute constant.  Hence, letting $k=C\ln (\frac1{\e}\cdot \sum_{a\in A} B_a)/\e^2$ for a sufficiently large constant $C>0$, we get by~\dualeqref{eq:int-2r} that every advertizer is satisfied with budget at least 
\begin{equation*}
\begin{split}
(1-\e/2)\cdot B_a-\frac{1}{k}\int_{k/(1-\e/2)}^\infty b^k(x)dx&\geq \alpha\cdot B_a-\left(\sum_{a\in A} B_a\right)\cdot ke^{-\Omega(\e^2 k)}\\
&\geq (1-\e/2)\cdot B_a-\e/2\\
&\geq (1-\e)\cdot B_a.\\
\end{split}
\end{equation*}
 This completes the proof of Theorem~\dualref{thm:gap}.

\newpage
\pdfbookmark[1]{\refname}{My\refname} 
%\bibliographystyle{alpha}
%\bibliography{doc-t}

\newcommand{\etalchar}[1]{$^{#1}$}

\begin{appendix}
%!TEX root = ./doc-t.tex 
\section{Proofs omitted from Section~\ref{sec:main}} \label{app:aux}

\noindent{\em {\bf Lemma~\ref{lm:intersection-size}} (Restated)
For every $m\geq 2$, integer $W\geq 1$ and $\delta'\in (0, 1)$ such that $1/\delta'$ is an integer, if $Y=[m^4]^m$ and the set $\S$ is defined by 
$$
\S=\{y\in Y: (y, \u) +\Delta_\u \mod W\in [a_\u, b_\u)\cdot W, \text{~for all~}\u\in \mathcal{U}\},
$$ where $\mathcal{U}$ is a collection of binary vectors of fixed length $w$ and $a_\u, b_\u\in [0, 1]$  are constant integer multiples of $1/L$ for an integer $L$, the following conditions hold if $W$ is an integer multiple of $w\cdot \text{lcm}(L, 1/\delta')$, $\Delta_\u/W$ are multiples of $1/L$  and $m$ is sufficiently large. 

 If $\max_{\substack{\u\in \mathcal{U}, \v\in \mathcal{U}\\\u\neq \v}} (\u,\v)/|\v|\leq \delta'$, then 
$$
\left||\S|-|Y|\cdot \prod_{\u\in \mathcal{U}} (b_\u-a_\u)\right|\leq |\mathcal{U}|^2(6L\delta'  +4/m)\cdot |Y|.
$$
}

Before proving the lemma we introduce some definitions. Throughout this section we use the notation  $Y=[m^4]^m$ for integer $m$. First define
\begin{definition}[Bad vertices] \label{def:bad-vertices}
We let $B\subseteq Y$ denote the set of {\em bad} vertices, i.e. vertices with at least one coordinate close to $0$ or $m^4$:
\begin{equation*}
B:=\{x\in Y: \exists i\in [m] \text{~such that~}x_i<m^2\text{~or~}x_i>m^4-m^2\}.
\end{equation*}
\end{definition}

We will use 
\begin{lemma}[The hypercube $Y={[m^4]}^m$ contains few bad vertices] \label{lm:few-bad-vertices}
For every integer $m>1$, if $Y=[m^4]^m$, then $|B|\leq  (2/m) |Y|$.
\end{lemma}
\begin{proof}
Follows directly by a union bound
$$
|B|\leq \sum_{i\in [m]} |\{x\in Y: x_i<m^2\text{~or~}x_i>m^4-m^2\}|\leq m\cdot (2m^2/m^4)\cdot |Y|\leq (2/m) |Y|.
$$
\end{proof}

We will extensively use the notion of  a discretization of the cube $Y=[m^4]^m$:
\begin{definition}[Discretization with precision $L$]\label{def:discretization}
For every integer $m, W, L>1$, $\delta\in (0, 1)$, every collection $\mathcal{U}$ of binary vectors of length $m$, every $q\in [L]^{\mathcal{U}}$  define
\begin{equation*}
S(q):=\left\{y\in Y: (y, \u) \mod W\in \left[\frac{q_\u-1}{L}, \frac{q_\u}{L}\right)\cdot W, \text{~for all~}\u\in \mathcal{U}\right\},
\end{equation*}
and
\begin{equation*}
\text{Int}_\delta(S(q)):=\left\{y\in Y: (y, \u) \mod W\in \left[\frac{q_\u-1}{L}+\delta, \frac{q_\u}{L}-\delta\right]\cdot W, \text{~for all~}\u\in \mathcal{U}\right\}.
\end{equation*}
\end{definition}

We will also use
\begin{definition}[Shifting map $\psi$]\label{def:shifting-map}
For every integer $m, W, L>1$, every collection $\mathcal{U}$ of binary vectors of weight $w$ such that $W/w$ is an integer, for every pair $q, r\in [L]^{\mathcal{U}}$ let 
$$
\psi_{r\to q}(y):=y+\sum_{\u\in \mathcal{U}}  \frac{W}{L\cdot w} (q-r)_{\u}\cdot \u.
$$

\end{definition}

We will use 

\begin{lemma}\label{lm:inclusion}
For every $\delta\in (0, 1)$, integer $m, W, L>1$, every collection $\mathcal{U}$ of binary vectors of weight $w$, if $Y=[m^4]^m$, $B\subseteq Y$ is the set of bad vertices (as per Definition~\ref{def:bad-vertices}), then the following conditions hold. If $|\mathcal{U}|\cdot (W/w)<m^2$, $\max_{\substack{\u\in \mathcal{U}, \v\in \mathcal{U}\\\u\neq \v}} (\u,\v)/|\v|\leq \delta'$ for some $\delta'\in (0, \delta/|\mathcal{U}|)$, then for every pair $q, r\in [L]^{\mathcal{U}}$ we have
$$
\psi_{r\to q}(\text{Int}_\delta(S(r))\setminus B)\subseteq S(q).
$$
\end{lemma}
\begin{proof}
First note that for every $y\in \text{Int}_\delta(S(r))\setminus B$ one has for every $q\in [L]^{\mathcal{U}}$
$$
\psi_{r\to q}(y)=y+\sum_{\u\in \mathcal{U}}  (q-r)_{\u}\cdot \u\cdot \frac{W}{L\cdot w}\in Y,
$$
since
\begin{equation*}
\begin{split}
\left\|\sum_{\u\in \mathcal{U}}  (q-r)_{\u}\cdot \u\cdot \frac{W}{L\cdot w}\right\|_\infty&\leq \sum_{\u\in \mathcal{U}}  \left\|(q-r)_{\u}\cdot \u\cdot \frac{W}{L\cdot w}\right\|_\infty\\
&\leq |\mathcal U| \frac{W}{L\cdot w} ||q-r||_\infty ||\u||_\infty \leq  |\mathcal{U}|\cdot (W/w)< m^2.
\end{split}
\end{equation*}
To obtain the last inequality we used the assumption that $\u$ is a binary vector, and the assumption of the lemma that $|\mathcal{U}|(W/w)<m^2$.

The rest of the proof proceeds in two steps. We first prove basic bounds on the dot product of $\psi_{r\to q}(y)$ with vectors $\u\in \mathcal{U}$, and then put these bounds together to obtain the result of the lemma. We have for every $y\in Y$ and $\u\in \mathcal{U}$
\begin{equation}\label{eq:23gwgewg}
\begin{split}
(\psi_{r\to q}(y), \u)&=(y, \u)+(q-r)_{\u} \cdot |\u|\cdot \frac{W}{L\cdot w}\text{~~~~~~~~~(since $\u$ is a binary vector)}\\
&+\sum_{\v \in \mathcal{U}, \v\neq \u} (q-r)_{\v} \cdot (\v, \u)\cdot \frac{W}{L\cdot w}\\
&=(y, \u)+\frac{(q-r)_{\u}}{L} \cdot W\text{~~~~~~~~~~~~~~~~~~~~~~~(intended shift in direction of $\u$)}\\
&+\sum_{\v \in \mathcal{U}, \v\neq \u} (q-r)_{\v} \cdot (\v, \u)\cdot \frac{W}{L\cdot w}\text{~~~~~~(small error term from near-orthogonality)}\\
\end{split}
\end{equation}
We now bound the error term (the last line) in the previous equation. We have, using the assumption that  $\max_{\substack{\u\in \mathcal{U}, \v\in \mathcal{U}\\\u\neq \v}} (\u,\v)/|\v|\leq \delta'$ as well as the assumption that all vectors in $\mathcal U$ have the same Hamming weight $w$, that
\begin{equation}\label{eq:23gwgegg34gwg}
\begin{split}
\left|\sum_{\v \in \mathcal{U}, \v\neq \u} (q-r)_{\v} (\v, \u)\cdot \frac{W}{L\cdot w}\right|&\leq \sum_{\v \in \mathcal{U}, \v\neq \u} \left|(q-r)_{\v}\right|\cdot  \delta' \cdot |\u|\cdot \frac{W}{L\cdot w}\\
&\leq \sum_{\v \in \mathcal{U}, \v\neq \u} \left|(q-r)_{\v}\right| \delta' \frac{W}{L}\\
&\leq \sum_{\v \in \mathcal{U}, \v\neq \u}  \delta' W\text{~~~~~~(since $||q-r||_\infty \leq L$)}\\
&\leq |\mathcal{U}|  \delta' W\\
\end{split}
\end{equation}

Combining ~\eqref{eq:23gwgewg} and~\eqref{eq:23gwgegg34gwg},  we thus get 
\begin{equation}\label{eq:13onjfngege}
\begin{split}
\left|(\psi_{r\to q}(y), \u)-((y, \u)+\frac{(q-r)_{\u}}{L} \cdot W)\right|\leq |\mathcal{U}|  \delta' W.
\end{split}
\end{equation}
Equipped with the bound above, we now proceed to complete the proof of the lemma. 

We now show that for every $y\in \text{Int}_\delta(S(r))\setminus B$ one has $\psi_{r\to q}(y)\in S(q)$. Indeed, for each $y\in \text{Int}_\delta(S(r))$ one has  by definition of $\text{Int}_\delta(S(q))$ (Definition~\ref{def:discretization})
\begin{equation*}
\begin{split}
((y, \u)+\frac{(q-r)_{\u}}{L} \cdot W) \mod W&\in   \left(\left[\frac{r_\u-1}{L}+\delta, \frac{r_\u}{L}-\delta\right]+\frac{(q-r)_{\u}}{L} \cdot W\right)\mod W\\
&\in   \left[\frac{q_\u-1}{L}+\delta, \frac{q_\u}{L}-\delta\right]\mod W
\end{split}
\end{equation*}

Combining the equation above with \eqref{eq:13onjfngege}, we thus get for every $y\in Y\setminus B$
$$
(\psi_{r\to q}(y), \u) \mod W\in \left[\frac{q_\u-1}{L}, \frac{q_\u}{L}\right)\mod W
$$
since $\delta'<\delta/|\mathcal{U}|$ by assumption of the lemma. We have thus proved that for every $q, r\in [L]^{\mathcal U}$ one has 
\begin{equation*}
\psi_{r\to q}(\text{Int}_\delta(S(r))\setminus B)\subseteq S(q),
\end{equation*}
as required.

\end{proof}

We will also use 
\begin{lemma}\label{lm:thin-slice}
For integer $m, w, W, L>1$ such that $m^2>W/w$, every vector $\u\in \{0, 1\}^m$ of Hamming weight $w$, if $Y=[m^4]^m$, $B\subseteq Y$ is the set of bad vertices (as per Definition~\ref{def:bad-vertices}), then the following conditions hold. If $W/(L\cdot w)$ is a positive integer, then for every $\mathcal{I}\subseteq [L]$ we have
$$
\left|\left\{y\in Y: (y, \u) \mod W\in \left[\frac{j-1}{L}, \frac{j}{L}\right]\cdot W, j\in \mathcal{I}\right\}\right|\leq (|\mathcal{I}|/L+2/m)|Y|.
$$
\end{lemma}
\begin{proof}
Consider a discretization of the cube (similarly to Definition~\ref{def:discretization}) with  $\mathcal{U}=\{\u\}$. Specifically, for $j\in [L]$ let
$$
Z_j:=\left\{y\in Y: (y, \u) \mod W\in \left[\frac{j-1}{L}, \frac{j}{L}\right]\cdot W\right\}
$$
and 
$$
Z'_j:=\left\{y\in Y\setminus B: (y, \u) \mod W\in \left[\frac{j-1}{L}, \frac{j}{L}\right]\cdot W\right\}.
$$
We also let $z_j:=|Z_j|$ and $z'_j:=|Z'_j|$ for every $j\in [L]$.

To bound the size of $Z_j$ and $Z_j'$, we use the shifting map $\psi$ from Definition~\ref{def:shifting-map} with $\mathcal{U}=\{\u\}$, so that for every $i, j\in [L]$
$$
\psi_{j\to i}(y)=y+(i-j)\cdot \u\cdot \frac{W}{L\cdot |\u|}.
$$
Note that for every $y\in Y$ one has\footnote{Note that here we prove  stronger properties of the shifting map than those proved in Lemma~\ref{lm:inclusion}, but only for the special case of $\mathcal{U}$ containing a single element.}
\begin{equation*}
\begin{split}
(\psi_{j\to i}(y), \u)&=(y, \u)+(i-j) \cdot |\u|\cdot \frac{W}{L\cdot |\u|}=(y, \u)+\frac{i-j}{L} \cdot W,
\end{split}
\end{equation*}
and for every coordinate $s\in [m]$ 
\begin{equation*}
\begin{split}
(\psi_{j\to i}(y))_s&=\left(y+(i-j) \cdot \u \cdot \frac{W}{L\cdot |\u|}\right)_s=y_s+(i-j) \cdot \u_s\cdot  \frac{W}{L\cdot |\u|}.
\end{split}
\end{equation*}

Since for every $y\in Y\setminus B$ and every coordinate $s\in [m]$ we have $y_s\in [m^2, m^4-m^2]$,  $|i-j|\leq L$ and $\u$ is a binary vector, we get that 
$$
y_s+(i-j) \cdot \u_s\cdot  \frac{W}{L\cdot |\u|}\in [m^2-W/|\u|, m^4-m^2+W/|\u|]\subseteq [m^4]
$$
since $m^2>W/|\u|$ by assumption of the lemma, as required. We thus conclude that for every $i, j\in [L]$ one has 
$$
\psi_{j\to i}(Z(j)\setminus B)\subseteq Z(i).
$$

Since $\psi_{j\to i}$ is injective for all $i, j\in [L]$, we therefore have that $z'_j=|Z'_j|\leq |Z_{i}|=z_i$ for all $i, j\in [L]$. We thus have, for any subset $\mathcal{I}\subseteq [L]$ of indices
\begin{equation*}
\begin{split}
\sum_{j\in \mathcal{I}} z_j&\leq \sum_{j\in \mathcal{I}} z'_j+\sum_{j\in \mathcal{I}} (z_j-z'_j)\\
&\leq (|\mathcal{I}|/L)\sum_{j\in [L]} z_j+\sum_{j\in \mathcal{I}} (z_j-z'_j)\text{~~~~~(since $z'_j\leq z_i$ for all $i\in [L]$)}\\
&\leq (|\mathcal{I}|/L)\sum_{j\in [L]} z_j+|B|\text{~~~~~~~~~~~~~~~~~~~~(since $Z_j\setminus Z'_j\subseteq B$ and $Z_j$'s are disjoint)}\\
&\leq (|\mathcal{I}|/L)\sum_{j\in [L]} z_j+(2/m) |Y|\text{~~~~~~~~~(by~Claim~\ref{def:bad-vertices})}\\
\end{split}
\end{equation*}
as required.
\end{proof}

\begin{lemma}\label{lm:delta-slice}
For every $\delta\in (0, 1)$ such that $1/\delta$ is an integer, integer $m, W, L>1$ such that $W/(\text{lcm}(L, 1/\delta)\cdot w)$ is an integer, every vector $\u\in \{0, 1\}^m$ of weight $w$, if $Y=[m^4]^m$, $B\subseteq Y$ is the set of bad vertices (as per Definition~\ref{def:bad-vertices}), then for every $\mathcal{I}\subseteq [L]$ we have
$$
\sum_{q\in [L]^{\mathcal{U}}} |S(q)\setminus \text{Int}_\delta(S(q))|\leq |\mathcal{U}| (3\delta L+2/m)\cdot |Y|
$$
\end{lemma}
\begin{proof}
One has using  Definition~\ref{def:discretization}
\begin{equation}\label{eq:thin-slice-union-bound}
\begin{split}
&\sum_{q\in [L]^{\mathcal{U}}} |S(q)\setminus \text{Int}_\delta(S(q))|\\
 \leq & |\mathcal{U}| \max_{\u\in \mathcal{U}} \left\{ y\in Y: (y, \u) \mod W\in \left[\frac{q}{L}-\delta, \frac{q}{L}+\delta\right]\cdot W, \text{~~~for some~}q\in [L]\right\},
\end{split}
\end{equation}

Let $L'$ be the least integer multiple of $1/\delta$ and $L$. By Lemma~\ref{lm:thin-slice} with parameter $L'$ (note that the preconditions as satisfied since $W/(\text{lcm}(L, 1/\delta)\cdot w)$ is an integer by assumption of the lemma) and
\begin{equation*}
\begin{split}
\mathcal{I}&:=\left\{q'\in [L']: \frac{q'}{L'}\in \left[\frac{q}{L}-\delta, \frac{q}{L}+\delta\right] \text{~for some~}q\in [L]\right\}\\
&=\left\{q'\in [L']: q'\in \left[q\cdot \frac{L'}{L}-\delta\cdot L', q\cdot \frac{L'}{L}+\delta L'\right] \text{~for some~}q\in [L]\right\},
\end{split}
\end{equation*}
we get, using the fact that $|\mathcal{I}|\leq (2\delta L'+1)\cdot L$, that 
\begin{equation*}
\begin{split}
&\left|\left\{y\in Y: (y, \u) \mod W\in \left[\frac{j-1}{L'}, \frac{j}{L'}\right]\cdot W, j\in \mathcal{I}\right\}\right|\\
&\leq (|\mathcal{I}|/L'+2/m)|Y|\\
&\leq ((2\delta+1/L')L+2/m)|Y|\\
&\leq (3\delta L+2/m)|Y|\text{~~~~(since $L'\geq 1/\delta$)}
\end{split}
\end{equation*}
Putting this together with~\eqref{eq:thin-slice-union-bound} yields the result.

\end{proof}

\begin{proofof}{Lemma~\ref{lm:intersection-size}}
Consider a discretization of the cube with parameters $L$ and $\delta\in (0, 1)$ (see Definition~\ref{def:discretization}). We use $\delta=2|\mathcal{U}|\cdot \delta'$.  Let  $A_\u=a_\u\cdot L, B_\u=b_\u\cdot L, \u\in \mathcal{U}$ be integers such that $a_\u=A_\u/L, b_\u=B_\u/L$. Recall that per Definition~\ref{def:discretization} we have
\begin{equation*}
S(q)=\left\{y\in Y: (y, \u) \mod W\in \left[\frac{q_\u-1}{L}, \frac{q_\u}{L}\right)\cdot W, \text{~for all~}\u\in \mathcal{U}\right\},
\end{equation*}
and
\begin{equation*}
\text{Int}_\delta(S(q))=\left\{y\in Y: (y, \u) \mod W\in \left[\frac{q_\u-1}{L}+\delta, \frac{q_\u}{L}-\delta\right]\cdot W, \text{~for all~}\u\in \mathcal{U}\right\}.
\end{equation*}
We let 
\begin{equation*}
\begin{split}
\mathcal J&:=\left\{q\in [L]^{\mathcal U}: ((q_\u/L)\cdot W+\Delta_\u) \mod W\in [A_\u/L, B_\u/L)\cdot W \right\}\\
&=\left\{q\in [L]^{\mathcal U}: ((q_\u/L+r_\u/L)\cdot W) \mod W\in [A_\u/L, B_\u/L)\cdot W \right\},\\
\end{split}
\end{equation*}
where $r_\u:=\Delta_\u\cdot L, \u\in \mathcal U$, are integers by assumption of the lemma. Note that 
$$
\S=\bigcup_{q\in \mathcal{J}} S(q),
$$
and hence, since $S(q)\cap S(q')=\emptyset$ for $q\neq q'$, we have
$$
|\S|=\sum_{q\in \mathcal{J}} |S(q)|.
$$
Also note that  $|\mathcal J|=\prod_{\u\in \mathcal U} (B_\u-A_\u)=L^{|\mathcal U|}\cdot \prod_{\u\in \mathcal U} (b_\u-a_\u)$. 
%Let $\mathcal{J}:=\prod_{\u\in \mathcal{U}} [A_\u: B_\u+1]\subseteq [L]^{\mathcal U}$ for convenience. 

The proof proceeds in two steps. We first lower bound the size of $\S$ and then upper bound it. The arguments are quite similar, and rely on technical lemmas derived in the rest of this section.   
 
\noindent{\bf Lower bound.}  First,  by Lemma~\ref{lm:inclusion} that for every $q, r\in [L]^{\mathcal{U}}$ one has 
\begin{equation}\label{eq:shifing-map-incl}
\psi_{q\to r}(\text{Int}_\delta(S(q))\setminus B)\subseteq S(r).
\end{equation}

 Note that the preconditions of the lemma are satisfied since $|\mathcal{U}|\cdot (W/w)<m^2$ ($m$ is sufficiently large as function of other parameters) and we set $\delta=2|\mathcal{U}|\cdot \delta'$.  Applying~\eqref{eq:shifing-map-incl} for every $q\in [L]^{\mathcal U}$ and $r\in \mathcal{J}$   and noting that the mapping $\psi_{q\to r}$ is injective gives

\begin{equation*}
\begin{split}
|{\mathcal J}| \cdot \sum_{q\in [L]^{\mathcal U}} |\text{Int}_\delta(S(q))\setminus B| \leq L^{|{\mathcal U}|} \sum_{q\in \mathcal{J}} |S(q)|.
\end{split}
\end{equation*}

We thus get  
\begin{equation}\label{eq:lb-12fng}
\begin{split}
\sum_{q\in \mathcal{J}} |S(q)|&\geq \left(\prod_{\u\in {\mathcal U}} (b_\u-a_\u)\right)\cdot \sum_{q\in [L]^{\mathcal U}} |\text{Int}_\delta(S(q))\setminus B|\\
&\geq \left(\prod_{\u\in {\mathcal U}} (b_\u-a_\u)\right)\cdot \sum_{q\in [L]^{\mathcal U}} (|S(q)\setminus B|-|S(q)\setminus \text{Int}_\delta(S(q))|)\\
&\geq \left(\prod_{\u\in {\mathcal U}} (b_\u-a_\u)\right)\cdot (|Y|-|B|-\sum_{q\in [L]^{\mathcal U}} |S(q)\setminus \text{Int}_\delta(S(q))|)\\
&\geq \left(\prod_{\u\in {\mathcal U}} (b_\u-a_\u)\right)\cdot |Y|-|\mathcal{U}| (3\delta  L+4/m)|Y|.\text{~~~~~(by Lemma~\ref{lm:delta-slice} and Lemma~\ref{lm:few-bad-vertices})}\\
\end{split}
\end{equation}
We used Lemma~\ref{lm:delta-slice} and the fact that $\prod_{\u\in {\mathcal U}} (b_\u-a_\u)\leq 1$ since $a_\u, b_\u\in [0, 1]$ by assumption of the lemma to go from line~3 to line~4, and Lemma~\ref{lm:few-bad-vertices} to go from line~4 to line~5.

\noindent{\bf Upper bound.} At the same time we also get, using again that by Lemma~\ref{lm:inclusion} that for every $q, r\in [L]^{\mathcal{U}}$ one has 
 $$
\psi_{q\to r}(\text{Int}_\delta(S(q))\setminus B)\subseteq S(r),
 $$
 that 
\begin{equation*}
\begin{split}
|{\mathcal J}| \cdot \sum_{q\in [L]^{\mathcal U}} |S(q)| \geq L^{|{\mathcal U}|} \sum_{q\in \mathcal{J}} |\text{Int}_\delta(S(q))\setminus B|.\\
\end{split}
\end{equation*}
The above bound follows by noting that for every $q\in \mathcal{J}$ and $r\in [L]^{\mathcal U}$  one has 
 $\psi_{q\to r}(\text{Int}_\delta(S(q))\setminus B)\subseteq S(r)$, and the mapping $\psi_{q\to r}$ is injective. We thus get 
\begin{equation}\label{eq:ub-13ronergergrweg}
\begin{split}
\sum_{q\in \mathcal{J}} |\text{Int}_\delta(S(q))\setminus B|&\leq \left(\prod_{\u\in {\mathcal U}} (b_\u-a_\u)\right)\cdot \sum_{q\in [L]^{\mathcal U}} |S(q)|= \left(\prod_{\u\in {\mathcal U}} (b_\u-a_\u)\right)\cdot |Y|.\\
\end{split}
\end{equation}

We also have by Lemma~\ref{lm:delta-slice}  (applied with $\delta'$) and Lemma~\ref{lm:few-bad-vertices}
\begin{equation*}
\begin{split}
\sum_{q\in \mathcal{J}} |\text{Int}_\delta(S(q))\setminus B|&\geq \sum_{q\in \mathcal{J}} |S(q)|- \sum_{q\in [L]^{\mathcal{U}}}|S(q)\setminus \text{Int}_\delta(S(q))|-|B|\\
&\geq \sum_{q\in \mathcal{J}} |S(q)|- |\mathcal{U}|(3\delta L+4/m)|Y|.
\end{split}
\end{equation*}

Substituting this into~\eqref{eq:ub-13ronergergrweg},  we get 
\begin{equation}\label{eq:ub-12fng}
\begin{split}
\sum_{q\in \mathcal{J}} |S(q)|&\leq \left(\prod_{\u\in {\mathcal U}} (b_\u-a_\u)\right)\cdot |Y|+ |\mathcal{U}|(3\delta L +4/m) |Y|\\
\end{split}
\end{equation}

Finally, putting~\eqref{eq:ub-12fng} together with~\eqref{eq:lb-12fng}, we obtain the bound

\begin{equation*}
\begin{split}
\left||\S|-|Y|\cdot \prod_{\u\in \mathcal{U}} (b_\u-a_\u)\right|&\leq |\mathcal{U}|(3\delta L +4/m)\cdot |Y|\\
&\leq |\mathcal{U}|^2(6L\delta'  +4/m)\cdot |Y|
\end{split}
\end{equation*}
as required.
 \end{proofof}

\end{appendix}

\end{document}